\def\llncs{0}
\def\fullpage{1}
\def\anonymous{0}
\def\draft{0}
\def\submission{0}

\ifnum\submission=1
\def\llncs{1}
\def\draft{0}
\def\anonymous{1}
\fi

\ifnum\llncs=1
    \documentclass{llncs}
\else
    \documentclass[letterpaper,hmargin=1.05in,vmargin=1.05in]{article}
    \ifnum\fullpage=1
    \usepackage{fullpage}
    \fi
\fi

\usepackage{verbatim}
\usepackage{authblk}
\usepackage{graphicx} 
\usepackage{physics} 
\usepackage{xcolor} 
\usepackage{amsmath}
\usepackage{amssymb}
\usepackage{mathtools}

\usepackage{amsthm}
\usepackage{enumitem} 
\usepackage{dsfont}
\usepackage[colorlinks=true,linkcolor=magenta,citecolor=blue,pagebackref=true]{hyperref}
\usepackage{dashbox}
\usepackage{fancybox,framed}
\usepackage[skins]{tcolorbox}
\usepackage{subcaption}
\usepackage{relsize}
\usepackage{comment} 
\usepackage[capitalise,nameinlink]{cleveref}
\makeatletter
\renewcommand\appendix{\par
  \setcounter{section}{0}\setcounter{subsection}{0}%
  \gdef\thesection{\Alph{section}}%
  \gdef\theHsection{\Alph{section}}
  \crefalias{section}{appendix}
  \crefalias{subsection}{subappendix}
}
\makeatother
\usepackage{tikz}
\usetikzlibrary{shadings}

\usepackage{stmaryrd}

\ifnum\draft=1
\newcommand{\minki}[1]{\textcolor{blue}{$\langle\langle$Minki: #1$\rangle\rangle$}}
\newcommand{\wonseok}[1]{\textcolor{red}{$\langle\langle$Wonseok: #1$\rangle\rangle$}}
\newcommand{\junyoung}[1]{\textcolor{violet}{$\langle\langle$Junyoung: #1$\rangle\rangle$}}
\else
\newcommand{\minki}[1]{}
\newcommand{\wonseok}[1]{}
\newcommand{\junyoung}[1]{}
\fi

\ifnum\llncs=0
	\newtheorem{theorem}{Theorem}[section]
	\newtheorem{lemma}[theorem]{Lemma}
	\newtheorem{corollary}[theorem]{Corollary}
	
	\newtheorem{definition}[theorem]{Definition}

	\theoremstyle{remark}
	\newtheorem{claim}[theorem]{Claim}

\else
\fi

\DeclarePairedDelimiterX{\inner}[1]{\langle}{\rangle}{\innerp{#1}}

\ExplSyntaxOn
\msg_new:nnn { innerp } { too-many-operands }
  { Too~many~operands~in~argument~'#1'~of~\tl_to_str:n { \inner: } #2. }

\cs_generate_variant:Nn \msg_error:nnnn { nnnx }

\NewDocumentCommand \innerp { m }
  {
    \int_case:nnF { \clist_count:n {#1} }
      {
        { 1 } { #1, #1 }
        { 2 } {#1}
      }
      {
        \msg_error:nnnx { innerp } { too-many-operands }
          {#1} { \clist_count:n {#1} }
      }
  }
\ExplSyntaxOff

\newcommand{\C}{\mathbb{C}}

\newcommand{\D}{\mathcal{D}}

\newcommand{\cS}{\mathcal{S}}

\newcommand{\Alg}{\mathcal{A}}
\newcommand{\Blg}{\mathcal{B}}

\renewcommand{\epsilon}{\varepsilon}

\newcommand{\Exp}{\mathbb{E}}

\newcommand{\supp}{\mathsf{supp}}

\newcommand{\bit}{\{0,1\}}
\newcommand{\pmone}{\{-1,1\}}
\newcommand{\Ftwo}{\mathbb{F}_2}

\newcommand{\F}{{\mathbf{F}}}
\renewcommand{\P}{{\mathbf{P}}}
\newcommand{\hF}{{\widehat{\mathbf{F}}}}
\newcommand{\hh}{{\widehat{H}}}
\renewcommand{\H}{\mathcal H}

\newcommand{\constzero}{{\mathbf{0}}}
\newcommand{\constone}{{\mathbf{1}}}

\newcommand\truth[1]{\llbracket #1 \rrbracket}

\newcommand{\xop}{{\mathsf{XoP}}}
\newcommand{\PC}{{\mathsf{PC}}}

\newcommand{\DPC}{{\D_{\PC}}}

\newcommand{\DthPC}{{\D_{3\PC}}}
\newcommand{\DfoX}{{\D_{\mathsf{4X}}}}
\newcommand{\DXoPr}{{\D_{\xop[r]}}}

\newcommand{\muP}{\mu_{\P}}
\newcommand{\muXoP}{\mu_{\xop}}
\newcommand{\muPC}{\mu_{\PC}}
\newcommand{\muPCt}{\mu_{\PC, 2}}
\newcommand{\muthPC}{\mu_{3\PC}}
\newcommand{\mufoX}{{\mu_{\mathsf{4X}}}}
\newcommand{\muXoPr}{\mu_{\xop[r]}}
\newcommand{\hmu}{\widehat{\mu}}
\newcommand{\hmuP}{\widehat{\mu}_{\P}}
\newcommand{\hmuXoPr}{\widehat{\mu}_{\xop[r]}}

\newcommand\dist[1]{\D_{#1}}
\newcommand\dens[1]{\mu_{#1}}

\newcommand{\Adv}{{\sf Adv}}

\newcommand{\newdef}[1]{\emph{\underline{{#1}}}}

\ifnum\llncs=1
    
    \spnewtheorem{claim}{Claim}{\bfseries}{\itshape}
    \numberwithin{claim}{section}
    \Crefname{claim}{Claim}{Claims}
\fi

\title{Quantum Security of XOR of Permutations via Fourier Analysis}
\date{}

\ifnum\llncs=1
    \ifnum\anonymous=1
        \author{}
        \institute{}
    \else
        \author{
            Wonseok Choi\inst{1}
            \and
            Minki Hhan\inst{2}
            \and
            Junyoung Jang\inst{2}
        }
        \institute{
            DGIST\\
            \email{wonseok@dgist.ac.kr}
    \vspace{0.1em}
            \and
            KAIST\\
            \email{minkihhan@kaist.ac.kr, junsarang74@kaist.ac.kr}
        }
    \fi
\else

    \author[1]{Wonseok Choi}
    \author[2]{Minki Hhan}
    \author[2]{Junyoung Jang}

    \affil[1]{{\small DGIST, Daegu, Korea}
    \authorcr{\small wonseok@dgist.ac.kr}}

    \affil[2]{{\small KAIST, Daejeon, Korea}
    \authorcr{\small minkihhan@kaist.ac.kr, \quad junsarang74@kaist.ac.kr}}
\fi

\begin{document}

\maketitle
\begin{abstract}
The XOR of two or more independent random permutations (XoP) is the prototypical pseudorandom function built from permutations achieving the security beyond the birthday bound.
The security of the XoP construction is now well established against classical adversaries, however, its security against quantum adversaries that query the construction in superposition has remained widely open.

We prove that the XOR of $r\ge 2$ independent random permutations over $\{0,1\}^n$ is indistinguishable from a random function by any $q$-query quantum algorithm with advantage
\[
O\left(\min\left\{
\frac{q^3}{2^{rn}},\frac{q^{1.5}}{2^{(r-0.5)n}},\frac{1}{2^{(r-1.5)n}}
\right\}
\right)
\]
for all $q\le 2^n/57774$, where the hidden factors depend only on $r$. In particular, the XoP construction remains secure throughout the entire query range, far beyond the $2^{n/3}$ quantum birthday bound due to the quantum collision finding attack. To our knowledge, this is the first construction from permutations that achieves the quantum version of the beyond birthday bound security.

We present several heuristic attacks suggesting the tightness of our bounds. 
For $q\lesssim 2^{n/2}$, the quantum collision finding-based attacks heuristically give the advantage $\Omega(q^3/2^{rn})$ and $\Omega(q^{1.5}/2^{(r-0.5)n})$, and for $q\approx 2^n$, the heuristic collision counting-based attack appears to have the advantage about $2^{-(r-1.5)n}$.
 
We use a Fourier-analytic variant of the polynomial method on the space of functions: the distinguishing advantage of any $q$-query quantum algorithm is controlled by the Fourier components of degree at most $2q$, which is in turn controlled by $2q$ input-output data of the construction. The norms of most components are bounded well, proving the bound $2^{-(r-3/2)n}$.

The norm of low-degree components turns out to be too large for the bound $q^3/2^{rn}$. We instead reinterpret these low-degree components as (sums of) distinguishing advantages of the other problems. For example, the degree-2 and degree-4 terms are interpreted as the advantages against random functions with and without \emph{planted collisions}, which in turn are bounded using Zhandry's small-range distributions.

Finally, the $q^{1.5}/2^{(r-0.5)n}$ bound can be proven using another bound for the planted collisions. This new bound is proven using the compressed oracle by interpreting the advantage as the other bound about random functions. It also gives a new bound for the small-range indistinguishability for (ironically) large ranges, which is of independent interest.
\end{abstract}
\ifnum\fullpage=1
\vfill

\paragraph{AI use disclosure.} The Fourier-analytic approach to the quantum security of XoP was formulated by the authors, with assistance from ChatGPT 5.4 and 5.5 Pro in formalizing many technical details. 
From the authors' perspective, the first nontrivial mathematical contribution from AI was a proof of the $O(q^3/N^r)$ bound of the degree-2 components (\cref{lem: xopdeg2_smallq}), obtained through a series of conversation with ChatGPT 5.5 Pro.
Initial proofs of several other lemmas were also generated by ChatGPT 5.5, 5.6 Pro and Astra (in the Ultra mode), as well as Fable 5.0 and 5.1.
The authors subsequently simplified these proofs and, in many cases, developed substantially different proofs that they considered more natural and accessible.
The current paper is mostly written by human authors from scratch, while AI assisted in writing some paragraphs.

    \clearpage
    \newpage
    \setcounter{tocdepth}{2}
    \tableofcontents
    \newpage
\fi

\ifnum\submission=1
\else
\fi

\section{Introduction}
\label{sec:intro}
Random permutations are arguably the most important idealized object in cryptography. 
Even in symmetric-key cryptography, {\it the standard model} assumes a keyed block cipher to be a private (pseudo)random permutation due to the wide adoption of the standardized block ciphers such as DES and AES~\cite{DES77,AES}. In this regard, earlier practical symmetric-key schemes are instantiated with block ciphers (i.e., random permutations) while the schemes could enjoy better security if instantiated with random functions~\cite{EC:BelKroRog98,C:HWKS98}. The representative examples are Wegman-Carter(-Shoup) MACs and counter mode~\cite{WegCar81,C:Shoup96,EC:Bernstein05}. 

The subsequent research on symmetric-key construction designs has sought to achieve better security,
especially the security beyond the birthday bound---
namely the \emph{beyond-birthday-bound} security. In its heart, the Luby-Rackoff backward problem~\cite{EC:BelKroRog98} asks how to construct (pseudo)random functions (PRFs) using random permutations.
It has become a prominent cryptographic field intersecting theory and practice~\cite{EC:BelKroRog98,C:HWKS98,EC:Lucks00,ICITS:Patarin08,INDOCRYPT:ManPatNac10,C:DaiHoaTes17,JC:GilGueMor18,EC:BhaNan18,C:CheLamMen19,AC:ChoLeeLee19,C:GunMen20,AC:BhaNan21,AC:CKLL22,EC:Dinur24,EC:Dinur25}. See the related works for a more detailed discussion.

This paper studies the PRF constructions based on random permutations in the quantum setting, where the question becomes substantially more difficult. 
A quantum PRF (QPRF) \cite{zhandry2021construct} is a function that is indistinguishable from random functions even allowing coherent access to the function.
The coherent queries prevent the classical security proofs as it potentially \emph{sees} all data of the constructions.
Worse, unlike the random functions, these coherent queries potentially reveal the \emph{correlation} between the input-output pairs of the permutations, which hinders or complicates the beautiful compressed oracle method \cite{C:Zhandry19} for permutations; the known results are very lossy or restrictive \cite{STOC:Carolan26,HHM+26,Ros21,EC:ABCGY26,FLMNW26,C:CHLYY25}.
This difficulty appears in the current disappointing state-of-the-art, establishing tight quantum security of permutation-based cryptography remains very challenging.
In particular, no permutation-based QPRFs is known to be secure beyond $q=2^{n/3}$ queries for $n$-bit permutations, which is already (and optimally) achieved by the plain random permutations \cite{zhandry2015note,BHT97}.

This leads to a practical concern for the future quantum cryptographic primitives. We are currently uncertain about the quantum security of any PRFs from permutations and applications \cite{EC:BonZha13,C:BonZha13,C:SonYun17,EC:AMRS20}. Importantly, any instantiation of these applications using block ciphers such as AES256 is only known to be secure up to $2^{128/3}\approx2^{43}$ quantum queries, which is unsatisfactory to cryptographers.
Resolving this question would give more efficient schemes for cryptography with quantum states and unitaries \cite{C:JiLiuSon18,C:AnaQiaYue22,FOCS:MPSY24,STOC:MaHua25} as well.

\subsection{Our results}
This paper shows that the exclusive-or of two or more permutations, i.e., the XoP constructions, also have a strong security in the quantum setting. 
This construction processes the input $x$ by
\[
    \xop[r](x):=P_1(x)+\cdots + P_r(x),
\]
which we simply refer to as the $\xop[r]$ construction,
where \(+\) denotes bitwise XOR and $P_1,...,P_r$ are independent random permutations over $n$-bit strings.

In the classical setting, the tight security of the $\xop$ construction for $r=2$ was $O\left(\min\left(\frac{q^2}{2^{2n}},\frac{q}{2^{1.5n}}\right)\right)$. 
The first bound dominates for $q\lesssim 2^{0.5n}$ and is proven in \cite{DNS22,C:CheChoLee23}. The second bound dominates for $q\gtrsim 2^{0.5n}$ and recently proven in \cite{EC:Dinur24,Eberhard17} using Fourier analysis. The latter bound can be generalized to $q/2^{r-0.5}$ for $r \ge 2$. These bounds are tight, matching the known attack by Patarin \cite{ICITS:Patarin08}.

Our main theorem proves the quantum security of the $\xop$ construction.
\begin{theorem}[Informal]
Fix a constant $r \ge 2$.
For every \(0 \le q \ll 2^n\),
any quantum distinguisher making $q$ quantum queries to either the $\xop[r]$ construction or the random functions
has advantage at most
\[
O\left(\min\left(
        \frac{q^3}{2^{rn}},\frac{q^{1.5}}{2^{(r-0.5)n}},
        \frac{1}{2^{(r-1.5)n}}
    \right)\right).
\]
\end{theorem}
In particular, the $\xop[r]$ construction for $r\ge 2$ is secure up to $q\approx2^n$ queries.

The above theorem concerns an ideal setting where the adversary does not have access to the primitive queries to $P_1,...,P_r$.
In practice, the XoP constructions can be instantiated using block ciphers with independent random keys in the quantum ideal cipher model.
When the underlying block cipher has a sufficiently long key, say $\kappa$-bit keys, 
the corresponding keyed permutations can be viewed as private random permutations in the quantum ideal cipher model, with the distinguishing advantage loss $rq^2/2^\kappa$ by applying the quantum search bound $r$ times \cite{BBBV97,BBHT98} (assuming there is no collisions in the $r$ keys). 
Therefore, our result can be lifted to the quantum ideal cipher model with such a key length; for example, AES256 provides a heuristic instantiation with $\kappa=256$.

To our knowledge, this gives the first QPRF constructions in the quantum ideal cipher model whose security goes beyond the $2^{n/3}$ quantum queries. In fact, our theorem shows that the XoP construction is still indistinguishable from random functions near $2^n$ queries; we are not even aware of any construction secure against $2^{n/2}$ queries.

\subsection{Technical overview}\label{ssec:overview}
We focus on the XOR of two permutations in this overview.
Our proof can be viewed as a quantum analogue of the Fourier analytic security proofs developed in a recent line of works \cite{Eberhard17,EMM19,EC:Dinur24,EC:Dinur25}, which have established tight classical security bounds for the XoP constructions. 
The main object in these analyses is the classical response transcript.  
More precisely, after fixing the $q$ distinct query points, the corresponding response vector lies in $(\bit^n)^q$. For the XoP construction, this vector is distributed as the sum of two vectors sampled without replacement, while this vector is uniform random for the truly random functions.
The Fourier analysis compares these two distributions of vectors in the Fourier domain.

\paragraph{Polynomial method.}
In our quantum setting, however, there is no proper notion of the transcript because of the coherent quantum queries.
We instead employ a Fourier-analytic variant of the polynomial method \cite{JACM:BBCMD01}. 
Our starting point is the following well-known observation. 
Let $\ket{\widehat a} := \frac{1}{\sqrt{2^{n}}}\sum_{z\in \bit^n} \chi_a(z)\ket{z}  = H^{\otimes n} \ket{a}$ be a Fourier basis state where $\chi_a (z)=(-1)^{\inner{a,z}}$. 
For an oracle $O_f$ that evaluates a function $f$, i.e., that maps $\ket{x,y} \mapsto \ket{x,y+f(x)}$ coherently, we have
\[
 O_f \ket{x, \widehat a} =(-1)^{\inner{a,f(x)}}\ket{x, \widehat a}= \chi_a(f(x)) \ket{x, \widehat a}.
\]
Motivated by this, we consider a vector $\gamma=(\gamma_x)_{x\in \bit^n}\in(\bit^n)^{2^n}$ and define its Fourier degree by $|\gamma|
:=
|\{x\in\bit^n:\gamma(x)\neq0^n\}|$.
Each vector gives its Fourier character defined by
\[
\chi_\gamma(f)
:=
(-1)^{\sum_{x\in\bit^n}\inner{\gamma(x),f(x)}}, \qquad \chi_{\gamma+\eta}(f) = \chi_{\gamma}(f) \cdot \chi_\eta(f).
\]
The above observation shows that the amplitudes of each basis vector can be written as a linear sum of the Fourier characters, and each query increases the degree of amplitudes by at most 1.

The acceptance probability $\Pr[A^{O_f}=1]$ therefore must be represented by a low-degree polynomial of the Fourier characters. By regarding it as a function from $f$ to the real number in $[0,1]$,\footnote{That is, this is a \emph{functional} that maps a function to a scalar. In the main body, we refer this Fourier analysis as the Fourier analysis of functionals to distinguish the different use of the term function.}
the Fourier expansion gives
\[
 P_{\mathcal A}(f):=\Pr[A^{O_f}=1] = \sum_{|\gamma| \le 2q}
\widehat P_\Alg(\gamma) \cdot \chi_\gamma(f), \quad \widehat P_\Alg(\gamma) = \Exp_{g \gets \F}\left[\overline{\chi_\gamma(g)}P_\Alg(g)\right].
\]
Here, 
$g\gets \F$ denotes a sample from the uniformly random functions.
The degree bound $\le 2q$ is because the amplitudes of the final state of a $q$-query algorithm are Fourier polynomials
of degree at most $q$. 
This observation, or closely related variants in representation-theoretic language, has been used in many analyses of quantum algorithms
\cite{zhandry2021construct,ITCS:SY23,C:Zhandry19,Ros21,EC:ABCGY26}.
It allows us to bound the probability of the oracle algorithm through their bounded-degree Fourier statistics.

\paragraph{Distinguishing advantage as an inner product.}
We introduce some definitions to proceed the overview.
Let $N:=2^n$, and 
let $\F = \{f:\bit^n\to\bit^n\}$ be the space of functions.
We write
\[
\inner{g,h}:=\Exp_{f \gets \F}[\overline{g(f)}h(f)]
\qquad\text{and}\qquad
\norm{g}_2^2:=\inner{g,g}.
\]
With this notation, we can write
\[
\Pr_{f \gets \F}[\Alg^{O_f}=1] = \sum_{f \in \F} \frac{P_{\mathcal A}(f)}{|\F|}
    = \Exp_{f\gets \F}[ \constone(f)P_{\mathcal A}(f)] = \inner{\constone,P_{\mathcal A}}
\]
where $\constone$ is a constant function with output 1.
Similarly, define the density of the XoP distribution relative to this uniform measure by
\[
\mu_{\xop}(f)
:=
\frac{\Pr_{P, Q \gets \P}[P+Q=f]}{\Pr_{P \gets \F}[P=f]}
=
N^N\Pr_{P, Q \gets \P}[P+Q=f].
\]
An analogous calculation gives 
\[
\Pr_{P,Q \gets \P}[\Alg^{O_{P+Q}}=1] = \inner{\mu_{\xop},P_{\mathcal A}}
\]
and the distinguishing advantage can be written as
\[
\Pr_{P,Q \gets \P}[\Alg^{O_{P+Q}}=1]
-
\Pr_{f \gets \F}[\Alg^{O_f}=1]
=
\inner{P_\Alg,\mu_{\xop}-\constone}.
\]
The above discussion about the Fourier analysis applies on these probabilities, showing that $P_\Alg$ has Fourier degree at most $2q$ and $\norm{P_\Alg}_2\le 1$.

\paragraph{Decomposition by Fourier degree.}
We consider the following two terms:
\begin{enumerate}[nosep]
    \item $(\cdot)^{=d}$ that denotes the sum of all Fourier components of degree exactly $d$;
    \item $(\cdot)^{\le d}$ that denotes the sum of all Fourier components of degree $\le d$, i.e., the sum of all $(\cdot)^{=k}$ for $k \le d$.
\end{enumerate}
The orthogonality of the Fourier components, following the so-called Efron–Stein orthogonal decomposition \cite[Chapter 8]{o2014analysis}, gives the advantage bound by
\[
|\inner{P_\Alg,\mu_{\xop}-\constone}|=
    |\inner{P_\Alg,(\mu_{\xop}-1)^{\le 2q}}|
    =
\left|
\sum_{k=1}^{2q}
\inner{P_\Alg,\mu_{\xop}^{=k}}
\right|\le \sum_{k=1}^{2q} |\inner{P_\Alg,\mu_{\xop}^{=k}}|
\]
where we use 
the fact that $P_\Alg$ is of degree at most $2q$ in the first equality, and 
use $(\mu_{\xop}-1)^{=0}=0$ and $(1)^{=k}=0$ for all $k\ge 1$.
The Cauchy–Schwarz inequality gives the first way to bound this:
\[ 
\le
\sum_{k=1}^{2q}
\norm{P_\Alg}_2
\norm{\mu_{\xop}^{=k}}_2
\le
\sum_{k=1}^{2q}
\norm{\mu_{\xop}^{=k}}_2
\]
where we use $\norm{P_\Alg}_2 \le 1$. 
For most degrees, we obtain the required bound directly by estimating
$\norm{\mu_{\xop}^{=k}}_2$ through a delicate combinatorial argument.
In particular, the second advantage bound $2^{-0.5n}$ is obtained by this strategy alone.

\paragraph{Turning Fourier term into planted collision distinguisher.}
To prove the security $q^3/2^{2n}$ for the XoP construction, the above strategy turns out to be insufficient. In particular, the degree-2,3,4,6 components have too large $\ell_2$-norm.
We use a different proof strategy; namely, we turn the Fourier-analytic components $\inner{P_\Alg,\mu_{\xop}^{=k}}$ to the distinguishing advantage to \emph{some other problems}.

We focus on the degree-2 component, or more precisely $\inner{P_\Alg,\mu_{\xop}^{=2}}$, in this overview.
The degree-2 term exhibits the statistics of two outputs of the functions, such as collisions.
This component turns out to be the dominant term in the bound, which was the dominant term in the classical security as well \cite{ACNS:Patarin13,EC:Dinur24}.

This intuition holds in the quantum setting as well.
Write $\truth{X}$ to denote the truth value of the statement $X$.
We will establish the following identity at the end of the overview:
\begin{align}
    \label{eqn:intro_xop2}
    \mu_{\xop}^{=2}(f) = \left(\frac{1}{N-1}\right)^2
    \sum_{x < x'}
    \left(N \truth{f(x)=f(x')}-1\right),
\end{align}
which means that the degree-two component can be understood as a collision statistic up to a multiplicative factor; the truth value $\truth{f(x)=f(x')}$ is 1 if and only if $f$ collides at $x\neq x'$.

Fix $x\neq x'$ and fix a function $f$. Observe that 
\begin{itemize}
    \item $\frac{1}{N^N}$ is the probability that $f$ is sampled from $\F$;
    \item $\frac{1}{N^{N}} \cdot N \truth{f(x)=f(x')}$ is the probability that $f$ is sampled from $\F$ conditioned on $f(x)=f(x')$.
\end{itemize}
From this observation, we can interpret $\mu_{\xop}^{=2}(f) $ as, up to a constant factor, the difference between the density functions of 1) the distribution of the true random function $\F$ and 2) the distribution of random functions conditioned on a \emph{planted collision} $f(x)=f(x')$ for random $x\neq x'$.
That is, we can reinterpret the inner product
\[
|  \inner{P_\Alg,\mu_{\xop}^{=2}}|
=
O\left(
\left|  \Pr_{f \gets \DPC}[A^{O_f}=1]-\Pr_{f \gets \F}[A^{O_f}=1]\right|
\right)
\]
as a distinguishing advantage of \emph{another problem}:
1) a random function with a planted collision and 2) an unconditional random function. Here, $\DPC$ denote the planted collision distribution described above.

\paragraph{Planted collisions vs. random functions.}
We directly prove the distinguishing advantage $O(q^3/N^2)$ of this problem using Zhandry's small range distribution indistinguishability \cite{zhandry2021construct}, which shows that the distinguishing advantage of the following two distributions is bounded above by $O(q^3/R)$: 
The first one is the uniformly random functions from $\bit^n$ to $\bit^n$, and the second one is the distribution $D_R$ of random ``small-range'' functions, which are obtained by composing two independent random functions $g \circ h$ where $h:\bit^n \to [R]$ and $g: {\mathsf{Im}}(h) \to \bit^n$.

Ironically, we use the \emph{large range} function with $R\gtrsim N^2$.
In this case, $h$ has a unique collision with a constant probability, in which case $g \circ h$ is distributed according to the planted collision distribution.
Therefore, the distinguishing $D_R$ and $\F$ is closely related to the above problem.
It turns out that any $R\gtrsim N^2$ exhibits such a relation; when there is no collision in $h$, the distribution of $g \circ h$ becomes the uniformly random function so that the related terms are canceled out, and $R \to \infty$ makes the unique collision dominate among the other terms, 
deriving the upper bound $O(q^3/N^2).$

The other degrees are more involved, but similar strategies show the desired bounds.
Very roughly, we connect the degree-3,4,6 terms with the other problems using similarly to \cref{eqn:intro_xop2}; for example, the degree-3 term is related to the distinguishing advantage between the random functions and the random planted-3-collision functions with a constraint $f(x_1)=f(x_2)=f(x_3)$. The degree-4 term is related to the random functions with two planted collisions together with the random functions with the planted 4-XOR $f(x_1)+f(x_2)+f(x_3)+f(x_4)=0^n$.
We reduce the hardness of these problems to the planted collision advantage similarly to the above (yet in a slightly more involved way).
The bound for the degree-2 is tight in the relevant parameter range, as suggested by the corresponding quantum collision-statistics attack below.

\paragraph{Another bound for the planted collisions.}
To derive the $O(q^{1.5}/N^{1.5})$ security bound, we observe that the term in \cref{eqn:intro_xop2} can be interpreted some collision-related term in the language of the compressed oracle \cite{C:Zhandry19,EC:CFHL21}. By writing the overall final state $\ket{\Psi}$ of the algorithm and (purified) database by $\frac{1}{\sqrt{N^N}}\sum_g\ket{\psi^g}\ket{g}$, 
the advantage can be expressed by
\begin{align*}
    \inner{\muPC-\constone,P_\Alg} &= \frac{2}{N-1}\Exp_{f\gets\F}\left[
        P_\Alg(f)
        \sum_{x<y}\left(\truth{f(x)=f(y)}-\frac1N\right)
    \right]
    \\&
    =\frac{2}{N^N(N-1)}\sum_{f,g\in\F}
        \bra{\psi_q^g}\Pi\ket{\psi_q^f}
        \bra{ g}\left(\sum_{x<y}\left(\mathsf{Col}_{xy}-\frac1N I\right)\right)\ket{ f}
        \\&
        =\frac{2}{N-1} \bra{\Psi}\Pi
        \left(\sum_{x<y}\left(\mathsf{Col}_{xy}-\frac1N I\right)\right)
        \ket{\Psi}
\end{align*}
where $\Pi$ is the projection to the acceptance space of the algorithm, and $\mathsf{Col}_{xy}$ is the indicator function for the database on inputs $x,y$.
This expression can be bounded by $O(q^{1.5}/N^{1.5})$ using the compressed oracle machinery alone, proving the desired result. Intriguingly, the above-discussed small-range to planted collision can be reversed, and $q^{1.5}/N^{1.5}$ gives an alternative bound for the small-range distribution indistinguishability for large range. See \cref{cor: improved small_range}.

\paragraph{Proof of \cref{eqn:intro_xop2}.}
We now discuss how to prove \cref{eqn:intro_xop2}. 
Let $\P$ be the distribution of uniformly random permutations, and let $\muP(f) = N^N\Pr_{p \gets \P}[p = f]$ be the density function of $\P$. 
The degree-2 term expands, together with the convolution formula,
\begin{align}
    \label{eqn: intro_deg2}
\mu_{\xop}^{=2}(f) = \sum_{|\gamma|=2} \hmu_{\xop}^{=2}(\gamma) \chi_\gamma(f) = 
\sum_{|\gamma|=2} \hmuP^2(\gamma) \chi_\gamma(f).
\end{align}

We extensively use the following basic observation $\sum_{x\in \bit^n} (-1)^{\inner{a,x}} = N\truth{a=0^n}$ and its variations. For example, for $y,y'\neq 0^n$, we have
\[
\sum_{u\neq v} (-1)^{\inner{u,y} + \inner{v,y'}} = \sum_{u, v} (-1)^{\inner{u,y} + \inner{v,y'}} - \sum_{u} (-1)^{\inner{u,y+y'}} = 
-N\truth{y=y'} 
\]
where all $u,v \in \bit^n$. In the second equality, the first term vanishes because of $\truth{y=0^n}=0$, and the last term is simplified using $\truth{y+y'=0^n} = \truth{y=y'}. $

Let $\gamma=(\gamma_x)_{x\in \bit^n}\in(\bit^n)^{2^n}$ be such that $|\gamma|
=|\supp(\gamma)|=2$ where $\supp(\gamma) = \{x\in\bit^n:\gamma(x)\neq0^n\}$. Let $x,x'$ be the elements of $\supp(\gamma)$.
Observe that for a uniform random permutation $p\gets \P$, $(p(x),p(x'))$ is uniform over ordered distinct pairs, hence
\[
\Exp_{p\gets \P}[\chi_\gamma(p)] = \frac{1}{N(N-1)}
\sum_{p(x)\neq p(x')}(-1)^{p(x) \cdot \gamma(x) + p(x') \cdot \gamma(x')} = - \frac{\truth{\gamma(x)=\gamma(x')}}{N-1}.
\]
It is not hard to see that the above term coincides with $\hmuP(\gamma)$ using the definition of the Fourier coefficient and the fact that $\chi_\gamma$ is real-valued. We omit this proof in this overview.

Now we go back to \cref{eqn: intro_deg2}. Plugging the above identity gives
\[
    \sum_{|\gamma|=2} \hmuP^2(\gamma) \chi_\gamma(f) = \sum_{\supp(\gamma)=\{x<x'\}} \frac{\truth{\gamma(x)=\gamma(x')}^2 \chi_\gamma(f)}{(N-1)^2}
\]
where we use the lexicographic order in $x<x'$.
Writing $\gamma(x)=\gamma(x')=y \neq 0^n$ and using $\chi_\gamma(f)  = (-1)^{\inner{y,f(x)+f(x')}}$, we can write this as
\begin{align*}
\sum_{x<x', y\neq 0^n} \frac{(-1)^{\inner{y,f(x)+f(x')}}}{(N-1)^2} 
&=\sum_{x<x', y} \frac{(-1)^{\inner{y,f(x)+f(x')}}}{(N-1)^2}  - \sum_{x<x'} \frac{1}{(N-1)^2}\\&
=\sum_{x<x'}\frac{ N\truth{f(x)=f(x')} - 1}{(N-1)^2}.
\end{align*}
This proves \cref{eqn:intro_xop2}.
Similar calculations derive analogous algorithmic interpretations of different degrees.

\subsection{Heuristic Attacks}\label{subsec:attacks}
We briefly discuss the (potential) distinguishing attack for the $\xop$ constructions in \cref{app: attacks}.
The description focuses on the high-level ideas of the attack, with some statistical evidence or heuristic reasoning of their advantages. 
Formalizing them needs further analysis.
Most attacks for $q\ll N$ are related to the pairwise statistics, i.e., collisions of the construction.
On the other hand, the attack for $q\ge N/2$ uses the quantum parity algorithm \cite{FGGS98}, achieving the advantage $1/2$.

\cref{tab:xop-attacks} summarizes the different attacks and security bounds on $\xop[r]$.
Our lower bounds match to the heuristic attacks for $q\lesssim N^{1/2}$ and $q\approx N$.
It is unclear which is tight for the range $N^{1/2}\gtrsim q <N/2$.

\begin{table}[h]
\centering
\begingroup
\small
\setlength{\tabcolsep}{5pt}
\renewcommand{\arraystretch}{1.15}

\begin{tabular}{ccccc}
\hline
Range
&
Classical
&
Q. attack
&
Our bound
&
Q. attack type
\\
\hline

$q\lesssim N^{1/3}$
&
$\displaystyle
\Theta\left(\frac{q^2}{N^r}\right)\,\star$
&
$\displaystyle
\frac{q^3}{N^r}\,?$
&
$\displaystyle
O\left(\frac{q^3}{N^r}\right)$
&
Collision finding
\\[1.2ex]

$N^{1/3}\lesssim q\lesssim N^{1/2}$
&
$\displaystyle
\Theta\left(\frac{q^2}{N^r}\right)$
&
$\displaystyle
\frac{q^{3/2}}{N^{r-1/2}}\,?$
&
$\displaystyle
O\left(\frac{q^{3/2}}{N^{r-1/2}}\right)$
&
Collision counting
\\[1.2ex]

$N^{1/2}\lesssim q<N/2$
&
$\displaystyle
\Theta\left(\frac{q}{N^{r-1/2}}\right)$
&
$\displaystyle
\frac{\sqrt q}{N^{r-1}}\,?$
&
$\displaystyle
O\left(N^{3/2-r}\right)$
&
Collision counting
\\[1.2ex]

$q\ge N/2$
&
-
&
$\Theta(1)$
&
-
&
Parity
\\

\hline
\end{tabular}

\endgroup
\caption{The classical results, quantum attacks and our bounds. $\star$ denotes that the security bound is only known for $r=2$. The question mark (?) denotes that the attack analysis is heuristic.}
\label{tab:xop-attacks}
\end{table}

\paragraph{Collision probability.}
All attacks above except for the parity attack exploit the collision statistics as follows from \cite{ACNS:Patarin13}.
For two distinct inputs $x,x'$, it holds that
\[
\Pr\left[\xop[r](x) = \xop[r](x')\right] = \frac1N + \frac{(-1)^r}{N(N-1)^{r-1}}.
\]
On the other hand, the collision probability of random functions is $1/N$, having the bias $\delta_r = \frac{(-1)^r}{N(N-1)^{r-1}} \approx (-1)^r N^{-r}$ in the collision probability in two cases.
This bias allows the collision-based distinguishing attack.
For example, running the quantum collision finding algorithm \cite{BHT97} and outputs 1 when a collision is found has the advantage $q^3/N^r$ for $q \lesssim N^{1/3}$.
The details of the attacks are described in \cref{app: attacks}.

\subsection{Related Works}

\paragraph{Classical analyses of $\xop$ and variants.}
Most of the security analysis of $\xop$ is focused on the $r=2$ case.
Various proof techniques have been developed, and the $\xop$ construction has been a testbed of the techniques.
Lucks~\cite{EC:Lucks00} first gave an early asymptotic analysis. Patarin~\cite{ICITS:Patarin08} 
 claimed to prove the bound $O(q^2/N^2)$, but the proof was incomplete.
Dai, Hoang, and
Tessaro~\cite{C:DaiHoaTes17} proved the bound $O(q^{1.5}/N^{1.5})$. Dutta, Nandi, and Saha~\cite{DNS22} resurrect the bound of Patarin~\cite{ICITS:Patarin08}.
The multi-user security was tightened by 
Chen, Choi, and Lee~\cite{C:CheChoLee23}.

Independently, the tight bound \(O(q/N^{1.5})\) was proven using the Fourier analysis by Eberhard \cite{Eberhard17} but had been unknown to the symmetric-key cryptography community.
Dinur extended this to $\xop[r]$ with $r>2$ even in the multi-user setting \cite{EC:Dinur24}, improving the previous analysis for $r>2$ \cite{FSE:CogLamPat14,AC:CKLL22}.

There are $\xop$-like constructions based on other building blocks rather than block ciphers.
For example, 
the sum of the Even-Mansour encryption \cite{C:CheLamMen19} and 
the XOR of the tweakable permutation
\cite{ToSC:ChoLeeLee24} have been studied in the literature.

\paragraph{Quantum analyses of $\xop$ and variants.}
As far as we are aware of, the only quantum result of the $\xop$ construction is \cite{PQCrypto:MenSze17}. They study the $\xop$ in the Q1 model, where the underlying block cipher can be quantumly accessible but the $\xop$ is only available through classical queries. 
They proved the security about $rq^2/2^\kappa + q/2^n$ in the Q1 model.\footnote{To be precise, they assumed that the Grover algorithm is optimal in this attack for the concrete block cipher. This can be proven in the ideal cipher model using \cite{BBBV97} as discussed above.}
Our result gives the Q2 model security bound about $rq^2/2^\kappa + \min(q^3/2^{rn},1/2^{(r-1.5)n})$, which is even better than their Q1 bound for $r\ge 2.$

The variants have been studied in the quantum setting as well. The keyed sum of permutations and other variants are proven secure up to $2^{n/3}$ queries in the Q1 model \cite{AC:DDMNS26}.
For the sum of Even-Mansour, a polynomial time quantum attack is known in the Q2 model \cite{IPL:SI22}.

\if=0
\subsection{Related works}
\minki{TODO related works, may be go to appendix?}
Modern symmetric-key schemes have shared or been inspired by the structure of $\xop$, though may not directly make use of it. In other words, many of them can be viewed as various generalization of the $\xop$ constructions. We briefly summarize these ideas here. We also discuss some recent quantum security results about them.

\paragraph{Encryption and AEAD.} Iwata's \(\mathsf{CENC}\) mode~\cite{FSE:Iwata06} is based on a block cipher and achieves beyond-birthday-bound security while retaining desirable mode properties such as parallelizability, precomputation, and random access. The same work also proposed the authenticated encryption mode \(\mathsf{CHM}\),
which combines \(\mathsf{CENC}\) with a Wegman--Carter style MAC. 
Several subsequent BBB-secure encryption and AEAD schemes also use XoP-style constructions~\cite{AFRICACRYPT:Iwata08,ToSC:BhaNan18,AC:CLLL21,DCC:BDLN22,EC:BHILLM23,EC:CHKLL25,TCHES:ACLMN25}.

\paragraph{MACs.}
Several (deterministic) double-block MACs, including \(\mathsf{SUM\text{-}ECBC}\)~\cite{RSA:Yasuda10}, \(\mathsf{PMAC}_{+}\)~\cite{C:Yasuda11}, \(\mathsf{3kf9}\)~\cite{AC:ZWSW12}, and \(\mathsf{LightMAC}_{+}\)~\cite{AC:Naito17,RSA:Naito18}, can be understood as using a hash value to derive two blockcipher inputs, computing the blockcipher outputs from the corresponding inputs, and Xoring the outputs. Datta et al. abstracted this structure as the double-block Hash-then-Sum \((\mathsf{DbHtS})\) paradigm~\cite{ToSC:DDNP18}. Subsequent works studied generic attacks, tight bounds, and key-reduced variants of this paradigm~\cite{C:LeuNanSib18,EC:KimLeeLee20,C:SWGW21}.

Nonce-based MACs are another important application. The encrypted Wegman--Carter with Davies--Meyer construction \(\mathsf{EWCDM}\)~\cite{C:CogSeu16} and its variants use blockcipher-based masking to obtain beyond-birthday security. Dutta, Nandi, and Talnikar~\cite{EC:DutNanTal19} introduced the nonce-based Enhanced Hash-then-Mask construction \(\mathsf{nEHtM}\), which achieves beyond-birthday security and graceful degradation under nonce misuse. Later works improved the analysis of such constructions,
and developed tighter bounds~\cite{AC:CLLL20,AC:ChoLeeLee24} as well as proposed another way to use outputs of tweakable block ciphers instead of that of block ciphers~\cite{ToSC:CILLLM20}.


%
%

\fi

\ifnum\fullpage=0
\paragraph{AI use disclosure.} The Fourier-analytic approach to the quantum security of XoP was formulated by the authors, with assistance from ChatGPT 5.4 and 5.5 Pro in formalizing many technical details. 
From the authors' perspective, the first nontrivial mathematical contribution from AI was a proof of the $O(q^3/N^r)$ bound of the degree-2 components (\cref{lem: xopdeg2_smallq}), obtained through a series of conversation with ChatGPT 5.5 Pro.
Initial proofs of several other lemmas were also generated by ChatGPT 5.5, 5.6 Pro and Astra (in the Ultra mode), as well as Fable 5.0 and 5.1.
The authors subsequently simplified these proofs and, in many cases, developed substantially different proofs that they considered more natural and accessible.
The current paper is mostly written by human authors from scratch, while AI assisted in writing some paragraphs.

\fi
\section{Preliminaries} \label{sec:prelim}
Let $\Ftwo=\bit$ be the finite field of two elements 0 and 1.
For $x=(x_i)_{i \in [n]},y=(y_i)_{i \in [n]}\in \bit^n$, we define their exclusive OR (XOR) by $x+ y:= (x_i \oplus y_i)_{i\in [n]}$ and inner product by $\inner{x,y} := \sum_{i=1}^n x_i \cdot y_i \bmod 2.$
For a finite set $S$, we write $s \gets S$ to denote that $s$ is sampled uniformly at random from $S$. We also write $(s_1,...,s_t) \gets S^{(*t)}$ to denote that the elements of $\{s_i\}_{i\in[1..t]}$ are sampled uniformly at random from $S$ {\it without replacement}, i.e., $\{s_i\}_{i\in[1..t]}$ is pairwise distinct. 
For an event \(X\), we write \(\truth{X}\) for its indicator, i.e., \(\truth{X}=1\) if \(X\) occurs and \(\truth{X}=0\) otherwise. Throughout this paper, \(n\) denotes a fixed positive integer, and we write \(N := 2^n\). For integers \(a \ge b \ge 0\), we write
\[
(a)_b := a(a-1)\cdots(a-b+1).
\]
\paragraph{Probability distributions.}
For a probabilistic distribution $\dist{D}$ over a finite set $X$, 
we write $\dens{D}(x) = |X|\cdot {\Pr_{y \gets \dist{D}}[x=y]}=|X|\cdot{\Pr[x\gets \dist{D}]}$.
Observe that $\Exp_{x\gets {X}}[\dens{D}(x)] =1$ and the following identity
\begin{align}
  \Exp_{x \gets {X}}\left[\dens{D}(x)\, h(x)\right]
  = \sum_{x \in X} \frac{1}{|X|}
        \cdot \frac{\Pr[x \gets \dist{D}]}{1/|X|} \cdot h(x)
  = \Exp_{x \gets \dist{D}}\left[h(x)\right]
  \label{eq:lift}
\end{align}
which shows that the multiplicative factor $\dens{D}$ converts the expectation with respect to the uniform distribution on ${X}$ into the expectation with respect to our interested distribution $\dist{D}$.

Let $\F$ be a finite function space, and let $\dist{D}$ and $\dist{E}$ be two probability distributions over $\F$. For $f\in\F$, write $O_f$ for the quantum oracle associated with $f$. Namely, $O_f$ maps $\ket{x,y} $ to $\ket{x,y+f(x)}$ coherently.

For a quantum oracle algorithm $\mathcal A$, we define its distinguishing advantage between two distributions $\dist{D}$ and $\dist{E}$ by
\[
    \Adv^{\mathsf{dist}}(\dist{D},\dist{E};\mathcal A)
    :=
    \left|
        \Pr_{f\gets \dist{D}}\big[\mathcal A^{O_f}\to 1\big]
        -
        \Pr_{f\gets \dist{E}}\big[\mathcal A^{O_f}\to 1\big]
    \right|.
\]
For $q\ge 0$, we define the advantage over all $q$-query algorithms
\[
\Adv_q^{\mathsf{dist}}(\dist{D},\dist{E})
:= 
\max_{q\text{-query} ~\mathcal A}
        \Adv^{\mathsf{dist}}(\dist{D},\dist{E};\mathcal A)
\]
where the maximum is over all quantum oracle algorithms making at most $q$ quantum queries to the oracle.

\section{Fourier Analysis of Functionals}\label{sec:fourier}
We denote the set of all functions from $n$-bit inputs to $n$-bit outputs by
\[
\F_n = \{f:\bit^n \to \bit^n\}.
\]
When the bit-length is clear from the context, we omit the subscript and just denote them by $\F$.
The set $\F$ has a natural group structure by extending the XOR operation by $(f+g)(x) = f(x)+g(x)$ for $f,g \in \F$ and $x \in \bit^n$.
We sometimes use the notation $\hF$ with the elements denoted by Greek characters, e.g., $\gamma$ or $\eta$, to denote the identical group $(\F,+)$ to distinguish their use.\footnote{$\hF$ denotes the (Pontryagin) dual group, which is isomorphic to $\F$ since it is finite.} Looking ahead, we will use $\hF$ to denote the set of Fourier basis only.

\paragraph{Functionals.}
We use the Fourier analysis on the set of \emph{functions of functions}:
\[
\H= \left\{
    H:\F \to \C
\right\}.
\]
To clarify the difference, we use the term \newdef{functionals} to denote the elements in $\H$ with the upper case letters. 
By identifying $\F$ as $(\bit^n)^{2^n}$, the Fourier analysis of the functionals is essentially equivalent to one of the Boolean functions with an exponentially long input.

The set of functionals $\H$ is equipped with 
the inner product for $H,G\in \H$
\[
\inner{H,G} = \Exp_{f\gets \F}\left[\overline{H(f)}G(f)\right] = \sum_{f \in \F} \frac{\overline{H(f)}G(f)}{|\F|}.
\]
The corresponding norm is defined by $\|H\|_2 = |\inner{H,H}|^{1/2}$.
The Cauchy–Schwarz inequality implies that
\[
\inner{H,G} \le \|H\|_2 \cdot \|G\|_2.
\]


\subsection{Fourier analysis of the functionals}
\paragraph{Characters.}

For $a\in\bit^n$, let $\chi_a:\bit^n\to\pmone$, $\chi_a(x):=(-1)^{\inner{a,x}}$. These satisfy, for all $a,b,x\in\bit^n$,
\begin{align}\label{eqn: ortho}
    \chi_a(x)\chi_b(x) =  \chi_{a+b} (x)\quad \text{and}\quad \Exp_{x\gets \bit^n}\left[\overline{\chi_a(x)} \chi_b(x)\right] = \truth{a=b}.
\end{align}
For $\gamma\in\hF$, the \newdef{character} $\chi_\gamma\in\H$ is
\begin{align}\label{eqn: character}
    \chi_\gamma(f):=\prod_{x\in\bit^n}\chi_{\gamma(x)}\big(f(x)\big)=(-1)^{\sum_{x\in\bit^n}\inner{\gamma(x),f(x)}}.
\end{align}
The \newdef{support} and the \newdef{degree} of $\gamma\in\hF$ are
\[
    \supp(\gamma):=\{x\in\bit^n:\gamma(x)\neq0^n\}
    \qquad\text{and}\qquad
    |\gamma|:=|\supp(\gamma)|,
\]
so $|\gamma|$ is the number of nontrivial factors in \cref{eqn: character}. Clearly $|\gamma+\eta|\le|\gamma|+|\eta|$. By \cref{eqn: ortho}, characters are multiplicative and orthonormal: for $\gamma,\eta\in\hF$,
\begin{align}\label{eqn: bilinear}
\begin{aligned}
    \inner{\chi_\gamma,\chi_\eta} = \Exp_{f \gets \F}\left[\overline{\chi_\gamma(f) }\chi_\eta(f)\right] = \truth{\gamma=\eta}, \\
    \chi_\gamma\chi_\eta=\chi_{\gamma+\eta},\qquad
    \chi_\gamma(f+g)=\chi_\gamma(f)\,\chi_\gamma(g).
\end{aligned}
\end{align}
Since $|\hF|=|\F|=\dim\H$, the characters form an orthonormal basis of $\H$.

\paragraph{Fourier expansions.}
The orthonormal basis of characters gives the \newdef{Fourier expansion} of $H:\F \to \C$ (i.e., $H\in \H$) as follows.
\begin{theorem}[Fourier expansion of functionals]
\label{thm: Fourier expansion}
    A functional $H:\F \to \C$ can be decomposed by
    \[
    H(f) = \sum_{\gamma \in \hF} \hh(\gamma) \chi_\gamma(f),
    \quad 
    \text{ such that }\quad\hh(\gamma) = \inner{\chi_\gamma,H} = \Exp_{f \gets \F}\left[\overline{\chi_\gamma(f)}H(f)\right],
    \]
    where $\hh(\gamma)$ is called the \newdef{Fourier coefficient} of $H$ at $\gamma$. 
\end{theorem}

\begin{proof}
Since $\{\chi_\gamma\}_{\gamma\in\hF}$ forms an orthonormal basis of $\H$, every $H\in\H$ admits a unique expansion of the form
\[
    H=\sum_{\gamma\in\hF} c_\gamma \chi_\gamma
\]
for some coefficients $c_\gamma\in\C$.

Taking the inner product with $\chi_\eta$ and using orthonormality, we obtain
\[
    \inner{\chi_\eta,H}
    =
    \inner{\chi_\eta,\sum_{\gamma\in\hF}c_\gamma\chi_\gamma}
    =
    \sum_{\gamma\in\hF}c_\gamma\inner{\chi_\eta,\chi_\gamma}
    =
    c_\eta.
\]
Therefore,
\[
    c_\eta=\inner{\chi_\eta,H}
    =
    \Exp_{f\gets\F}\left[\overline{\chi_\eta(f)}H(f)\right].
\]
Since this holds for every $\eta\in\hF$, we have
\[
    H(f)=\sum_{\gamma\in\hF}\widehat H(\gamma)\chi_\gamma(f),
    \qquad
    \widehat H(\gamma)=\inner{\chi_\gamma,H}.
\]
This proves the theorem.
\end{proof}
The \newdef{Fourier spectrum} of $H$ is the collection of all Fourier coefficients
\[
    \left(\hh(\gamma)\right)_{\gamma \in \hF}.
\]
The \newdef{Fourier support} of $H$ is the set of $\gamma$ such that $\hh(\gamma)$ is nonzero.

The orthogonality of the characters show that for $H,G\in \H$,
\begin{align}\label{eqn: inner_Fourier}
    \inner{H,G} = \sum_{\gamma\in \hF} \overline{\widehat{H}(\gamma)}\widehat{G}(\gamma).
\end{align}

For a nonnegative integer $d$, the \newdef{degree-$d$ component} of $H$, denoted by $H^{=d}$, and the \newdef{low-degree components} of $H$ up to degree $d$, 
denoted by $H^{\le d}$, are defined as follows:
\begin{align*}
    H^{=d} := \sum_{|\gamma|=d} \hh(\gamma) \chi_\gamma,
    \qquad \text{ and }\qquad
    H^{\le d} := \sum_{i \le d } H^{=i}.
\end{align*}
It is worth noting that $H^{=0}=\Exp_{f \gets \F} [H(f)] \cdot  \constone$.
We refer to the \newdef{degree-$d$ weight}\footnote{Dinur~\cite{EC:Dinur24} uses $W^{=d}[H]$ instead of $\|H^{=d}\|^2_2$ to refer to the degree-$d$ weight of $H$.}
and the \newdef{degree-$d$ maximal magnitude} of $H$
as
\begin{align}
    \label{eqn: degreeweight}
    \|H^{=d}\|^2_2 = \sum_{|\gamma|=d} \left|\hh(\gamma)\right|^2,\qquad
    M^{=d}[H]:= \max_{|\gamma|=d}\left\{\left|\hh(\gamma)\right|\right\}
\end{align}
where the equality holds due to the orthogonality of the characters.

For a functional $H\in \H$, we say the \newdef{Fourier degree} of $H$, denoted by $\deg(H)$, by the minimum nonnegative integer $d$ such that $H=H^{\le d}$ holds, which means that $H$ is the sum of Fourier characters of degree $\le d$. Constant functionals have Fourier degree 0.
Taking conjugate does not change the degree. 
For $H,G \in \H$, we have
\begin{align}\label{eqn:deg triangle}
    \deg(H+G) \le \max(\deg(H),\deg(G)),\quad\deg(HG) \le \deg(H)+\deg(G).
\end{align}

\paragraph{Convolution.} 
The \newdef{convolution} of $G,H\in\H$ is $(G*H)(f):=\Exp_{g\gets\F}[G(f+g)H(g)]$.
\begin{lemma}
\label{lem: convol}
$\widehat{H *G }(\gamma) = \widehat{H}(\gamma)\widehat{G}(\gamma)$.
\end{lemma}
\ifnum\fullpage=0
We prove this lemma in \cref{app:Fourier} for the completeness.
\else
\begin{proof}
In \cite[Theorem 1.27]{o2014analysis}, the same statement for real-valued boolean functions is proven. It can be extended to our case as follows.
\begin{align*}
\widehat{H *G }(\gamma) ={}& \Exp_{f \gets \F}\left[\overline{\chi_\gamma(f)}H*G(f)\right]\\
={}& \Exp_{f \gets \F}\left[\overline{\chi_\gamma(f)}\Exp_{g\gets \F} [H(f+g)G(g)]\right]\\
={}& \Exp_{g,h \gets \F}\left[\overline{\chi_\gamma(g+h)}H(h)G(g)\right]\\
={}& \Exp_{g,h \gets \F}\left[\overline{\chi_\gamma(h)}H(h)\overline{\chi_\gamma(g)}G(g)\right]\\
={}& \widehat{H}(\gamma)\widehat{G}(\gamma)
\end{align*}
since $\overline{\chi_\gamma}=\chi_\gamma$ and by \cref{eqn: bilinear}.
\end{proof}
\fi

\paragraph{Distributions of XOR of functions.}
One main target of the functional Fourier analysis is the probabilistic density functions over $\F$. Let $\D_D, \D_E$ be distributions over $\F$ and write $\dens{D},\dens{E}$ to denote their density functions, i.e., $\dens{D}(f) = \frac{\Pr[f \gets \dist{D}]}{1/|\F|}$.
\ifnum\fullpage=0
The proof of the following lemma is in \cref{app:Fourier}.
\fi
\begin{lemma}\label{lem: convolDensity}
    $\dens{\D}*\dens{\mathcal E}$ is the density of the distribution of $h=f+g$ where $f\gets\D$ and $g\gets\mathcal E$ are independent.
\end{lemma}
\ifnum\fullpage=1
\begin{proof}
Using $\dens{\D}(f)=|\F|\Pr_{\D}[f]$,
\begin{align*}
    (\dens{\D}*\dens{\mathcal E})(h)
    =\sum_{g\in\F}\frac{\dens{\D}(h+g)\dens{\mathcal E}(g)}{|\F|}
    &=|\F|\sum_{g\in\F}\Pr_{f\gets\D}[f=h+g]\Pr_{g'\gets\mathcal E}[g'=g]\\
    &=|\F|\Pr_{f\gets\D,\,g\gets\mathcal E}[f+g=h].\qedhere
\end{align*}
\end{proof}
\fi




\subsection{Quantum oracle algorithms as functional}
Recall the quantum oracle $O_f$ for $f$ computes $O_f:\ket{x,y} \to \ket{x,y + f(x)}$ coherently.
For $a \in \bit^n$, we define the state $\ket{\widehat a}$ by 
\begin{align}\label{eqn: Fourier_basis_state}
\ket{\widehat a} := \frac{1}{\sqrt{2^{n}}}\sum_{z\in \bit^n} \chi_a(z)\ket{z}  = H^{\otimes n} \ket{a},
\end{align}
then, it is immediate that
\begin{align}\label{eqn: query}
 O_f \ket{x, \widehat a} =(-1)^{\inner{a,f(x)}}\ket{x, \widehat a}= \chi_a(f(x)) \ket{x, \widehat a}
\end{align}
where $\chi_a$ is a character.
This means that a single query increases the degree of the amplitudes by 1. The following theorem shows that this gives a polynomial-like representation of the success probability of the algorithm.

\begin{theorem}\label{thm: Prob by functionals}
    Let $f$ be a function from $\bit^n$ to $\bit^n$.
    Let $\mathcal A$ be a quantum oracle algorithm having access to $O_f$ and outputting a bit $b \in \bit$. 
    Then, there exists a functional $P_{\mathcal A}$ (which may depend on the input of $\mathcal A$, if any) of Fourier degree $\le 2q$ such that 
    the probability that $\mathcal A$ outputs $1$ equals to $P_{\mathcal A}(f)$.
    In other words, the functional $P_{\mathcal A}(f):=\Pr[\mathcal A^{O_f} \to 1]$ has the Fourier degree $\le 2q$ and can be decomposed by
    \[
        \Pr[\mathcal A^{O_f} \to 1] = P_{\mathcal A}(f) = \sum_{|\gamma|\le 2q} \widehat{P}_{\mathcal A}(\gamma) \chi_\gamma(f).
    \]
    In particular, $\|P_{\mathcal A}\|_2 \le 1.$ 
\end{theorem}
We call the functional $P_{\mathcal A}$ above the \newdef{representing polynomial} of $\mathcal A$.
The proof of this theorem is essentially the same as the proof used in the lower bounds by polynomials \cite{JACM:BBCMD01}. 
\ifnum\fullpage=0
We defer the proof of this theorem to \cref{app:Fourier}.
\else
We defer the proof of this theorem to the end of this section.
\fi

We extensively use the following theorem derived from the above theorem.

\begin{theorem}\label{thm: advantage by inner product}
    Let $\dist{D}$ be a distribution of the functions in $\F = \{\bit^n \to \bit^n\}$. Let $\mathcal A$ be a binary output quantum algorithm making at most $q$ quantum queries to the oracle. 
    Let $P_{\mathcal A}=\Pr[\mathcal A^{O_f} \to 1]$ be the functional from \cref{thm: Prob by functionals}.
    Then, it holds that
    \[
    \Pr_{f \gets \dist{D}}[\mathcal A^{O_f} \to 1] - \Pr_{f \gets \F}[\mathcal A^{O_f} \to 1] = \inner{\dens{D} -1,P_{\mathcal A}} =\inner{(\dens{D} -1)^{\le 2q},P_{\mathcal A}}
    \]
    where $f \gets \F$ denotes the uniform random sampling and $\dens{D}(f)$ denotes the density function $\Pr[f \gets \dist{D}] / (1/|\F|)$ of $\dist{D}$.
\end{theorem}

\begin{proof}
    For the random function $f$ sampled from $\dist{D}$, the averaged probability that $\Alg$ outputs 1 is
\begin{align*}
    &\Pr_{f\gets \dist{D}}[\mathcal A^{O_f} \to 1] = \sum_{f\in \F} \Pr[f \gets \dist{D}] \Pr[\mathcal A^{O_f} \to 1] 
    \\&
    =\sum_{f\in \F} \frac{1}{|\F|} \frac{\Pr[f \gets \dist{D}]}{1/|\F|} P_{\mathcal A}(f)
    = \Exp_{f\gets \F}[ \dens{D}(f)P_{\mathcal A}(f)] = \inner{\dens{D},P_{\mathcal A}}
\end{align*}
where in the last equality we use the fact that $\dens{D}$ is real-valued so $\dens{D} = \overline{\dens{D}}$.
Noting that the density function of the uniform distribution is the constant function $\constone$, we have
\begin{align*}
    \Pr_{f\gets \dist{D}}[\mathcal A^{O_f} \to 1]  - \Pr_{f\gets \F }[\mathcal A^{O_f} \to 1]  = 
    \inner{\dens{D},P_{\mathcal A}} - \inner{\constone,P_{\mathcal A}}
    = \inner{\dens{D}-\constone,P_{\mathcal A} }.
\end{align*}
The final equality holds because $P_{\mathcal A}$ is of Fourier degree $\le 2q$ and \cref{eqn: inner_Fourier}.
\end{proof}

The theorem immediately yields the following advantage bound.
\begin{theorem}\label{cor: advantage by inner product}
Let $\dist{\F}$ be the uniform distribution over $\F$ and $\dist{D}$ be an arbitrary distribution over $\F$. For any $q$-query quantum algorithm $\mathcal A$, one has 
    \[
    \Adv^{\mathsf{dist}}(\dist{D},\dist{\F};\mathcal A) = \abs{ \inner{(\dens{D} -1)^{\le 2q},P_{\mathcal A}} } .
    \]
\end{theorem}

By decomposing $\dens{D}-\constone$ into the degree-$d$ components for $d\le 2q$, we have
\begin{align}\label{eqn: inner product decomposition by degree}
    \inner{(\dens{D} -\constone)^{\le 2q},P_{\mathcal A}} = \sum_{d=0}^{2q} \inner{(\dens{D} -\constone)^{=d},P_{\mathcal A}} = \sum_{d=1}^{2q} \inner{\dens{D}^{=d},P_{\mathcal A}} 
\end{align}
because the constant function $\constone$ has degree 0 and 
$\dens{D}^{=0} = \Exp_{f\gets \F}[\dens{D}(f)]  \constone= \constone.$
Therefore, this theorem reduces the PRF advantage for $f\gets \dist{D}$ to analyze the low-degree components of $\dens{D}$. 


\ifnum\fullpage=1
We now prove \cref{thm: Prob by functionals}.

\begin{proof}
    We first show the following claim.
    \begin{claim}
    The final state of a $q$-query algorithm (with the deferred measurement) can be written by
    \[
    \sum_z p_z^{(q)}(f) \ket{z}
    \]
    where $p_z^{(q)}$ is a functional of degree $\le q$.
    \end{claim} 
    \begin{proof}[Proof of Claim]
    Note that a $q$-query algorithm is represented by an alternating sequence of oracle query $O_f$ and an oracle-independent unitary $U$.
    Applying unitary $U$, as a linear map, does not increase the Fourier degree. We only focus on the oracle queries.

    Before proceeding, observe that for each query for the Fourier basis (\cref{eqn: Fourier_basis_state}) it holds that
    \[
    O_f (\chi_\gamma (f) \ket{x,\widehat a})  =\chi_\gamma (f)  \cdot \chi_{a}(f(x))\ket{x,\widehat a} = \chi_{\gamma + a \delta_x}(f) \ket{x,\widehat a}
    \]
    where $a\delta_x \in \hF$ is such that $a\delta_x(x)=a$ and $a\delta_x(y)=0^n$ for $y\neq x$.
    We use \cref{eqn: query} and \cref{eqn: bilinear} in the first and second equality. It holds that
    \[
    |\gamma + a \delta_x| \le |\gamma| + |a\delta_x| \le |\gamma|+1.
    \]

    We use the mathematical induction to prove the claim. For the state before any query, $p_z^{(0)}$ is independent of $f$, thus the base step holds. 
    
    Next, consider the state before applying the $k$-th query $\sum_z p_z^{(k-1)}(f) \ket{z}$ for $k\le q$. The Fourier degree of $p_z^{(k-1)}$ is bounded above by $k-1$ by the inductive hypothesis.
    Then, a single query increases the Fourier degree of each term in $p_z^{(k-1)}$ by at most 1, so that the degree after the query is bounded above by $(k-1)+1=k$.
    This proves the claim.
    \end{proof}

    Write $\ket{\phi_q} = \sum_z p_z(f) \ket{z}$ to denote the final state of the algorithm where $p_z$ is a functional of Fourier degree $\le q$.
    Let $(\Pi_0,\Pi_1=I-\Pi_0)$ be the final binary measurement of the algorithm. Then, the probability that $\mathcal A$ outputs 1 is
    \begin{align*}
        \Pr[\mathcal A^{O_f}\to 1] &= \bra{\phi_q} \Pi_1 \ket{\phi_q}\\
        &= \sum_{z,w} \bra z \overline{p_z(f)} \Pi_1 p_w(f) \ket w\\
        &= \sum_{z,w} \bra z  \Pi_1  \ket w \cdot \overline{p_z(f)}p_w(f)
    \end{align*}
    where $\deg(\overline{p_z}p_w) \le \deg(\overline{p_z}) + \deg(p_w) \le 2q$ because of \cref{eqn:deg triangle}.
    Since the summation does not increase the Fourier degree, $P_{\mathcal A}(f) = \Pr[\mathcal A^{O_f}\to 1]$ is a functional of Fourier degree $\le 2q$.

    The final ``in particular'' part is obvious because $|P_{\mathcal A}(f)|\le 1$ for all $f$.
\end{proof}

\section{Quantum Security of XoP}


This section introduces notation and states our main security theorem for $\xop$ constructions. We first define the set of all permutations on $n$-bit strings by
\[
\P_n = \{p\in \F_n :p \text{ is injective}\}.
\]
In the following, we fix $n$, omit the subscript, and just denote $\P_n$ by $\P$. Recall that we defined the XOR of permutations
\[
    \xop[r] := P_1 + \dots + P_r
\]
where $P_1,\ldots,P_r\gets\P$. We denote the distribution of $\xop[r]$ as $\DXoPr$. In particular, $\D_{\xop} = \D_{\xop[2]}$. 

We consider quantum-accessible $\xop[r]$ constructions in which a quantum adversary can make at most $q$ quantum queries to a given $\xop[r]$ but has no access to the inner primitives $P_r$. 
We note that the asymptotic formula hides multiplicative factors depending on $r$, as $r$ is typically chosen as a constant.

In this setting, the quantum security of $\xop$ can be stated as follows:

\begin{theorem}[Quantum security of XoP]\label{thm: main_bound}
Let $n,r,q$ be integers such that $n\ge5$, $2 \le r \le 2^n/16$ and $0\le q\le 2^n/57774$. Let $N=2^n$. It holds that
\[
    \Adv_q^{\mathsf{dist}}(\DXoPr,\dist{\F})
    = O \left( \min\left\{\frac{q^3}{N^r},\frac{q^{3/2}}{N^{r-1/2}},
    \frac1{N^{r-3/2}}\right\} \right).
\]
\end{theorem}

We prove the last bound in \Cref{sec:bound1}, the first bound in \Cref{sec:bound2}, and
the second bound in \cref{sec: compressed oracle}.

\section{XoP Security for Many Quantum Queries}\label{sec:bound1}
The goal of this section is to prove the following lemma.
We assume that $n\ge 5, 2\le r\le N/16$, and $q\le N/57774$ hold in this section.
\begin{lemma}\label{lem: XoP bound 1}
If $n\ge 5, r\le \frac N{16}$, and $q\le \frac N{57774}$, then
\[
\Adv_q^{\mathsf{dist}}(\DXoPr,\dist{\F}) = O\left(\frac{1}{N^{r-1.5}}\right).
\]
\end{lemma}
To prove the above lemma, we will use Fourier analysis explained in~\Cref{sec:fourier}. We denote the density function of $\xop[r]$ as $\muXoPr$. Then we can express the advantage of $q$-query quantum algorithm $\Alg$ attacking $\xop[r]$ by \Cref{cor: advantage by inner product} and \Cref{eqn: inner product decomposition by degree} as follows:
\begin{align}
\label{eqn: advantage bound by degree}
    \Adv^{\mathsf{dist}}(\DXoPr,\dist{F};\Alg)= \abs{ \inner{(\muXoPr -1)^{\le 2q},P_{\Alg}} } \leq  \sum_{d=1}^{2q} \left|\inner{\muXoPr^{=d},P_{\Alg}^{=d}}\right|.
\end{align}
We bound the advantage using the Cauchy–Schwarz inequality as follows
\begin{align}\label{ineq: xop1}
    \leq \sum_{d=1}^{2q}  \|\muXoPr^{=d}\|_2 \cdot \|P_{\Alg}^{=d}\|_2\leq \left(\sum_{d=1}^{2q}  \|\muXoPr^{=d}\|_2^2 \right)^{\frac{1}{2}}
\end{align}
where we use $\sum_d \|P_\Alg^{=d}\|^2_2 \le \|P_\Alg\|^2_2 \le 1$ as shown in \cref{thm: Prob by functionals}.

The only remaining task is bounding the sum of \(\muXoPr^{=d}\). 
The following lemma shows such upper bounds. We separately deal with $d\in\{1,2,3,4,6\}$, which will be used in the later section.
The remainder of this section will be dedicated to prove this lemma.
\begin{lemma}\label{lem: xop_weights}
    For $r\ge 2$ and $N=2^n \ge \max\{16r, 96\}$ and $5\le d_{\max} \le  N / 28887$,
    \begin{itemize}[nosep]
        \item $\muXoPr^{=1} = \constzero$,
        \item $\|\muXoPr^{=2}\|_2^2 \le {1}/{N^{2r-3}},$
        \item $\|\muXoPr^{=3}\|_2^2 \le O\left( {1}/{N^{4r-5}} \right),$
        \item $\|\muXoPr^{=4}\|_2^2 \le O\left( {1}/{N^{4r-6}} \right),$
        \item $\|\muXoPr^{=6}\|_2^2 \le O\left({1}/{N^{6r-9}}\right),$ and
        \item $\|\muXoPr^{=5}\|_2^2+\sum_{d=7}^{d_{\max}}\|\muXoPr^{=d}\|_2^2
        \le
        O\left({1}/{N^{6r-8}}\right).$
    \end{itemize}
    In particular, all of the above terms are bounded above by $O(1/N^{2r-3})$.
\end{lemma}

Now we can give an upper bound of \Cref{ineq: xop1} by plugging the last part of \cref{lem: xop_weights} by decomposing the summand accordingly.
\begin{align*}
    \left(\sum_{d=1}^{2q}  \|\muXoPr^{=d}\|_2^2 \right)^{\frac{1}{2}} 
    \le 
    \left( O\left(\frac{1}{N^{2r-3}} \right) \right)^{\frac{1}{2}} 
    =  O\left(\frac{1}{N^{r-1.5}}\right).
\end{align*}
which directly proves \Cref{lem: XoP bound 1} as it is an upper bound of $\Adv^{\mathsf{dist}}(\DXoPr,\dist{F};\Alg) .$

\ifnum\fullpage=1
\subsection{Proofs for low degrees}
\else
\subsection{Proof of \Cref{lem: xop_weights}}
\fi
We first observe that the density function of the distribution of $\xop[r]$ follows the convolution $\muP * ... * \muP$ of $r$ $\muP$'s because of \cref{lem: convolDensity}. Then, by \cref{lem: convol}, we have
\begin{align}
\label{eqn:xopconvol}
\hmuXoPr(\gamma)=\hmuP(\gamma)^r.
\end{align}
Therefore, our analysis boils down to analyze the Fourier coefficient $\hmuP(\gamma)$ of uniformly random permutations.
We will use the following formula.
Here, recall $z_1,...,z_d \gets (\bit^n)^{(*d)}$ denotes the sampling without replacement, i.e., $z_i \neq z_j$ for $1\le i<j\le d$.
\begin{lemma}
\label{eqn: simplifiaction}
Suppose $|\gamma|=d$ with $\supp(\gamma) = \{x_1,...,x_d\}.$
Then
\[
\hmuP(\gamma)
=\Exp_{z_1,...,z_d \gets ((\bit^n)^{(*d)}} \left[
       \prod_{i=1}^{d}\chi_{\gamma(x_i)}(z_i)
    \right].
\]
\end{lemma}
\begin{proof}
    We expand $\hmuP(\gamma)$ using the Fourier expansion formula in \cref{thm: Fourier expansion} by
    \[
    \hmuP(\gamma) = \inner{ \chi_\gamma,\muP}
    =\Exp_{f \gets \F} \left[
        \chi_\gamma(f)\muP(f) 
    \right] = 
    \Exp_{p \gets \P} \left[
       \chi_\gamma(p)
    \right]
    \]
    where we use the property of density from \cref{eq:lift} in the last equality.
    Using \cref{eqn: character} and $\supp(\gamma) = \{x_1,...,x_d\}$, we can compute this as follows
    \begin{align*}
    \Exp_{p\gets \P} \left[
       \prod_{i=1}^{d}{\chi_{\gamma(x_i)}(p(x_i))}
    \right] 
    =\Exp_{z_1,...,z_d \gets ((\bit^n)^{(*d)}} \left[
       \prod_{i=1}^{d}\chi_{\gamma(x_i)}(z_i)
    \right]
    \end{align*}
    which completes the proof.
\end{proof}
We also use the following identity occasionally, which can be easily seen by writing $\chi_y(z) = (-1)^{\inner {y,z}}:$
 \begin{align}\label{eqn: characsum}
     \sum_{z\in\bit^n}\chi_y(z) = N\truth{y=0^n}, \quad\text{ and }\quad\sum_{y\in\bit^n} \chi_y(z) = N\truth{z=0^n}.
 \end{align}


\subsubsection{Degree-1 Component Vanishes} 
We first show that the degree-1 component becomes $\constzero$. 
Fix $\gamma \in \hF$ with $|\gamma|=1$, then there exists a unique $x
_1$ such that $\gamma(x_1)=y_1\neq 0^n$. 
Using \cref{eqn: simplifiaction} and \cref{eqn: characsum}, we have
\begin{align*}
    \hmuP(\gamma) = \Exp_{z\gets \bit^n} \left[\chi_{y_1}(z)\right] = 0.
\end{align*}
This shows the degree-1 Fourier coefficient of $\hmuXoPr(\gamma)=\hmuP(\gamma)^r$ is all zero,
proving the first item of \cref{lem: xop_weights}.

\subsubsection{Degree-2 Component} 
For $d=2$, fix $\gamma \in \hF$ with $\supp(\gamma)=\{x_1,x_2\}.$ 
Write $\gamma(x_1) = y_1 \neq 0^n, \gamma(x_2) =y_2 \neq 0^n$. 
\cref{eqn: simplifiaction} gives
\begin{align*}
    \hmuP(\gamma) &= \Exp_{z_1,z_2 \gets (\bit^n)^{(*2)}} \left[
        \chi_{y_1}(z_1)\chi_{y_2}(z_2)
    \right]
    = \sum_{z_1\neq z_2} \frac{\chi_{y_1}(z_1)\chi_{y_2}(z_2)}{N(N-1)}
    \\&
    =
    \frac{\sum_{z_1,z_1}\chi_{y_1}(z_1) \chi_{y_2}(z_2)}{N(N-1)} - \frac{\sum_{z} \chi_{y_1}(z)\chi_{y_2}(z)}{N(N-1)}
\end{align*}
where we omit $z_1,z_2,z \in \bit^n$ in the last terms. 
By \cref{eqn: characsum}, the first term vanishes because $y_1,y_2 \neq 0^n$, and the second term is computed by
\[
- \sum_{z\in \bit^n} \frac{\chi_{y_1}(z)\chi_{y_2}(z)}{N(N-1)}
    = - \sum_{z\in \bit^n} \frac{\chi_{y_1+y_2}(z)}{N(N-1)}
    = 
    -\frac{\truth{y_1+y_2=0^n}}{N-1} 
    =
    -\frac{\truth{y_1=y_2}}{N-1} 
\]
where we use the bilinearity (\cref{eqn: ortho}) and \cref{eqn: characsum}.

Consequently, the degree-2 component becomes
\begin{align}\label{eqn: XoPdeg2}
    \muXoPr^{=2} = \sum_{|\gamma|=2} \hmuP(\gamma)^r \chi_\gamma 
    =\sum_{\supp(\gamma)=\{x_1<x_2\}} \frac{(-1)^r\truth{y_1=y_2} \chi_\gamma}{(N-1)^r}
\end{align}
where we use $\truth{y_1=y_2}^r = \truth{y_1=y_2}$.
Consequently, the degree-$2$ weight is, by letting $y=y_1(=\gamma(x_1))=y_2(=\gamma(x_2))$, computed as in \cref{eqn: degreeweight} by
\begin{align*}
    \|\muXoPr^{=2}\|_2^2
    &=
    \sum_{|\gamma|=2} |\hmuP(\gamma)|^{2r}
    =
    \sum_{x_1<x_2}
    \sum_{y\neq 0^n}
    \frac{1}{(N-1)^{2r}}
    =
    \binom N2 (N-1)\frac{1}{(N-1)^{2r}}
    \\&
    =
    \frac{N}{2(N-1)^{2r-2}}
    \leq
    \frac{1}{2N^{2r-4}(N- 2r+2)}
    \leq
    \frac{1}{N^{2r-3}}
\end{align*}
where the last inequalities use $N/2 \ge 2r-2$.
This proves the second item of \cref{lem: xop_weights}.

\subsubsection{Degree-3 Component} 
Fix $\gamma \in \hF$ with $\supp(\gamma)=\{x_1,x_2,x_3\}.$ 
Write $\gamma(x_i) = y_i \neq 0^n$ for $i=1,2,3.$
We first prove that
\begin{align}
\label{eqn:deg3permutation}
    \hmuP(\gamma)
    =
    \frac{
    \truth{y_1+y_2+y_3=0^n}
    }{\binom{N-1}{2}}.
\end{align}
The proof of this equation uses some combinatorial argument using the inclusion-exclusion principle. The reader may skip the proof of this equation.

\paragraph{Proof of \cref{eqn:deg3permutation}.}
By \cref{eqn: simplifiaction}, this can be written as
\begin{align*}
    \hmuP(\gamma)
    &=
    \sum_{z_1,z_2,z_3 \text{ distinct}}
    \frac{
        \chi_{y_1}(z_1)
        \chi_{y_2}(z_2)
        \chi_{y_3}(z_3)
    }{(N)_3}
    \\&=
    \sum_{z_1,z_2,z_3\in \bit^n}
    \frac{\truth{z_1,z_2,z_3 \text{ distinct}}
        \chi_{y_1}(z_1)
        \chi_{y_2}(z_2)
        \chi_{y_3}(z_3)
    }{(N)_3}.
\end{align*}

For three variables $z_1,z_2,z_3$, applying the inclusion-exclusion principle to the three events
$E_{ij}=[z_i=z_j]$ gives
\[
    \truth{z_1,z_2,z_3 \text{ distinct}}
    =
    1-\sum_{i<j}\truth{z_i=z_j}
    +2\truth{z_1=z_2=z_3}.
\]
where the last term comes from
\[
    \truth{E_{12}\cap E_{13}}
    +
    \truth{E_{12}\cap E_{23}}
    +
    \truth{E_{13}\cap E_{23}}
    -
    \truth{E_{12}\cap E_{13}\cap E_{23}},
\]
all of which are the same event $[z_1=z_2=z_3]$.

Applying this identity to the summation over all $z_1,z_2,z_3$, we have
\begin{align*}
    \hmuP(\gamma)
    ={}&
    \sum_{z_1,z_2,z_3\in \bit^n}
    \frac{
        \chi_{y_1}(z_1)
        \chi_{y_2}(z_2)
        \chi_{y_3}(z_3)
    }{(N)_3}
    \\
    &
    -
    \sum_{i<j}
    \sum_{z_1,z_2,z_3\in \bit^n: z_i=z_j}
    \frac{
        \chi_{y_1}(z_1)
        \chi_{y_2}(z_2)
        \chi_{y_3}(z_3)
    }{(N)_3}
    \\
    &
    +
    2\sum_{z\in \bit^n}
    \frac{
        \chi_{y_1}(z)
        \chi_{y_2}(z)
        \chi_{y_3}(z)
    }{(N)_3} .
\end{align*}
This can be factored into
\begin{align*}
    &=\frac{
        (\sum_{z_1\in \bit^n}\chi_{y_1}(z_1))
        (\sum_{z_2\in \bit^n}\chi_{y_2}(z_2))
        (\sum_{z_3\in \bit^n}\chi_{y_3}(z_3))
    }{(N)_3}
    \\
    &
    -
    \sum_{1\le i<j\le 3; k:=6-i-j}
    \frac{
        (\sum_{z\in \bit^n}\chi_{y_i}(z)\chi_{y_j}(z))
        (\sum_{z_k\in \bit^n}\chi_{y_k}(z_k))
    }{(N)_3}
    \\&
    +
    2\sum_{z\in \bit^n}
    \frac{
        \chi_{y_1}(z)
        \chi_{y_2}(z)
        \chi_{y_3}(z)
    }{(N)_3}
\end{align*}
where $k$ is chosen so that $\{i,j,k\}=\{1,2,3\}$. 
The first and second terms become 0 by \cref{eqn: characsum}, because they have a factor $\sum_{z_w\in \bit^n}\chi_{y_w}(z_w)=0$ for some $w$.
Thus only the last term remains, proving
\begin{align*}
    \hmuP(\gamma)
    &=
    \sum_{z\in \bit^n}
    \frac{2
        \chi_{y_1}(z)
        \chi_{y_2}(z)
        \chi_{y_3}(z)
    }{(N)_3}
    =
    \sum_{z\in \bit^n}
    \frac{2
        \chi_{y_1+y_2+y_3}(z)
    }{(N)_3}
    \\
    &=
    \frac{2N}{(N)_3}
    \truth{y_1+y_2+y_3=0^n}
    =
    \frac{
    \truth{y_1+y_2+y_3=0^n}
    }{\binom{N-1}{2}},
\end{align*}
where we use \cref{eqn: characsum} in the third equality.

\paragraph{Bounding the degree-3 component.}
Recall $ \hmuXoPr(\gamma)=\hmuP(\gamma)^r$ from \cref{eqn:xopconvol}. 
Therefore, the degree-3 weight is
\[
    \|\muXoPr^{=3}\|_2^2
    =
    \sum_{|\gamma|=3}|\hmuP(\gamma)|^{2r}
    =
    \sum_{\substack{\supp(\gamma)=\{x_1<x_2<x_3\}\\\gamma(x_i)=y_i \neq0^n}}
    \frac{
    \truth{y_1+y_2+y_3=0^n}
    }{\binom{N-1}{2}^{2r}}.
\]
The last term is straightforward to calculate. The number of possible support is $\binom{N}{3}$, and the number of nonzero triples
$(y_1,y_2,y_3)$ such that $y_1+y_2+y_3=0^n$
is $(N-1)(N-2)$.\footnote{Choose $y_1\neq 0^n$ and
$y_2\neq 0^n,y_1$, then $y_3$ is determined and nonzero.} Therefore,
\begin{align*}
    \|\muXoPr^{=3}\|_2^2
    &=
    \binom{N}{3}
    (N-1)(N-2)
    \frac{1}{\binom{N-1}{2}^{2r}}
    =\frac{2^{2r-1}}{3N^{4r-5}\left(1-\frac1N\right)^{2r-2}\left(1-\frac2N\right)^{2r-2}}
    \\
    &\leq 
    \frac{2^{2r-1}}{3N^{4r-5}\left(1-\frac{6r-6}{N}\right)} \leq \frac{2^{2r}}{3N^{4r-5}}
    < 
    \frac{2^{2r-1}}{N^{4r-5}}.
\end{align*}
where we use $(1-\frac{1}{N})^{2r-2}(1-\frac{2}{N})^{2r-2} \ge 1-\frac{6r-6}{N}$ in the first inequality.
The last two inequalities follow from \(N/2 \ge 6r-6\) and $3>2$. This concludes the proof of the third item of \cref{lem: xop_weights}.

\subsubsection{Degree-4 Component}
Fix $\gamma \in \hF$ with $\supp(\gamma)=\{x_1,x_2,x_3,x_4\}.$ 
Write $\gamma(x_i) = y_i \neq 0^n$ for $i=1,2,3,4.$
\ifnum\fullpage=0
We prove the following identity in \cref{app: deg4}:
\else
We prove the following identity at the end of this section.
\fi
\begin{align}\label{eqn:deg4permutation}
    \hmuP(\gamma)=\dfrac{N\,C_{2,2}(\gamma)-6\,C_4(\gamma)}{(N-1)_3}
\end{align}
using the inclusion-exclusion principle similar to the proof of \cref{eqn:deg3permutation},
where
\begin{align*}
    C_{2,2}(\gamma)&:=\truth{y_1=y_2}\truth{y_3=y_4}+\truth{y_1=y_3}\truth{y_2=y_4}+\truth{y_1=y_4}\truth{y_2=y_3},\\
    C_4(\gamma)&:=\truth{y_1+y_2+y_3+y_4=0^n}.
\end{align*}

In the main body, we derive the upper bound of the degree-4 component of the $\xop[r]$ using this formula. 
Observe that $0 \le C_{2,2}(\gamma)\le 3$ and $0 \le C_4(\gamma)\le 1$.
This gives the upper bound of the maximal magnitude
\begin{align*}
    M^{=4}[\muP] = \max_{|\gamma|=4}|\hmuP(\gamma)| \le \frac{3N}{(N-1)_3} .
\end{align*}
Using $(a-b)^2\le a^2+b^2$ for $a,b\ge0$, we have
\begin{align}\label{eqn: deg4pointwise}
    |\hmuP(\gamma)|^{2r}
    &\le
    \left(M^{=4}[\muP]\right)^{2r-2}
    |\hmuP(\gamma)|^2
    \notag\\
    &\le
    \left(\frac{3N}{(N-1)_3}\right)^{2r-2}
    \frac{
        N^2C_{2,2}(\gamma)^2+36C_4(\gamma)
    }{((N-1)_3)^2}.
\end{align}

Again, using $0\le C_{2,2}(\gamma)\le3$ and $0\le C_{4}(\gamma) \le 1$, we can give upper bounds
\[
    \sum_{|\gamma|=4}C_{2,2}(\gamma)^2
    \le
    3\sum_{|\gamma|=4}C_{2,2}(\gamma)
    =
    9\binom N4(N-1)^2
    ,~
    \sum_{|\gamma|=4}C_4(\gamma)
    \le
    \binom N4(N-1)^3
\]
where $\binom{N}{4}$ denotes the number of the support of $\gamma$, and $(N-1)^2$ and $(N-1)^3$ is the number of $y_1,y_2,y_3,y_4\neq 0^n$ satisfying the constraints from $C_{2,2}(\gamma)$ and $C_4(\gamma)$, e.g., $\truth{y_1=y_2}\truth{y_3=y_4}.$
Therefore, \cref{eqn: deg4pointwise} is bounded above by
\begin{align}
    \label{eqn: simplified deg4 upper}
    \left(\frac{3N}{(N-1)_3}\right)^{2r-2}
    \frac{
        \binom N4
        \left(
            9N^2(N-1)^2+36(N-1)^3
        \right)
    }{((N-1)_3)^2} \le \frac{3^{2r-1}}{N^{4r-6}}
\end{align}
\ifnum\fullpage=0
where the detailed calculation can be found in \cref{app: deg4}.
\else
The detailed calculation is as follows:
\begin{align*}
    &
    \left(\frac{3N}{(N-1)_3}\right)^{2r-2}
    \frac{
        \binom N4
        \left(
            9N^2(N-1)^2+36(N-1)^3
        \right)
    }{((N-1)_3)^2}
    \\
    &\qquad\qquad=
    \frac{3^{2r}}{N^{4r-6}}
    \frac{
        \displaystyle
        \frac1{24}+\frac{N-1}{6N^2}
    }{
        \displaystyle
        \left(1-\frac1N\right)^{2r-3}
        \left(1-\frac2N\right)^{2r-1}
        \left(1-\frac3N\right)^{2r-1}
    }.
\end{align*}
Observe that
$    \frac1{24}+\frac{N-1}{6N^2}\le \frac{1}{12}$ for $N\ge 4$
and
\begin{align*}
    \left(1-\frac1N\right)^{2r-3}
    \left(1-\frac2N\right)^{2r-1}
    \left(1-\frac3N\right)^{2r-1}
    \ge
    1-\frac{12r-8}{N}
    \ge \frac14
\end{align*}
for $N\ge16r$. This bounds the above term by
\[
\le \frac{3^{2r}}{N^{4r-6}} \cdot \frac{1/12}{1/4} = \frac{3^{2r-1}}{N^{4r-6}}
\]
as desired.
\fi
Hence
\[
    \|\muXoPr^{=4}\|_2^2
    \le
    \frac{3^{2r-1}}{N^{4r-6}}.
\]
This concludes the proof of the fourth item of \cref{lem: xop_weights}.

\paragraph{Proof of \cref{eqn:deg4permutation}.}
Fix $\gamma \in \hF$ with $\supp(\gamma)=\{x_1,x_2,x_3,x_4\}.$ Write $\gamma(x_i) = y_i \neq 0^n$ for $i\in [4]$. For convenience, we recall \cref{eqn:deg4permutation}:
\[
    \hmuP(\gamma)=\dfrac{N\,C_{2,2}(\gamma)-6\,C_4(\gamma)}{(N-1)_3}.
\]
By \cref{eqn: simplifiaction},
\begin{align}\label{eqn: simple4}
    \hmuP(\gamma)
    =
    \sum_{z_1,z_2,z_3,z_4 \text{ distinct}}
    \frac{
    \chi_{y_1}(z_1)\chi_{y_2}(z_2)
    \chi_{y_3}(z_3)\chi_{y_4}(z_4)
    }{(N)_4}.
\end{align}

For four variables $z_1,z_2,z_3,z_4$, the inclusion-exclusion principle
for the six events $E_{ij}=[z_i=z_j]$ gives, after simplification,
\begin{align*}
    &\truth{z_1,z_2,z_3,z_4 \text{ distinct}}
    =
    1-\sum_{i<j}\truth{z_i=z_j}
    +2\sum_{i<j<k}\truth{z_i=z_j=z_k}
    \\
    &\qquad
    +
    \sum_{i<j, k<\ell, \{i,j,k,\ell\}=\{1,2,3,4\}}
    \truth{z_i=z_j\land z_k=z_\ell}
    -6\truth{z_1=z_2=z_3=z_4}.
\end{align*}
Applying this to \cref{eqn: simple4} over all $z_1,z_2,z_3,z_4$,
every term with a separated single vertex vanishes by \cref{eqn: characsum}.
Hence only the last two terms remain:
\begin{align*}
    \hmuP(\gamma)
    &=
    \sum_{\substack{i<j, k<\ell\\ \{i,j,k,\ell\}=\{1,2,3,4\}}}
    \sum_{z,w\in \bit^n}
    \frac{\chi_{y_i+y_j}(z)\chi_{y_k+y_\ell}(w)}{(N)_4}
    -
    \sum_{z\in \bit^n}
    \frac{6\chi_{y_1+y_2+y_3+y_4}(z)}{(N)_4}
    \notag \\
    &=
    \sum_{\substack{i<j, k<\ell\\ \{i,j,k,\ell\}=\{1,2,3,4\}}}
    \frac{N\truth{y_i=y_j}\truth{y_k=y_\ell}}{(N-1)_3}
    -
    \frac{6\truth{y_1+y_2+y_3+y_4=0^n}}{(N-1)_3}
    \notag \\
    &=
    \frac{N C_{2,2}(\gamma)-6C_4(\gamma)}{(N-1)_3},
\end{align*}
where we use \cref{eqn: characsum} in the second equality.
This proves \cref{eqn:deg4permutation}.

\ifnum\fullpage=1
\subsection{Proofs for high degrees}
We prove the high-degree upper bounds.
The proof of the following two lemmas can be found in the next subsection.
\begin{lemma}\label{lem: deg d magnitude bound}
For $1\le d\le N$, it holds that
    $M^{=d}[\muP]
    \le
    \left({2d}/{N}\right)^{\lceil d/2\rceil}.$
\end{lemma}
\begin{lemma}\label{lem: deg d weight bound}
For $5\le d\le N/16$, it holds that
    $\sum_{|\gamma|=d}|\hmuP(\gamma)|^2
    \le
    3^{3d}
    \left( N/d\right)^{\lfloor d/2\rfloor}.$
\end{lemma}


Now we extend the low-degree calculations to the higher-degree terms.
We prove the following lemma.

\begin{lemma}\label{lem: higher degree tail}
For $r\ge 2$ and $5\le d_{\max}\le N/28887$, it holds that $\|\muXoPr^{=6}\|_2^2 \leq 3^{18}2^3\left(\frac{12}{N}\right)^{6r-9}$ and
\[
    \|\muXoPr^{=5}\|_2^2
    +
    \sum_{d=7}^{d_{\max}}
    \|\muXoPr^{=d}\|_2^2
    \le
    3^{22}2^{11}\left(\frac{10}{N}\right)^{6r-8}, \qquad \|\muXoPr^{=6}\|_2^2 = O\left( \frac{1}{N^{6r-9}}\right)
\]
\end{lemma}

\begin{proof}
\cref{lem: deg d magnitude bound,lem: deg d weight bound} give,
for every $5\le d\le d_{\max}$,
\begin{align*}
    \|\muXoPr^{=d}\|_2^2
    &=
    \sum_{|\gamma|=d}|\hmuP(\gamma)|^{2r}
    \\
    &\le
    \left(M^{=d}[\muP]\right)^{2r-2}
    \sum_{|\gamma|=d}|\hmuP(\gamma)|^2
    \\
    &\le
    \left(\frac{2d}{N}\right)^{
        (2r-2)\lceil d/2\rceil
    }
    3^{3d}
    \left(\frac Nd\right)^{\lfloor d/2\rfloor}
    \\
    &=
    3^{3d}
    2^{(2r-2)\lceil d/2\rceil}
    \left(\frac dN\right)^{
        (2r-1)\lceil d/2\rceil-d
    }
    =:
    W_d,
\end{align*}
where we use the following in the last line
\[
    \left\lfloor\frac d2\right\rfloor
    +
    \left\lceil\frac d2\right\rceil
    =
    d.
\]
Degree-6 is immediate by $W_6$, in particular $N\geq 96$. We will prove
\[
    \frac{W_{d+2}}{W_d}\le\frac12
\]
for every $5\le d\le d_{\max}-2$. Since
\(
    \left\lceil\frac{d+2}{2}\right\rceil
    =
    \left\lceil\frac d2\right\rceil+1,
\)
we have
\begin{align*}
    \frac{W_{d+2}}{W_d}
    &=
    3^6 2^{2r-2}
    \left(1+\frac2d\right)^{
        (2r-1)\lceil d/2\rceil-d
    }
    \left(\frac{d+2}{N}\right)^{2r-3}.
\end{align*}
For $d\ge5$,
\[
    \left\lceil\frac d2\right\rceil
    \le
    \frac{3d}{5},
\]
and hence
\[
    (2r-1)\left\lceil\frac d2\right\rceil-d
    \le
    \frac{6r-8}{5}d.
\]
Using $1+x\le e^x$, we obtain
\begin{align*}
    \left(1+\frac2d\right)^{
        (2r-1)\lceil d/2\rceil-d
    }
    &\le
    \exp\left(
        \frac2d
        \left(
            (2r-1)\left\lceil\frac d2\right\rceil-d
        \right)
    \right)
    \le
    e^{(12r-16)/5}.
\end{align*}
In addition, we observe that, for all $r\ge 2$,
\[
            3^6 2^{2r-2}e^{(12r-16)/5}
    < \frac{1}{2} \cdot 28887^{2r-3}.
\]
Using the above,
\begin{align*}
    \frac{W_{d+2}}{W_d}
    &\le
    3^6 2^{2r-2}e^{(12r-16)/5}
    \left(\frac{d+2}{N}\right)^{2r-3}
    \le
    \frac{1}{2} \left( 28887 \right) ^{2r-3}
    \left(\frac{d_{\max}}{N}\right)^{2r-3}
    \le
    \frac12
\end{align*}

Summing the odd and even degrees separately gives
\begin{align*}
    &\|\muXoPr^{=5}\|_2^2
    +
    \sum_{d=7}^{d_{\max}}
    \|\muXoPr^{=d}\|_2^2
    \le
    W_5+
    \sum_{d=7}^{d_{\max}}W_d
    \\
    &\le
    \left(
        1+\frac12+\frac1{2^2}+\cdots
    \right)
    (W_5+W_8)
    =
    2(W_5+W_8).
\end{align*}
Moreover,
\[
    \frac{W_5}{W_8}
    =
    \frac{5^{6r-8}}{3^9\cdot 2^{26r-38}}N^{2r-4}
    =
    \frac{5^4}{3^9\cdot 2^{14}}\left(\frac{125N}{8192}\right)^{2r-4}
    \ge
    \frac{5^4}{3^9\cdot 2^{14}}.
\]
Consequently,
\begin{align*}
    &\|\muXoPr^{=5}\|_2^2
    +
    \sum_{d=7}^{d_{\max}}
    \|\muXoPr^{=d}\|_2^2
    \le
    2\cdot\left(1+\frac{3^9\cdot2^{14}}{5^4}\right)W_5
    \\&\le
    2\cdot\frac{3^9\cdot2^{14}}{3^2\cdot2^6}\cdot3^{15}\cdot2^{6r-6}
    \left(\frac5N\right)^{6r-8}
    =
    3^{22}2^{11}
    \left(\frac{10}{N}\right)^{6r-8}.\qedhere
\end{align*}
\end{proof}

\subsection{Maximal Magnitude and Weight of the components of $\muP$}\label{app: magniweight}

\subsubsection{Maximal magnitude of components}
The goal of this section is to prove \cref{lem: deg d magnitude bound}, or more precisely
\[
    M^{=d}[\muP]
    \le
    \left(\frac{2d}{N}\right)^{\lceil d/2\rceil}.
\]

Let $\gamma\in\hF$ be a character with $|\gamma|=d$ such that $\supp(\gamma)=\{x_1<\cdots<x_d\}.$
Write $\gamma(x_i)=y_i$ for $i\in[d].$
As shown in \cref{eqn: simplifiaction}, we can write
\[
    \hmuP(\gamma)
    =
    \Exp_{(z_1,\ldots,z_d)\gets(\bit^n)^{(*d)}}
    \left[
        \prod_{i=1}^d\chi_{y_i}(z_i)
    \right].
\]
We observe the following fact: The conditional distribution of $z_d$ given the other elements
$(z_1,\ldots,z_{d-1})$ is uniform over the set
$\bit^n\setminus\{z_1,\ldots,z_{d-1}\},$
which has size $N-d+1$.
From this, we have
\begin{align*}
    &\Exp_{z_d\gets
    \bit^n\setminus\{z_1,\ldots,z_{d-1}\}}
    \left[
        \chi_{y_d}(z_d)
    \right]
    \\
    &\quad=
    \frac1{N-d+1}
    \sum_{z\notin\{z_1,\ldots,z_{d-1}\}}
    \chi_{y_d}(z)
    \\
    &\quad=
    \frac1{N-d+1}
    \left(
        \sum_{z\in\bit^n}\chi_{y_d}(z)
        -
        \sum_{i=1}^{d-1}\chi_{y_d}(z_i)
    \right)
    \\
    &\quad=
    -\frac1{N-d+1}
    \sum_{i=1}^{d-1}\chi_{y_d}(z_i),
\end{align*}
where the last equality follows from
\(
    \sum_{z\in\bit^n}\chi_{y_d}(z)=0
\)
because $y_d\neq0^n$.

Using this equation, we have
\begin{align*}
    \hmuP(\gamma)
    &=
    \Exp_{(z_1,\ldots,z_d)\gets(\bit^n)^{(*d)}}
    \left[
        \chi_{y_d}(z_d)
        \prod_{j=1}^{d-1}\chi_{y_j}(z_j)
    \right]
    \\
    &=
    \Exp_{(z_1,\ldots,z_{d-1})\gets(\bit^n)^{(*(d-1))}}
    \left[
        \Exp_{z_d\gets
        \bit^n\setminus\{z_1,\ldots,z_{d-1}\}}
        \left[
            \chi_{y_d}(z_d)
        \right]
        \prod_{j=1}^{d-1}\chi_{y_j}(z_j)
    \right]
    \\
    &=
    -\frac1{N-d+1}
    \sum_{i=1}^{d-1}
    \Exp_{(z_1,\ldots,z_{d-1})\gets(\bit^n)^{(*(d-1))}}
    \left[
        \chi_{y_d}(z_i)
        \prod_{j=1}^{d-1}\chi_{y_j}(z_j)
    \right]
    \\
    &=
    -\frac1{N-d+1}
    \sum_{i=1}^{d-1}
    \Exp_{(z_1,\ldots,z_{d-1})\gets(\bit^n)^{(*(d-1))}}
    \left[
        \chi_{y_i+y_d}(z_i)
        \prod_{\substack{j=1\\j\neq i}}^{d-1}
        \chi_{y_j}(z_j)
    \right].
\end{align*}

To simplify the above expression, define
$\gamma^{(i\leftarrow d)}\in\hF$ by
\[
    \gamma^{(i\leftarrow d)}(x)
    :=
    \begin{cases}
        y_i+y_d,
        & x=x_i,\\
        y_j,
        & x=x_j
          \text{ for some }j\in[d-1]\setminus\{i\},\\
        0^n,
        & \text{otherwise}.
    \end{cases}
\]
If $y_i+y_d=0^n$, then $\gamma^{(i\leftarrow d)}(x_i)=0$ and therefore $\gamma^{(i\leftarrow d)}$ has degree $d-2$. In other cases, $\gamma^{(i\leftarrow d)}$ has degree $d-1$.
This gives the following recursion.

\begin{lemma}\label{lem: xoprecursion}
Let $2\le d\le N$, and let $\gamma\in\hF$ satisfy $|\gamma|=d$.
Then
\[
    \hmuP(\gamma)
    =
    -\frac1{N-d+1}
    \sum_{i=1}^{d-1}
    \hmuP\!\left(\gamma^{(i\leftarrow d)}\right).
\]
\end{lemma}

We are now ready to prove \cref{lem: deg d magnitude bound}.

\begin{proof}[Proof of \cref{lem: deg d magnitude bound}]
It is immediate that $M^{=0}[\muP]=1$ and $M^{=1}[\muP]=0.$
Now fix $2\le d\le N$, and let $\gamma\in\hF$ satisfy
$|\gamma|=d$. By \cref{lem: xoprecursion},
\[
    \hmuP(\gamma)
    =
    -\frac1{N-d+1}
    \sum_{i=1}^{d-1}
    \hmuP\!\left(\gamma^{(i\leftarrow d)}\right).
\]
For each $i\in[d-1]$, the character
$\gamma^{(i\leftarrow d)}$ has support size either $d-1$ or $d-2$.
Therefore,
\[
    \left|
        \hmuP\!\left(\gamma^{(i\leftarrow d)}\right)
    \right|
    \le
    \max\left\{
        M^{=(d-1)}[\muP],
        M^{=(d-2)}[\muP]
    \right\}.
\]
Since the recursion holds for every degree-$d$ character, taking
absolute values gives
\begin{align}\label{ineq: recursive}
    M^{=d}[\muP]
    \le
    \frac{d-1}{N-d+1}
    \max\left\{
        M^{=(d-1)}[\muP],
        M^{=(d-2)}[\muP]
    \right\}.
\end{align}
We now prove the claimed bound by induction on $d\ge 2$. If $d>N/2$, the bound is immediate as $M^{=d}[\muP]\le 1\le (2d/N)^{\lceil d/2\rceil}$.
So let $2\le d\le N/2$ and assume the claim for all smaller degrees. Since $2d/N\le1$ and $\lceil (d-1)/2\rceil\ge\lceil (d-2)/2\rceil$, the induction hypothesis gives
\begin{align*}
    &\max\left\{
        M^{=(d-1)}[\muP],
        M^{=(d-2)}[\muP]
    \right\}\\
    &\quad\le
    \max\left\{
        \left(\frac{2(d-1)}N\right)^{\lceil(d-1)/2\rceil},
        \left(\frac{2(d-2)}N\right)^{\lceil(d-2)/2\rceil}
    \right\}
    \le
    \left(\frac{2d}{N}\right)^{\lceil(d-2)/2\rceil}.
\end{align*}
Moreover, $\frac{d-1}{N-d+1}\le\frac{2d}{N}$ for $d\le N/2$.
Plugging both bounds into~\eqref{ineq: recursive},
\[
    M^{=d}[\muP]
    \le
    \frac{2d}{N}
    \left(\frac{2d}{N}\right)^{\lceil(d-2)/2\rceil}
    =
    \left(\frac{2d}{N}\right)^{\lceil d/2\rceil},
\]
this completes the induction and proves the lemma.
\end{proof}

\subsubsection{Weight of components}
We prove that, for $5\le d\le N/16$, it holds that
\[
    \sum_{|\gamma|=d}|\hmuP(\gamma)|^2
    \le
    3^{3d}
    \left(\frac Nd\right)^{\lfloor d/2\rfloor}.
\]

\begin{proof}[Proof of \cref{lem: deg d weight bound}]
By \cref{eqn: simplifiaction}, $\hmuP(\gamma)$ depends only on the labels of $\gamma$, so for a support $S$ of size $d$ the quantity
\[
    B_d:=\sum_{\supp(\gamma)=S}|\hmuP(\gamma)|^2
\]
does not depend on $S$, and $\sum_{|\gamma|=d}|\hmuP(\gamma)|^2=\binom Nd B_d$. We will show that
\begin{align}\label{eqn: Bd identity}
    B_d = \sum_{j=0}^d (-1)^{d-j}\binom{d}{j}\frac{N^j}{(N)_j}
\end{align}
and
\begin{align}\label{eqn: Bdupper}
    B_d \le 2^{d+1} \left(\frac{8d}N\right)^{\lceil d/2 \rceil}
    \qquad\text{for } 5\le d\le N/16.
\end{align}
Given~\eqref{eqn: Bdupper}, the lemma follows from $\binom Nd\le(eN/d)^d$ and $d-\lceil d/2\rceil=\lfloor d/2\rfloor$:
\begin{align*}
    \sum_{|\gamma|=d}|\hmuP(\gamma)|^2
    =
    \binom Nd B_d
    &\le
    2^{d+1}
    \left(\frac{eN}{d}\right)^d
    \left(\frac{8d}{N}\right)^{\lceil d/2\rceil}\\
    &=
    2^{d+1}8^{\lceil d/2\rceil}e^d
    \left(\frac Nd\right)^{\lfloor d/2\rfloor}
    \le
    3^{3d}
    \left(\frac Nd\right)^{\lfloor d/2\rfloor},
\end{align*}
where the last inequality uses $2^{d+1}8^{\lceil d/2\rceil}\le 8^d$ for $d\ge 5$ (equivalently, $d+1\le 3\lfloor d/2\rfloor$) and $8e<27$.
 
\paragraph{The identity \cref{eqn: Bd identity}.}
Fix $S=\{x_1,\ldots,x_d\}$ and, for $y=(y_1,\ldots,y_d)\in(\bit^n)^d$, let $\gamma_y\in\hF$ be the index with $\gamma_y(x_i)=y_i$ and $\gamma_y(x)=0^n$ for $x\notin S$; thus $\supp(\gamma_y)=\{x_i:y_i\neq 0^n\}$.
Let $\nu$ be the uniform distribution on $(\bit^n)^{(*d)}$, viewed as a probability mass function on $(\bit^n)^d$. By~\eqref{eqn: simplifiaction} (which does not require the $y_i$ to be nonzero),
\[
    \hmuP(\gamma_y)=\sum_{z\in(\bit^n)^d}\nu(z)\prod_{i=1}^d\chi_{y_i}(z_i)
    \qquad\text{for all }y\in(\bit^n)^d,
\]
i.e., $(\hmuP(\gamma_y))_{y}$ is $N^d$ times the Fourier transform of $\nu$ on the group $(\bit^n)^d$. Parseval's identity on $(\bit^n)^d$ therefore gives
\[
    \sum_{y\in(\bit^n)^d}|\hmuP(\gamma_y)|^2
    =N^{2d}\cdot\frac{1}{N^d}\sum_{z\in(\bit^n)^d}\nu(z)^2
    =N^d\cdot\frac{(N)_d}{(N)_d^2}
    =\frac{N^d}{(N)_d}.
\]
On the other hand, if exactly the coordinates in $J\subseteq[d]$ of $y$ are nonzero, then $\hmuP(\gamma_y)$ is the Fourier coefficient of a degree-$|J|$ index, so grouping the $y$ according to $J$ yields $\sum_{y\in(\bit^n)^d}|\hmuP(\gamma_y)|^2=\sum_{J\subseteq[d]}B_{|J|}=\sum_{j=0}^d\binom djB_j$, with $B_0:=1$. Hence $N^d/(N)_d=\sum_{j=0}^d\binom dj B_j$ for every $d\ge 0$, and binomial inversion gives~\cref{eqn: Bd identity}.
 
\paragraph{The bound~\cref{eqn: Bdupper}.}
The right-hand side of~\cref{eqn: Bd identity} resembles the binomial expansion of $(1-1)^d=0$, and we exploit the resulting cancellations. Recall the elementary fact that
\begin{align}\label{eqn: binomial cancellation}
    \sum_{j=0}^d
    (-1)^{d-j}\binom{d}{j}p(j)
    =
    0
    \qquad
    \text{for every polynomial $p$ with } \deg(p)<d.
\end{align}
Write
\[
    \frac{N^j}{(N)_j}
    =
    \prod_{\ell=0}^{j-1}\frac{1}{1-\ell/N}
    =
    u_j(1/N),
    \qquad
    u_j(t)
    :=
    \prod_{\ell=0}^{j-1}\frac{1}{1-\ell t}
    =
    \sum_{k\ge0}p_k(j)t^k,
\]
where the last expression is the Taylor expansion of $u_j$ at $t=0$, which converges absolutely for $|t|<1/(j-1)$, and in particular at $t=1/N$ for all $j\le d\le N/16$.
Comparing the coefficients of $t^k$ in $u_{j+1}(t)=u_j(t)/(1-jt)=u_j(t)\sum_{\ell\ge0}j^\ell t^\ell$ gives
\[
    p_k(j+1)
    =
    \sum_{\ell=0}^k
    j^\ell p_{k-\ell}(j),
    \qquad\text{i.e.,}\qquad
    p_k(j+1)-p_k(j)
    =
    \sum_{\ell=1}^k
    j^\ell p_{k-\ell}(j),
\]
with $p_0(j)=1$ and $p_k(0)=0$ for $k\ge 1$ (since $u_0=1$).
 
We claim that $j\mapsto p_k(j)$ is a polynomial of degree at most $2k$. This is clear for $k=0$. If it holds for all $k'<k$, then each summand $j^\ell p_{k-\ell}(j)$ on the right-hand side above has degree at most $\ell+2(k-\ell)\le 2k-1$, so the forward difference $p_k(j+1)-p_k(j)$ is a polynomial of degree at most $2k-1$; since $p_k(j)=\sum_{i<j}(p_k(i+1)-p_k(i))$ and partial sums of a polynomial of degree at most $2k-1$ form a polynomial of degree at most $2k$, the claim follows.
Consequently, by~\cref{eqn: Bd identity} and~\cref{eqn: binomial cancellation}, all terms with $2k<d$, i.e., with $k<\lceil d/2\rceil$, cancel:
\begin{align*}
    B_d
    &=
    \sum_{j=0}^d
    (-1)^{d-j}\binom{d}{j}
    \sum_{k\ge0}p_k(j)N^{-k} \\
    &=
    \sum_{k\ge0}N^{-k}
    \sum_{j=0}^d
    (-1)^{d-j}\binom{d}{j}p_k(j)\\
    &=
    \sum_{k\ge\lceil d/2\rceil}N^{-k}
    \sum_{j=0}^d
    (-1)^{d-j}\binom{d}{j}p_k(j).
\end{align*}
 
It remains to bound $|p_k(j)|$ for $j\le d$ and $k\ge\lceil d/2\rceil$. For $j\ge 2$, expanding each factor of $u_j(t)=\prod_{\ell=1}^{j-1}(1-\ell t)^{-1}$ as a geometric series shows that
\[
    p_k(j)
    =
    \sum_{\substack{
        s_1,\ldots,s_{j-1}\ge0\\
        s_1+\cdots+s_{j-1}=k
    }}
    \prod_{\ell=1}^{j-1}\ell^{s_\ell}
\]
is a sum of $\binom{j+k-2}{k}\le 2^{j+k-1}$ nonnegative terms, each at most $j^k$; hence $0\le p_k(j)\le j^k2^{j+k-1}$, and this also holds for $j\in\{0,1\}$ since $p_k(0)=p_k(1)=0$ for $k\ge 1$. For $j\le d\le 2k$ we get
\[
    |p_k(j)|\le d^k 2^{j+k-1}\le d^k2^{3k}=(8d)^k.
\]
Therefore, using $\sum_j\binom dj=2^d$ and $8d/N\le 1/2$,
\begin{align*}
    B_d
    \le
    \sum_{k\ge\lceil d/2\rceil}N^{-k}
    \sum_{j=0}^d
    \binom{d}{j}|p_k(j)|
    \le
    2^d
    \sum_{k\ge\lceil d/2\rceil}
    \left(\frac{8d}{N}\right)^k
    \le
    2^{d+1}
    \left(\frac{8d}{N}\right)^{\lceil d/2\rceil},
\end{align*}
which is \cref{eqn: Bdupper}.
\end{proof}

\else
\subsubsection{Higher-degree Components}
Here we provide an outline of the proof of the last two items of \cref{lem: xop_weights}.
We prove the following lemmas in \cref{app: magniweight}.
\begin{lemma}\label{lem: deg d magnitude bound}
For $1\le d\le N$, it holds that
    $M^{=d}[\muP]
    \le
    \left({2d}/{N}\right)^{\lceil d/2\rceil}.$
\end{lemma}
\begin{lemma}\label{lem: deg d weight bound}
For $5\le d\le N/16$, it holds that
    $\sum_{|\gamma|=d}|\hmuP(\gamma)|^2
    \le
    3^{3d}
    \left( N/d\right)^{\lfloor d/2\rfloor}.$
\end{lemma}

Given these lemmas, we can prove an upper bound of the degree-$d$ component of $\xop[r]$ for $5\le d \le d_{\max}$ following the same approach in \cref{eqn: deg4pointwise} as follows
\begin{align*}
    \|\muXoPr^{=d}\|_2^2
    &=
    \sum_{|\gamma|=d}|\hmuP(\gamma)|^{2r}\le
    \left(M^{=d}[\muP]\right)^{2r-2}
    \sum_{|\gamma|=d}|\hmuP(\gamma)|^2
    \\
    &\le
    \left(\frac{2d}{N}\right)^{
        (2r-2)\lceil d/2\rceil
    } \cdot
    3^{3d}
    \left(\frac Nd\right)^{\lfloor d/2\rfloor}.
\end{align*}
For small even $d$, this becomes $3^{3d} 2^{(r-1)d} \left(\frac{d}{N}\right)^{\frac{(2r-3)d}{2}} $, and in particular for $d=6$
\[
O(N^{-(r-1.5)d}) = O(N^{-6r+9})
\]
proving the fifth item. The last item is proven by bounding $\|\muXoPr^{=d}\|_2^2$ by a certain geometric progression based on the above idea.
The full proof of the last item, formally written in \cref{lem: higher degree tail}, is given in \cref{app: higher}.

\fi
\section{Few Quantum Query Security via Planted Collisions}\label{sec:bound2}
This section proves the following lemma.
\begin{lemma}\label{len: small q bound}
If $n\ge 5, 2\le r\le \frac N{16}$, and $q\le \frac N{57774}$, then
\[
\Adv^{\mathsf{dist}}_q(\DXoPr,\dist{\F}) = O\left(\frac{q^3}{N^{r}}\right).
\]
\end{lemma}
Recall that the query-independent $\ell_2$-norm bounds in \cref{lem: xop_weights}: For $d=5$ and $7 \le d \le N/28887$, that lemma give $O(N^{-3r+4})=O(N^{-r})$ upper bounds. 
We give the following alternative query-dependent bounds for $d=2,3,4,6$ in the subsections. 

\begin{lemma}\label{lem:bounds query dependent collisions}
    If $n\ge 5, 2\le r\le \frac N{16}$, and $q\le \frac N{57774}$, then
    $$|\inner{\muXoPr^{=d},P_\Alg}| = O\left(\frac{q^3}{N^r}\right)$$
    for $d=2,3,4,6.$
\end{lemma}

\begin{proof}[Proof of \cref{len: small q bound}]
    Let $\Alg$ be a $q$-query algorithm. As shown in \cref{eqn: advantage bound by degree}, \cref{cor: advantage by inner product} and \cref{eqn: inner product decomposition by degree} gives
    \[
    \Adv^{\mathsf{dist}}(\DXoPr,\dist{\F};\Alg) \le \sum_{d=1}^{2q} \inner{\muXoPr^{=d},P_{\Alg}^{=d}}.
    \]
    By decomposing the above term for $d=2,3,4,6$ and the others, we have the following upper bound
    \[
    \left(\sum_{d=2,3,4,6} \inner{\muXoPr^{=d},P_{\Alg}^{=d}} \right) +
        \sum_{d=5, 7 \le d \le 2q} \|\muXoPr^{=d}\|_2 \|P_{\Alg}^{=d}\|_2
    \]
    where the second term is bounded above by the Cauchy–Schwarz as in \cref{ineq: xop1}
    \[
    \left(\sum_{d=5, 7 \le d \le 2q}  \|\muXoPr^{=d}\|_2^2 \right)^{\frac{1}{2}}
    = O\left(\frac{1}{N^{3r-4}}\right) 
    = O\left(\frac{1}{N^{r}}\right) 
    \]
    where we use the last inequality of \cref{lem: xop_weights} and $r\ge 2$. The remainder terms are bounded by $O(q^3/N^r)$ \cref{lem:bounds query dependent collisions} for $d=2,3,4,6$, we have the overall upper bound $O(q^3/N^r)$.
\end{proof}

\subsection{Preparation: Distributions with planted constraints}\label{subsec: prep}

We will use the indistinguishability of the small-range distributions \cite{zhandry2021construct}.
\begin{definition}[Small-range distributions]\label{def:small_range}
For an integer $R\ge 1$, the distribution $\cS_R$ over $\F=\F_n$ is sampled as follows:
\begin{enumerate}[nosep]
    \item sample a uniformly random function $f:\bit^n\to[R]$ and let $I$ be its image;
    \item sample a uniformly random function $g:I\to\bit^n$;
    \item output $h:=g\circ f$.
\end{enumerate}
\end{definition}
Note that $\mathcal{S}_{\infty}$ is identical to the distributions of random functions.
The following result is proven in \cite{zhandry2021construct}.
\begin{theorem} \label{thm:smallrange}
$\Adv^{\mathsf{dist}}_q(\mathcal{S}_R,\mathcal{S}_{\infty}) \le \frac{27q^3}{R}.$
\end{theorem}

We use the several distributions of functions with some specific constraints. 
First one is the distribution $\D_{\PC}$ of random functions with a planted collision, which is defined as follows:
\begin{description}
    \item[$\D_{\PC}$:] Sample $x_1 \neq x_2$ from $\bit^n$. Then, sample a uniform random $f\gets \F$ such that $f(x_1) = f(x_2)$.
\end{description}
It is not hard to see that its density function is computed by
\begin{align}
    \label{eqn: PCdensity}
\muPC(f)=\binom N2^{-1}\sum_{\{x_1,x_2\}}N\truth{f(x_1)=f(x_2)}.
\end{align}

We use the following observation.
\begin{lemma}\label{lem: decomposing small to pc}
    Let $N\ge 4$ and $R \ge N^2$. 
    For $h\gets \cS_R$ which is defined by $h=g \circ f$, let $C_f:=|\{\{x,x'\}:x\neq x',\,f(x)=f(x')\}|$.
    Then
    \[
    \Pr[C_f=1]\ge \frac{N^2}{8R} , \qquad \Pr[C_f\ge 2]\le\frac{N^4}{8R^2}.
    \]
    Moreover, conditioned on $C_f=0$, the distribution of $h$ is identical to the uniform distribution $\dist{F}$. Conditioned on 
    $C_f=1$, the distribution of $h$ is $\dist{\PC}$.
\end{lemma}
\begin{proof}
    Observe that $\Pr[f(x)=f(x')]=1/R$ for any $x\neq x',$ and any two distinct pairs $(x,x')$ and $(z,z')$ collide simultaneously with probability $1/R^2$, even when they share an input. Thus
\[
\Exp[C_f]=\frac{\binom N2}{R},
\qquad
\Exp\left[\binom{C_f}{2}\right]
=\frac{\binom{\binom N2}{2}}{R^2}
\le \frac{N^4}{8R^2}.
\]
Using $\truth{C\ge 2}\le\binom C2$ and
$\truth{C=1}\ge C-2\binom C2$ for nonnegative integers $C$, we obtain
\[
\Pr[C_f\ge 2]\le \frac{N^4}{8R^2},
\qquad
\Pr[C_f=1]
\ge \frac{N(N-1)}{2R}-\frac{N^4}{4R^2}
\ge \frac{N^2}{8R}.
\]
where the last inequality uses $N\ge 4$ and $R \ge N^2.$

If $C_f=0$, the values of $h$ are independent and uniform because of $g$.
If $C_f=1$, the colliding pair of $f$ is uniform by symmetry, thus 
the conditional distribution is exactly $\D_{\PC}$.
\end{proof}

\begin{lemma}\label{lem: planted_collision}
    $\Adv^{\mathsf{dist}}_q(\D_{\PC},\dist{\F}) =O(q^3/N^2)$.
\end{lemma}
\begin{proof}
    Fix a $q$-query algorithm $\Alg$. Recall $\inner{\dens{D},P_{\mathcal A}} =\Pr_{h\gets\mathcal D}[\Alg^h=1] $ for any oracle distribution $\mathcal D$.

    Let $r \ge N^2$ and choose $h \gets\cS_R$.
    Let $p_0,p_1,p_2$ be the probabilities of $C_f=0$, $C_f=1$, and $C_f\ge 2$, respectively, as defined in \cref{lem: decomposing small to pc}.
    By the same lemma, there exists a distribution $\mathcal R_R$ such that, writing
$\mu_R$ and $\rho_R$ for the densities of $\cS_R$ and
$\mathcal R_R$, respectively,
\begin{align}
    \label{eqn:smallrange decomposition}
\mu_R=p_0\cdot \constone+p_1\cdot \muPC+p_2\cdot \rho_R,
\qquad
p_1\ge \frac{N^2}{8R},
\qquad
p_2\le \frac{N^4}{8R^2}
\end{align}
where $\constone$ denotes the density function of $\D_F$.
Since $p_0+p_1+p_2=1$, we have
\[
\mu_R-\constone=p_1(\muPC-\constone)+p_2(\rho_r-\constone).
\]
Taking inner products with $P_{\mathcal A}$ and applying the
triangle inequality gives
\[
p_1\left|\inner{\muPC-\constone,P_{\mathcal A}}\right|
\le
\left|\inner{\mu_R-\constone,P_{\mathcal A}}\right|
+p_2\left|\inner{\rho_R-\constone,P_{\mathcal A}}\right| \le \frac{27q^3}{R} + p_2
\]
where we use \cref{thm:smallrange} for the first term and $\left|\inner{\rho_R-\constone,P_{\mathcal A}}\right|\le 1$ for the second term.
This proves 
\[
\left|\inner{\muPC-\constone,P_{\mathcal A}}\right|
\le
\frac{8R}{N^2}
\left(\frac{27q^3}{R}+\frac{N^4}{8R^2}\right)
=
\frac{216q^3}{N^2}+\frac{N^2}{R}.
\]
Taking $R$ sufficiently large, we have the same bound $O(q^3/N^2)$ for all $q$-query algorithms $\Alg$. This completes the proof.
\end{proof}

\subsection{Degree-2 component as planted collision}
We first prove \cref{lem:bounds query dependent collisions} for $d=2$.

\begin{lemma}\label{lem: xopdeg2_smallq}
    $|\inner{\muXoPr^{=2},P_\Alg}| = O\left(\frac{q^3}{N^r}\right).$
\end{lemma}
\begin{proof}
We start from the degree-2 component represented in \cref{eqn: XoPdeg2}.
The equation shows that it suffices to consider the degree-2 $\gamma\in \hF$ with 
the support $\{x_1<x_2\}$ satisfying $\gamma(x_1)=\gamma(x_2)$. Letting $\gamma(x_1)=\gamma(x_2)=y$ and plugging $f$,
\begin{align*}
    \muXoPr^{=2}&(f)
    =\sum_{\supp(\gamma)=\{x_1<x_2\}} \frac{(-1)^r\truth{\gamma(x_1)=\gamma(x_2)} \chi_\gamma(f)}{(N-1)^r}\\
    &
    =\sum_{x_1< x_2 ; y\neq 0^n} \frac{(-1)^r}{(N-1)^r} \chi_y (f(x_1)+f(x_2)) 
    \\
    &
    =\sum_{x_1 < x_2} \frac{(-1)^r}{(N-1)^r} (N \cdot \truth{f(x_1)=f(x_2)} -\constone)
    =\frac{(-1)^r\cdot N}{2(N-1)^{r-1}} \left(\muPC(f)  - \constone\right) 
\end{align*}
where we use \cref{eqn: characsum} in the third equality, and use \cref{eqn: PCdensity} in the last equality.
Therefore, 
\[
|\inner{\muXoPr^{=2},P_\Alg}| = \frac{ N}{2(N-1)^{r-1}}  \left|\inner{\muPC,P_\Alg} - \inner{\constone,P_\Alg} \right|.
\]
The right hand side is the distinguishing advantage between the planted collision function and random functions up to the multiplicative factor $\frac{ N}{2(N-1)^{r-1}}.$
Plugging \cref{lem: planted_collision} here gives the desired upper bound.
\end{proof}

\subsection{Degree-3 component as planted-three-collision}
To bound the degree-3 component, we need to define $\DthPC$ of random functions with a planted 3-collision, which is defined as follows:
\begin{description}
    \item[$\DthPC$:] Sample three distinct $x_1,x_2,x_3$ from $\bit^n$. Then, sample a uniform random $f\gets \F$ such that $f(x_1) = f(x_2)= f(x_3)$. 
\end{description}
Its density function is
    \[
        \muthPC(f)=\frac{1}{(N)_{3}}\sum_{(x_1,x_2,x_3)\in(\bit^n)^{*3}}N^2\,\truth{f(x_1)=f(x_2)=f(x_3)}.
    \]
We prove the following lemma in \cref{app: collision}.
\begin{lemma}\label{lem: muthpc_smallq}
        $\Adv^{\mathsf{dist}}_q(\DthPC,\dist{\F}) =O(q^2/N)$.
\end{lemma}

We actually prove the following stronger bound than $O(q^3/N^r)$.
\begin{lemma}\label{lem: xopdeg3_smallq}
    $|\inner{\muXoPr^{=3},P_\Alg}| = O\left(\frac{q^2}{N^r}\right).$
\end{lemma}
\begin{proof}
For $\supp(\gamma)=\{x_1<x_2<x_3\}$, write $y_i=\gamma(x_i)$. 
Recall
\[
\muXoPr(\gamma) = \hmuP(\gamma)^r = \frac{
    \truth{y_1+y_2+y_3=0^n}
    }{\binom{N-1}{2}^r}.
\]
Expanding the degree-3 component with plugging $f$ becomes
\begin{align*}
    \binom{N-1}{2}^r\muXoP^{=3}(f)
    &=\sum_{\supp(\gamma)=\{x_1<x_2<x_3\}} \truth{y_1+y_2+y_3=0^n}\chi_\gamma(f)\\
    &=\sum_{\substack{x_1<x_2<x_3\\
        y_1\neq0^n,\ y_2\neq0^n,\ y_1\neq y_2}} \chi_{y_1}(f(x_1))\chi_{y_2}(f(x_2))\chi_{y_1+y_2}(f(x_3))
    \\
    &=\sum_{\substack{x_1<x_2<x_3\\
        y_1\neq0^n,\ y_2\neq0^n,\ y_1\neq y_2}} \chi_{y_1}(f(x_1)+f(x_3))\chi_{y_2}(f(x_2)+f(x_3)).
\end{align*}
For each $x_1,x_2,x_3$, observe that
\begin{align*}
    &\sum_{y_1\neq0^n,\ y_2\neq0^n,\ y_1\neq y_2} \chi_{y_1}(f(x_1)+f(x_3))\chi_{y_2}(f(x_2)+f(x_3))\\
    &=\sum_{y_1\neq0^n,\ y_2\neq0^n} \chi_{y_1}(f(x_1)+f(x_3))\chi_{y_2}(f(x_2)+f(x_3))
      -\sum_{y\neq0} \chi_y(f(x_1)+f(x_2))\\
    &=
        (N\truth{f(x_1)=f(x_3)}-\constone)
        (N\truth{f(x_2)=f(x_3)}-\constone)-(N\truth{f(x_1)=f(x_2)}-\constone)
\end{align*}
where we use \cref{eqn: characsum} in the last equality. Noting that $\truth{f(x_1)=f(x_3)}\truth{f(x_2)=f(x_3)} = \truth{f(x_1)=f(x_2)=f(x_3)}$, the overall sum becomes
\begin{align*}
    &\sum_{x_1<x_2<x_3}
      \Big(
        N^2\truth{f(x_1)=f(x_2)=f(x_3)}
    \\
    &\quad
        -N\truth{f(x_1)=f(x_2)}
        -N\truth{f(x_1)=f(x_3)}
        -N\truth{f(x_2)=f(x_3)}+2\cdot \constone
      \Big)
    \\
    &=\binom{N}{3}\bigl((\muthPC(f)-\constone)-3(\muPC(f)-\constone)\bigr).
\end{align*}

Hence, $|\inner{\muXoPr^{=3},P_\Alg}| $ can be decomposed into a linear sum of the distinguishing advantage of the planted 3-collision and the planted collision functions from random functions up to the multiplicative factor $\binom{N}{3} / \binom{N-1}{2}^r$. This gives
\[
|\inner{\muXoPr^{=3},P_\Alg}| \le 
O\left(\frac{2^r}{N^{2r-3}}\right) \cdot \left(O\left(\frac{q^2}{N} \right) + O\left(\frac{q^3}{N^2} \right)\right) = O\left(\frac{q^2}{N^r}\right)
\]
where we omit the factor regarding $r$ and use $q\le N$ and \cref{lem: planted_collision,lem: muthpc_smallq}.
\end{proof}

\subsection{Degree-4 component as planted two collisions plus small terms}
We define the distributions of random functions with multiple planted collisions to bound the degree-4 component.
For $k \in \mathbb N$, they are defined as follows.
\begin{description}
    \item[$\D_{\PC, k}$:] Sample a tuple $(x_1,\dots,x_{2k})\gets(\bit^n)^{*(2k)}$ of distinct inputs. Then sample a uniform random $f\gets\F$ such that $f(x_{2j-1})=f(x_{2j})$ for all $j\in[k]$. 
    \item[$\DfoX$:] Sample four distinct $x_1,\ldots,x_4$ from $\bit^n$. Then, sample a uniform random $f\gets\F$ such that $f(x_1)+f(x_2)+f(x_3)+f(x_4)=0^n$. 
\end{description} 
The density functions of these distributions can be computed as follows
\begin{align*}
    \mu_{\PC, k}(f)&=\frac{1}{(N)_{2k}}\sum_{(x_1,\dots,x_{2k})\in(\bit^n)^{*(2k)}}N^k\prod_{j=1}^k\truth{f(x_{2j-1})=f(x_{2j})},\\
    \mufoX(f)&=\frac{1}{(N)_{4}}\sum_{(x_1,\dots,x_4)\in(\bit^n)^{*4}}N\,\truth{f(x_1)+f(x_2)+f(x_3)+f(x_4)=0^n}.
\end{align*}

Note that $\DPC=\D_{\PC, 1}$
and $\D_{\PC, 0}=\dist{\F}$. 
Similarly to the previous distributions, 
we prove the following two lemmas in \cref{app: collision}.
\begin{lemma}\label{lem: multi_planted}
    $\Adv^{\mathsf{dist}}_q(\D_{\PC,k},\dist{\F}) =O(q^3/N^2)$ for $k=2,3.$
\end{lemma}
\begin{lemma}\label{lem: mufox_smallq}
    $\Adv^{\mathsf{dist}}_q(\DfoX,\dist{\F}) =O(q^3/N^2)$.
\end{lemma}

We prove the following lemma in this subsection using the above lemmas.
\begin{lemma}\label{lem: xopdeg4_smallq}
    $|\inner{\muXoPr^{=4},P_\Alg}| = O\left(\frac{q^3}{N^r}\right).$
\end{lemma}
\begin{proof}
Let $\gamma$ be such that $|\gamma|=4$.
For $\supp(\gamma)=\{x_1<x_2<x_3<x_4\}$, write $y_i=\gamma(x_i)$. We begin with the degree-4 Fourier component from \cref{eqn:deg4permutation}
\begin{align*}
    \hmuP(\gamma) = \frac{NC_{2,2}(\gamma) - 6C_4(\gamma)}{(N-1)_3}
\end{align*}
where we recall the following for convenience.
\begin{align*}
    C_{2,2}(\gamma)&=\truth{y_1=y_2}\truth{y_3=y_4}+\truth{y_1=y_3}\truth{y_2=y_4}+\truth{y_1=y_4}\truth{y_2=y_3},\\
    C_4(\gamma)&=\truth{y_1+y_2+y_3+y_4=0^n}.
\end{align*}
We write $X,Y,Z$ to denote $\truth{y_1=y_2}\truth{y_3=y_4}$, $\truth{y_1=y_3}\truth{y_2=y_4}$, and $\truth{y_1=y_4}\truth{y_2=y_3}$, respectively, so that $C_{2,2}(\gamma) = X+Y+Z$. 
Observe the following.
\begin{enumerate}
    \item $X^2=X,Y^2=Y,Z^2=Z$ and $C_{4}^2=C_4$ as they are in $\bit$,
    \item $XY=YZ=ZX = \truth{y_1=y_2=y_3=y_4}=:E(\gamma)$,
    \item $C_{2,2}(\gamma) \ge 1 $ implies $C_4(\gamma)=1$, thus $C_{2,2}(\gamma)C_4(\gamma)=C_{2,2}(\gamma)$,
    \item similarly $C_{2,2} E = 3E$ and $C_{4}E = E$.
\end{enumerate}
From this, we have $C_{2,2}^k = C_{2,2} + (3^k-3)E$ inductively for $k \in \mathbb N$.
Further calculation based on the above gives the following identity
\begin{align*}
    \hmuP(\gamma)^r 
    &=
    \left(
    \frac{NC_{2,2}(\gamma) - 6C_4(\gamma)}{(N-1)_3}
    \right)^r
    =\frac{a_r C_{2,2}(\gamma)+ b_{r} E(\gamma) + c_r C_4(\gamma)}{(N-1)_3^r}
\end{align*}
where $
    a_r = (N-6)^r - (-6)^r,
    b_r = (3N-6)^r -3(N-6)^r+2 (-6)^r,
    c_r = (-6)^r.$

Now we compute $\muXoPr^{=4}(f) $ .
Observe that $C_{2,2}=X+Y+Z$ terms can be computed as follows, where $8=4!/3$ comes from the order of $x_1,x_2,x_3,x_4$ (4!) and the symmetry of $X,Y,Z$ (3),
\begin{align*}
    &\sum_{|\gamma|=4} C_{2,2}(\gamma) \chi_\gamma (f) = \frac{1}{8} \sum_{\supp(\gamma)=\{x_1,x_2,x_3,x_4\}}\truth{y_1=y_2}\truth{y_3=y_4} \chi_{\gamma}(f)\\
    &=\frac{1}{8}\sum_{\substack{\supp(\gamma)=\{x_1,x_2,x_3,x_4\}\\ y_1\neq0^n,\ y_3\neq0^n}} \chi_{y_1}(f(x_1)+f(x_2)) \chi_{y_3}(f(x_3)+f(x_4))\\
    &=\frac{1}{8}\sum_{\substack{x_1,x_2,x_3,x_4 \\ \text{distinct}}} (N\truth{f(x_1)=f(x_2)} - \constone)(N\truth{f(x_3)=f(x_4)} - \constone)\\
    &=\frac{(N)_4}{8}((\muPCt(f) - \constone) - 2(\muPC(f) - \constone)).
\end{align*}
Here, the last term can be interpreted as a linear sum of distinguishing advantage between random functions with planted collision(s) and uniformly random functions. \cref{lem: multi_planted,lem: planted_collision} gives the bound $\Big|\inner{\sum_{|\gamma|=4}C_{2,2}(\gamma)\chi_\gamma,P_\Alg}\Big| = O(q^3N^2)$. 

For the $E$ term, we have
\begin{align*}
    &\sum_{|\gamma|=4}E(\gamma)\chi_\gamma(f) = \sum_{\supp(\gamma)=\{x_1,x_2,x_3,x_4\}}\truth{y_1=y_2=y_3=y_4} \chi_{\gamma}(f)\\
    &= \sum_{\substack{\supp(\gamma)=\{x_1,x_2,x_3,x_4\} \\ y\ne0^n}}\chi_y(f(x_1)+f(x_2)+f(x_3)+f(x_4))
    =\frac{(N)_4}{24}(\mufoX(f) - 1),
\end{align*}
which gives $\Big|\inner{\sum_{|\gamma|=4}E(\gamma)\chi_\gamma,P_\Alg}\Big| = O(q^3N^2)$ because of \cref{lem: mufox_smallq}.

For the last term with $C_4$,  we can directly bound
\[
\Big|\inner{\sum_{|\gamma|=4}C_4(\gamma)\chi_\gamma,P_\Alg}\Big| \leq \norm{\sum_{|\gamma|=4}C_4(\gamma)\chi_\gamma}_2 = O(N^4).
\]

Plugging three bounds on the expansion of $\inner{\muXoPr^{=4},P_\Alg}$ gives
\begin{align*}
    |\inner{\muXoPr^{=4},P_\Alg}| = O\left( 
        \frac{N^r \cdot q^3 N^2 + N^r \cdot q^3 N^2 + N^4}{N^{3r}}
    \right) = O\left(
    \frac{q^3}{N^{2r-2}}
    \right)
\end{align*}
which is $O(q^3/N^r)$ because $r\ge2$,
as desired.
\end{proof}

\subsection{Degree-6 component as planted three collisions plus small terms}
\ifnum\fullpage=0
For the degree 6, it holds that $|\inner{\muXoPr^{=6},P_\Alg}| = O\left({q^3}/{N^r}\right).$
The proof of this inequality is rather complicated, and can be found in \cref{app: xopdeg6_smallq}.

\else
We use the following lemma.
\begin{lemma}\label{lem: binom_inequality}
    $\mathrm{For\ } x, y \in \mathbb{R}, |(x+y)^r-x^r|\leq r|y|(|x|+|y|)^{r-1}.$
\end{lemma}
\begin{proof}
Noting that $\binom{r}{k}\le k\binom{r}{k}=r\binom{r-1}{k-1}$ for $1\le k\le r$, we have
\[
|(x+y)^r-x^r| \le \sum_{k=1}^r  \binom{r}{k} |x^{r-k}y^k| \le r|y| \sum_{k=1}^r \binom{r-1}{k-1} |x|^{r-k}|y|^{k-1}
\]
which equals to $r|y|(|x|+|y|)^{r-1}.$
\end{proof}

\begin{lemma}
    $|\inner{\muXoPr^{=6},P_\Alg}| = O\left(\frac{q^3}{N^r}\right).$
\end{lemma}
\begin{proof}
We need to compute $\hmuP$ first. Let us define $Z_{2, 2, 2}, Z_{2, 4}, Z_{3, 3}, Z_6$ by
\begin{align*}
  Z_{2,2,2}
  &:=\sum_{M\ \mathrm{perfect\ matching\ of\ }[6]}\prod_{\{i,j\}\in M}\truth{z_i=z_j}, \\
  Z_{2,4}
  &:=\sum_{\substack{\{i, j\} \subset [6] \\ \{i, j\} \cap \{g, h, k, \ell\}=\varnothing}}\truth{z_i=z_j}\truth{z_g=z_h=z_k=z_\ell}, \\
  Z_{3,3}
  &:=\sum_{\substack{1 \in \{i, j, k\} \subset [6] \\ \{i, j, k\} \cap \{g, h, \ell\}=\varnothing}}\truth{z_i=z_j=z_k}\truth{z_g=z_h=z_\ell}, \\
  Z_6
  &:=\truth{z_1=z_2=z_3=z_4=z_5=z_6}. \\
\end{align*}
By the inclusion-exclusion principle for the fifteen events $E_{ij}=[z_i=z_j]$ gives,
\begin{align*}
    \truth{z_1, \cdots, z_6 \mathrm{\ distinct}} = 1 - \cdots - Z_{2, 2, 2} + 6Z_{2, 4} + 4Z_{3, 3} - 120Z_6,
\end{align*}
where the omitted terms correspond to partitions containing at least one
singleton block, which will be vanished.

Similarly, we define $C_{2,2,2}, C_{2,4}, C_{3,3}, C_6$ by
\begin{align*}
  C_{2,2,2}(\gamma)
  &:=\sum_{M\ \mathrm{perfect\ matching\ of\ }[6]}\prod_{\{i,j\}\in M}\truth{y_i=y_j}, \\
  C_{2,4}(\gamma)
  &:=\sum_{\substack{\{i, j\} \subset [6] \\ \{i, j\} \cap \{g, h, k, \ell\}=\varnothing}}\truth{y_i=y_j}\truth{y_g+y_h+y_k+y_\ell=0^n}, \\
  C_{3,3}(\gamma)
  &:=\sum_{\substack{1 \in \{i, j, k\} \subset [6] \\ \{i, j, k\} \cap \{g, h, \ell\}=\varnothing}}\truth{y_i+y_j+y_k=0^n}\truth{y_g+y_h+y_\ell=0^n}, \\
  C_6(\gamma)
  &:=\truth{y_1+y_2+y_3+y_4+y_5+y_6=0^n}. \\
\end{align*}
By applying \cref{eqn: characsum} and simplifying, we obtain
\begin{align*}
    \hmuP(\gamma) = \frac{-N^2C_{2,2,2}(\gamma) + 6NC_{2,4}(\gamma)+4NC_{3,3}(\gamma)-120C_6(\gamma)}{(N-1)_5}.
\end{align*}
set $R(\gamma) = 6NC_{2,4}(\gamma)+4NC_{3,3}(\gamma)-120C_6(\gamma)$. Note that $R(\gamma) = O(N)$. We rewrite
\begin{align*}
    \left(\hmuP(\gamma) (N-1)_5 \chi_\gamma\right)^r 
    &=\big(-N^2C_{2,2,2}(\gamma) + R(\gamma) \big)^r\chi_\gamma \\
    &=S(\gamma)\chi_\gamma + (-N^2)^rC_{2,2,2}(\gamma)\chi_\gamma,
\end{align*}
where $S(\gamma) = \big(-N^2C_{2,2,2}(\gamma) + R(\gamma) \big)^r - (-N^2)^rC_{2,2,2}(\gamma)$. For the second term, we have
\begin{align*}
    &\sum_{\supp(\gamma)=\{x_1<\ldots<x_6\}}C_{2,2,2}(\gamma)\chi_\gamma(f)\\
    &=\frac{15}{6!}\sum_{\supp(\gamma)=\{x_1,\ldots,x_6\}}\truth{y_1=y_2}\truth{y_3=y_4}\truth{y_5=y_6}(\gamma)\chi_\gamma(f)\\
    &=\frac{1}{48}\sum_{\substack{x_1,\ldots,x_6 \\ \text{distinct}}}\prod_{s=1}^3 (N\truth{f(x_{2s-1})=f(x_{2s})}-\constone) \\
    &=\frac{(N)_6}{48}(\mu_{\PC, 3} - 3\mu_{\PC, 2} + 3\muPC - \constone) \\
    &=\frac{(N)_6}{48}((\mu_{\PC, 3}-\constone) - 3(\mu_{\PC, 2}-\constone) + 3(\muPC-\constone)).
\end{align*}
This will be routinely bounded using \cref{lem: planted_collision,lem: multi_planted} below.

Now let us bound $S(\gamma)$. Partition the degree-six characters into
\begin{align*}
    \Gamma_0&:=\{\gamma:C_{2,2,2}(\gamma)=0\},\\
    \Gamma_1&:=\{\gamma:C_{2,2,2}(\gamma)=1\},\\
    \Gamma_{\geq2}&:=\{\gamma:C_{2,2,2}(\gamma)\geq2\}.
\end{align*}
For $J\in\{0,1,{\geq2}\}$, let $B_J:=\sum_{\gamma\in\Gamma_J}S(\gamma)\chi_\gamma$. Then we have
\begin{align*}
    \Big|\inner{\sum_{|\gamma|=6}S(\gamma)\chi_\gamma,P_\Alg}\Big| \leq |\inner{B_0,P_\Alg}| + |\inner{B_1,P_\Alg}| + |\inner{B_{\geq 2},P_\Alg}|
\end{align*}
and using Cauchy-Schwarz inequality
\begin{align*}
    |\inner{B_J,P_\Alg}| \leq \|B_J\|_2\|P_\Alg\|_2 \leq \|B_J\|_2 \le \left(\sum_{\gamma\in\Gamma_J}|S(\gamma)|^2\right)^{1/2}.
\end{align*}
If $\gamma\in\Gamma_0$, then $S(\gamma)=R(\gamma)^r=O(N^r)$.
\begin{align*}
    |\inner{B_0,P_\Alg}|
    &\leq \left( \sum_{\substack{|\gamma|=6 \\ C_{2,2,2}(\gamma)=0}} |S(\gamma)|^2 \right)^{1/2} \\
    &\leq \left( O(N^{12})O(N^{2r}) \right)^{1/2} = O(N^{r+6}).
\end{align*}
If $\gamma\in\Gamma_1$, we apply \cref{lem: binom_inequality}, which gives
\begin{align*}
    |S(\gamma)| &= \big|\big(-N^2C_{2,2,2}(\gamma) + R(\gamma) \big)^r - (-N^2)^rC_{2,2,2}(\gamma)\big| \\
    &= \big|\big(-N^2 + R(\gamma) \big)^r - (-N^2)^r\big| = O(N^{2r-1}).
\end{align*}
Since $C_{2,2,2}(\gamma)=1$, we have at most $O(N^3)$ choices of labels per support. Hence,
\begin{align*}
    |\inner{B_1,P_\Alg}|
    &\leq \left( \sum_{\substack{|\gamma|=6 \\ C_{2,2,2}(\gamma)=1}} |S(\gamma)|^2 \right)^{1/2} \\
    &\leq \left( O(N^{3+6})O(N^{(2r-1)\cdot 2}) \right)^{1/2} = O(N^{2r+3.5}).
\end{align*}
If $\gamma\in\Gamma_{\geq2}$, then
\begin{align*}
    |S(\gamma)| &= \big|\big(-N^2C_{2,2,2}(\gamma) + R(\gamma) \big)^r - (-N^2)^rC_{2,2,2}(\gamma)\big| \\
    &= \big|O(N^2)^r - (-N^2)^rO(1)\big| = O(N^{2r}).
\end{align*}
Since $C_{2,2,2}(\gamma)\ge2$, there exist indices $i, j$ and $g, h, k, l$ such that $y_i=y_j$ and $y_g=y_h=y_k=y_l$. Hence, we have at most $O(N^2)$ choices of labels per support. Hence,
\begin{align*}
    |\inner{B_{\geq 2},P_\Alg}|
    &\leq \left( \sum_{\substack{|\gamma|=6 \\ C_{2,2,2}(\gamma)\geq 2}} |S(\gamma)|^2 \right)^{1/2} \\
    &\leq \left( O(N^{2+6})O(N^{2r\cdot 2}) \right)^{1/2} = O(N^{2r+4})
\end{align*}
Finally,
\begin{align*}
    |\inner{\muXoPr^{=6},P_\mathcal{A}}|
    &\leq \frac{N^{2r}\Big|\inner{\sum_{|\gamma|=6}C_{2,2,2}(\gamma)\chi_\gamma,P_\mathcal{A}}\Big| + \sum_{J\in\{0,1,\ge 2\}}|\inner{B_J,P_\mathcal{A}}|}{(N-1)_5^r} \\
    &\leq \frac{O(q^3N^{2r+4}) + O(N^{r+6}) + O(N^{2r+3.5}) + O(N^{2r+4})}{\Theta(N^{5r})} \\
    &= O\Big(\frac{q^3}{N^{3r-4}}\Big) + O\Big(\frac{1}{N^{4r-6}}\Big) + O\Big(\frac{1}{N^{3r-3.5}}\Big) + O\Big(\frac{1}{N^{3r-4}}\Big)
    \\&= O\left(\frac{q^3}{N^r}\right).
\end{align*}
This concludes the proof.
\end{proof}

\fi
\section{Middle Quantum Query Security via Improved Small-range Indistinguishability}\label{sec: compressed oracle}
In this section, we prove the following theorem.
As in \cref{thm: main_bound}, the asymptotic notation may hide a constant depending on $r$.

\begin{lemma}\label{thm: improved main_bound}
If $n\ge 5, r\le \frac N{16}$, and $q\le \frac N{57774}$, then
\[
    \Adv_q^{\mathsf{dist}}(\DXoPr,\dist{\F})
    =O\left(\frac{q^{3/2}}{N^{r-1/2}}\right).
\]
\end{lemma}

The main result of this section is the following bound for the planted collision pseudorandomness. Compared to $q^3/N^2$ bound in \cref{lem: planted_collision}, the following lemma gives a better bound for $q^3 \ge N$.
\begin{lemma}\label{thm: improved planted_collision}
Let $N=2^n\ge2$. For every integer $q\ge0$, it holds that $\Adv_q^{\mathsf{dist}}(\DPC,\dist{\F})
    =O\left({q^{3/2}}/{N^{3/2}}\right).$
\end{lemma}

This lemma can be used to give an alternative bound for the small-range indistinguishability \cite{zhandry2021construct} for the large-range regime.
This improves the bound in \cref{thm:smallrange}
when $q^3\gg N$ and $R$ is sufficiently large, or more precisely for $q^3 \gg N$ and $R\gg N^4/q^3$.

\begin{corollary}\label{cor: improved small_range}
Let $N=2^n\ge4$, $q\ge1$, and $R\ge N^2$ be integers. Then
\[
    \Adv_q^{\mathsf{dist}}(\cS_R,\cS_\infty)
    =O\left(\frac{q^{3/2}\sqrt N}{R}+\frac{N^4}{R^2}\right).
\]
In particular, if $R\ge N^{7/2}/q^{3/2}$, then
$
    \Adv_q^{\mathsf{dist}}(\cS_R,\cS_\infty)
    =O\left(\frac{q^{3/2}\sqrt N}{R}\right).$
\end{corollary}
\begin{proof}
We use the fact from the proof of \cref{lem: planted_collision} that the acceptance probability of the small-range distribution can be decomposed into the random functions and the random functions with a planted collision as in \cref{eqn:smallrange decomposition}.
Using \cref{lem: decomposing small to pc}, the coefficient $\Pr[C_f=1]$ of the planted collision part satisfies
$\Pr[C_f=1]\le\Exp[C_f]=\binom N2/R$.
By \cref{thm: improved planted_collision}, we have
\[
    \Adv_q^{\mathsf{dist}}(\cS_R,\cS_\infty)
    \le
    \frac{\binom N2}{R}
    \Adv_q^{\mathsf{dist}}(\DPC,\dist{\F})
    +\frac{N^4}{8R^2}
    =
    O\left(\frac{q^{3/2}\sqrt N}{R}+\frac{N^4}{R^2}\right).
\]
The last claim follows from
$N^4/R^2\le q^{3/2}\sqrt N/R$ when $R\ge N^{7/2}/q^{3/2}$.
\end{proof}

We prove the main result. The remainder of this section is devoted to prove the above lemma.

\begin{proof}[Proof of \cref{thm: improved main_bound}]
If $q\le N^{1/3}$, then \cref{len: small q bound} proves the desired result.
We therefore assume $q\ge N^{1/3}$.
Fix a $q$-query algorithm $\Alg$.
Recall from the proof of \cref{lem: xopdeg2_smallq} that
\[
    \muXoPr^{=2}
    =\frac{(-1)^rN}{2(N-1)^{r-1}}(\muPC-\constone).
\]
Applying \cref{thm: improved planted_collision}, we have
\[
    \left|\inner{\muXoPr^{=2},P_\Alg}\right|
    \le\frac{N}{2(N-1)^{r-1}}
        \Adv_q^{\mathsf{dist}}(\DPC,\dist{\F})
    =O\left(\frac{q^{3/2}}{N^{r-1/2}}\right).
\]
For the remaining components, \cref{lem: xop_weights} with $d_{\max}=2q$ gives
\begin{align*}
    \sum_{d=3}^{2q}\|\muXoPr^{=d}\|_2^2
    &=
    O\left(
        \frac1{N^{4r-5}}+\frac1{N^{4r-6}}
        +\frac1{N^{6r-9}}+\frac1{N^{6r-8}}
    \right)
    =O\left(\frac1{N^{4r-6}}\right).
\end{align*}
The conditions of that lemma follow from $q\ge N^{1/3}$, $q\le N/57774$, and $N\ge16r$.
Since $\|P_\Alg\|_2\le1$, the Cauchy--Schwarz inequality bounds the contribution of these components by
\[
    \left|\sum_{d=3}^{2q}\inner{\muXoPr^{=d},P_\Alg^{=d}}\right|
    \le
    \left(\sum_{d=3}^{2q}\|\muXoPr^{=d}\|_2^2\right)^{1/2}
    =O\left(\frac1{N^{2r-3}}\right).
\]
Finally, $r\ge2$ and $q\ge N^{1/3}$ imply
\[
    \frac1{N^{2r-3}}\le\frac1{N^{r-1}}
    \le\frac{q^{3/2}}{N^{r-1/2}}.
\]
Combining the above bounds with $\muXoPr^{=1}=\constzero$ and \cref{cor: advantage by inner product} proves the theorem.
\end{proof}

\subsection{Planted collisions in the compressed oracle}\label{subsec: pc compressed}

Fix a $q$-query algorithm $\Alg$.
By \cref{eqn: PCdensity}, we immediately have
\begin{align}\label{eqn: pc density centered}
    \inner{\muPC-\constone,P_\Alg}
    &=
    \frac{2}{N-1}\,
    \Exp_{f\gets\F}\left[
        P_\Alg(f)
        \sum_{x<y}\left(\truth{f(x)=f(y)}-\frac1N\right)
    \right].
\end{align}
We express the expectation on the right-hand side using the compressed oracle \cite{C:Zhandry19}. We follow the formalization of~\cite{EC:CFHL21}.

\paragraph{The compressed oracle.}
We briefly explain the compressed oracle for random functions, which is perfectly indistinguishable from a quantum random oracle.
The database consists of one register for each input $x$, with basis
$\{\ket{\bot}\}\cup\{\ket v:v\in\bit^n\}$.
Write $\ket{+}=N^{-1/2}\sum_{v\in\bit^n}\ket v$, and let the compression operator $\mathsf{Comp}$ be the unitary that swaps $\ket{\bot}$ and $\ket{+}$ while fixing their orthogonal complement.
For the compressed representation of $v$, we use the notation
\[
    \ket{\widehat v}:=\mathsf{Comp}\ket v
    =\ket v-\frac{\ket{+}}{\sqrt N}+\frac{\ket{\bot}}{\sqrt N}.
\]
The vectors $\{\ket{\widehat v}\}_v$ are orthonormal.
Their span is the orthogonal complement of $\ket{+}$ in the database register.
Throughout the execution, each database register remains supported on this span.
For $f\in\F$, write
$\ket{\widehat f}:=\bigotimes_x\ket{\widehat{f(x)}}$.
The empty database and a compressed query $O$ satisfy
\[
    \ket{\emptyset}=\ket{\bot}^{\otimes N}
    =N^{-N/2}\sum_{f\in\F}\ket{\widehat f},
    \qquad
    O\ket{x,y}\ket{\widehat f}
    =\ket{x,y+f(x)}\ket{\widehat f}.
\]
A compressed query applies $\mathsf{Comp}$ to the database register indexed by the query input, performs the usual quantum random oracle query, and applies $\mathsf{Comp}$ again.
In particular, $O^\dagger=O$, and a query on input $x$ acts only on database register $x$.

Define the collision operators
\begin{align}\label{eqn: pc collision operators}
    \mathsf{Col}_{xy}
    &:=
    \sum_v(\ketbra{\widehat v})_x\otimes(\ketbra{\widehat v})_y,
    &
    \Delta\mathsf{Col}
    &:=
    \sum_{x<y}\left(\mathsf{Col}_{xy}-\frac1N I\right).
\end{align}
They satisfy the following equations; see the similarity between the second equation and \cref{eqn: pc density centered}:
\[
    \mathsf{Col}_{xy}\ket{\widehat f}
    =\truth{f(x)=f(y)}\ket{\widehat f},
    \qquad
    \Delta\mathsf{Col}\ket{\widehat f}
    =
    \sum_{x<y}\left(\truth{f(x)=f(y)}-\frac1N\right)\ket{\widehat f}.
\]

Let $\ket{\psi_q^f}$ be the final state of the purified algorithm $\Alg$ with oracle $f$.
Let $\Pi_{\mathcal A}=\ketbra1_Z\otimes I_{\mathrm{rest}}$ be its acceptance projector on the output qubit $Z$, so
$P_\Alg(f)=\bra{\psi_q^f}\Pi_{\mathcal A}\ket{\psi_q^f}$.
On the joint algorithm--database space, write
$\Pi=\Pi_{\mathcal A}\otimes I_{\mathrm{db}}$.
The final state of the compressed execution is
\[
    \ket{\Psi_q}
    =
    N^{-N/2}\sum_{f\in\F}\ket{\psi_q^f}\ket{\widehat f}.
\]
Expanding its acceptance expectation and using the orthogonality of $\{\ket{\widehat f}\}_f$, we obtain
\begin{align}\label{eqn: pc compressed expectation}
    \bra{\Psi_q}\Pi\Delta\mathsf{Col}\ket{\Psi_q}
    &=
    \frac1{N^N}\sum_{f,g\in\F}
        \bra{\psi_q^g}\Pi_{\mathcal A}\ket{\psi_q^f}
        \bra{\widehat g}\Delta\mathsf{Col}\ket{\widehat f}
        \notag\\
    &=
    \Exp_{f\gets\F}\left[
        P_\Alg(f)
        \sum_{x<y}\left(\truth{f(x)=f(y)}-\frac1N\right)
    \right]
        \notag\\
    &=\frac{N-1}{2}\inner{\muPC-\constone,P_\Alg}.
\end{align}
Thus, it suffices to bound the left-hand side by $O(q^{3/2}/\sqrt N)$ to prove \cref{thm: improved planted_collision}.

\subsection{The number of collisions in the database}\label{subsec: pc recorded}


For a database $D\in(\bit^n\cup\{\bot\})^{\bit^n}$, define the two diagonal operators
\begin{align*}
    \mathsf D_{\mathrm{cnt}}\ket D
    &:=
    |\{x:D(x)\ne\bot\}|\ket D,\\
    \mathsf{Col}_{\mathrm{cnt}}\ket D
    &:=
    \sum_{x<y}\truth{D(x)=D(y)\ne\bot}\ket D.
\end{align*}
These count the queried inputs and collision pairs in the database, respectively.
We prove the following bound in \cref{app: compressed}; the second bound is not tight but sufficient for our purpose.

\begin{lemma}\label{lem: pc recorded counts}
Let $\ket{\Xi_t}$ be the final state over the algorithm and database registers outputted by some algorithm after making $t\le N$ queries to the compressed oracle $O$. 
Then
\[
    \bra{\Xi_t}\mathsf D_{\mathrm{cnt}}\ket{\Xi_t}\le t,
    \qquad
    \bra{\Xi_t}\mathsf{Col}_{\mathrm{cnt}}\ket{\Xi_t}\le129t.
\]
\end{lemma}

Write the compressed execution of $\Alg$ as
\[
    W=U_qOU_{q-1}\cdots U_1OU_0,
    \qquad
    \ket{\Psi_q}=W\ket{\Omega},
\]
where $\ket{\Omega}=\ket{0}\ket{\emptyset}$ is the initial algorithm state and the $U_i$ are oracle-independent unitaries acting only on the algorithm registers.
On each database register, define
\[
    Q:=I-\ketbra{\bot}.
\]
Observe the following identities
\[
    Q_xQ_y\mathsf{Col}_{xy}\ket{\emptyset}
    =
    \left(\mathsf{Col}_{xy}-\frac1N I\right)\ket{\emptyset},
    \qquad
    \mathsf{Col}_{xy}Q_xQ_y\ket{\emptyset}=0
\]
followed by expanding
$\mathsf{Col}_{xy}\ket{\bot,\bot}=N^{-1}\sum_v\ket{\widehat v,\widehat v}$.
Define, using the commutator $[A,B]=AB-BA$,
\begin{align}\label{eqn: pc correction}
    T:=\sum_{x<y}[Q_xQ_y,\mathsf{Col}_{xy}],
    \qquad
    R:=\Delta\mathsf{Col}-T
\end{align}
that acts on the database registers.
The two identities above gives
\begin{align*}
    T\ket{\emptyset}
    &=
    \sum_{x<y}
    \left(
        Q_xQ_y\mathsf{Col}_{xy}\ket{\emptyset}
        -
        \mathsf{Col}_{xy}Q_xQ_y\ket{\emptyset}
    \right)
    =
    \sum_{x<y}
        \left(\mathsf{Col}_{xy}-\frac1N I\right)
        \ket{\emptyset}
    =\Delta\mathsf{Col}\ket{\emptyset}.
\end{align*}
Since $R=\Delta\mathsf{Col}-T$ acts only on the database registers, we also have $R\ket{\Omega}=0.$
Using $RW=WR+[R,W]$, we therefore obtain
\begin{align}\label{eqn: pc circuit commutator}
    RW\ket{\Omega}
    &=WR\ket{\Omega}+[R,W]\ket{\Omega}
      =[R,W]\ket{\Omega}
        \notag\\
    &=
    \sum_{j=1}^q
        (U_qO\cdots OU_j)[R,O]
        (U_{j-1}O\cdots OU_0)\ket{\Omega}.
\end{align}
The last equality uses $[R,U_i]=0$ and the product rule for commutators.
The first and last products are interpreted as $U_q$ when $j=q$ and $U_0$ when $j=1$, respectively.

We further observe that $[R,O]=-[T,O]$.
The query $O$ preserves each database state $\ket{\widehat f}$,
while $\mathsf{Col}_{xy}$ acts on it by multiplication by
$\truth{f(x)=f(y)}$.
This gives
\[
    \mathsf{Col}_{xy}O\ket{z,w}\ket{\widehat f}
    =
    \truth{f(x)=f(y)}
        \ket{z,w+f(z)}\ket{\widehat f}
    =
    O\mathsf{Col}_{xy}\ket{z,w}\ket{\widehat f}.
\]
Thus $[\mathsf{Col}_{xy},O]=0$ for every $x<y$, and hence
$[\Delta\mathsf{Col},O]=0$.
Since $R=\Delta\mathsf{Col}-T$, it follows that
\[
    [R,O]
    =
    [\Delta\mathsf{Col},O]-[T,O]
    =
    -[T,O].
\]

Define the states immediately before a query and the corresponding backward states by
\[
    \ket{\Psi_j^-}:=U_{j-1}O\cdots OU_0\ket{\Omega},
    \qquad
    \ket{\Phi_j}:=U_j^\dagger O\cdots OU_q^\dagger\Pi\ket{\Psi_q}.
\]

Note also that $T^\dagger=-T$ and $[\Pi,T]=0$ because $T$ acts on the database while $\Pi$ acts on the algorithm's register. This shows that the expectation of $\Pi T$ is purely imaginary.
Therefore,
\[
    \bra{\Psi_q}\Pi\Delta\mathsf{Col}\ket{\Psi_q}
    =
    \operatorname{Re}\bra{\Psi_q}\Pi R\ket{\Psi_q}.
\]

Combining the preceding identities gives
\begin{align}\label{eqn: pc query sum}
    \bra{\Psi_q}\Pi\Delta\mathsf{Col}\ket{\Psi_q}
    =
    \operatorname{Re}\bra{\Psi_q}\Pi RW\ket{\Omega}
    &
    =
    \operatorname{Re}\bra{\Psi_q}\Pi[R,W]\ket{\Omega}
        \notag\\
    &
    =
    \sum_{j=1}^q
        \operatorname{Re}\bra{\Phi_j}[R,O]\ket{\Psi_j^-}  
        \notag\\
    &=-\sum_{j=1}^q
        \operatorname{Re}\bra{\Phi_j}[T,O]\ket{\Psi_j^-}  
\end{align}

\subsection{Proof of the planted-collision bound}\label{subsec: pc completion}

\begin{proof}[Proof of \cref{thm: improved planted_collision}]
The cases $q=0$ and $q>N/2$ are immediate.
We assume $1\le q\le N/2$.

Define
\[
    B_x:=\ketbra{\bot}_x\sum_{y\ne x}Q_y\mathsf{Col}_{xy},
    \qquad
    B:=\sum_x\ketbra x\otimes B_x,
\]
where the first register in $B$ is the query input register.
For a fixed query input $x$, we have
\[
    [Q_xQ_y,\mathsf{Col}_{xy}]
    =
    [Q_y,\mathsf{Col}_{xy}]
    -\ketbra{\bot}_xQ_y\mathsf{Col}_{xy}
    +\mathsf{Col}_{xy}Q_y\ketbra{\bot}_x.
\]
The first term commutes with the query.
The last two terms, summed over $y\ne x$, give $-B_x+B_x^\dagger$.
All terms not involving $x$ also commute with the query.
Hence
\begin{align}\label{eqn: pc local commutator}
    [T,O]=OB-BO+B^\dagger O-OB^\dagger.
\end{align}

We show in \cref{app: compressed} that every joint algorithm-database state $\ket{\Xi}$ whose database registers are supported on $\operatorname{span}\{\ket{\widehat f}:f\in\F\}$ satisfies
\begin{align}\label{eqn: pc B bound}
    \|B^\dagger\ket{\Xi}\|^2
    \le\frac4N
        \bra{\Xi}\bigl(\mathsf D_{\mathrm{cnt}}
                    +\mathsf{Col}_{\mathrm{cnt}}\bigr)\ket{\Xi}.
\end{align}

Each of the four vectors $\ket{\Phi_j}, O\ket{\Phi_j},
    \ket{\Psi_j^-}, O\ket{\Psi_j^-}$
has norm at most one.
Recall the states $\ket{\Psi_j^-}$ and $O\ket{\Psi_j^-}$ are obtained after
$j-1$ and $j$ queries to $O$, respectively.

To apply \cref{lem: pc recorded counts} to $\ket{\Phi_j}$,
consider an algorithm that runs $W$, copies its output bit to an
auxiliary qubit using a CNOT gate, and then applies
$U_j^\dagger O\cdots OU_q^\dagger$ to the original registers.
The component of the resulting state in which the auxiliary qubit
is $1$ is precisely $\ket{\Phi_j}$.
This algorithm makes $2q-j$ queries, since $O^\dagger=O$.
One additional query gives the component $O\ket{\Phi_j}$ and uses
$2q-j+1\le2q$ queries in total.
Both count operators are positive and act only on the database registers,
so their expectations in either component are at most their expectations
in the corresponding full state.
Thus, since $2q\le N$, \cref{lem: pc recorded counts} gives
\[
    \bra{\Xi}\mathsf D_{\mathrm{cnt}}\ket{\Xi}\le2q,
    \qquad
    \bra{\Xi}\mathsf{Col}_{\mathrm{cnt}}\ket{\Xi}\le258q
\]
for each of the four vectors $\ket{\Xi}$ above.
Applying \cref{eqn: pc B bound}, we obtain
\[
    \|B^\dagger\ket{\Xi}\|^2
    \le\frac4N(2q+258q)
    =\frac{1040q}{N}.
\]

By the Cauchy--Schwarz inequality, each of the four terms in
\cref{eqn: pc local commutator}, between $\bra{\Phi_j}$ and
$\ket{\Psi_j^-}$, has absolute value at most $\sqrt{1040q/N}$.
For example,
\[
    \left|\bra{\Phi_j}OB\ket{\Psi_j^-}\right|
    \le
    \|B^\dagger O\ket{\Phi_j}\|\,
    \|\ket{\Psi_j^-}\|
    \le\sqrt{\frac{1040q}{N}}.
\]
It follows that
\[
    \left|\bra{\Phi_j}[T,O]\ket{\Psi_j^-}\right|
    \le4\sqrt{\frac{1040q}{N}}.
\]
Using \cref{eqn: pc query sum}, we obtain
\[
    \left|\bra{\Psi_q}\Pi\Delta\mathsf{Col}\ket{\Psi_q}\right|
    \le
    \sum_{j=1}^q
        \left|\bra{\Phi_j}[T,O]\ket{\Psi_j^-}\right|
    \le4q\sqrt{\frac{1040q}{N}}
    =16\sqrt{65}\,\frac{q^{3/2}}{\sqrt N}.
\]
Finally, \cref{eqn: pc compressed expectation} concludes the proof of the lemma.
\end{proof}

\ifnum\anonymous=1
\else
\fi

\bibliographystyle{plain}
\bibliography{ref,cryptobib/abbrev3,cryptobib/crypto}

\appendix
\section{Collision Bias and Heuristic Attacks}\label{app: attacks}

We present a heuristic outline of the quantum attacks described in
\cref{subsec:attacks}.
Throughout, $r\ge 2$ is fixed, and the oracle $O$ is either a uniformly
random function ($O=\mathrm{rf}$) or the $\xop[r]$ construction
($O=\xop$).
The collision statistic calculations below are intended to explain the expected scalings of advantages.

Recall the collision probability of $\xop[r]$ 
\[
\Pr\left[\xop[r](x) = \xop[r](x')\right] = \frac1N + \frac{(-1)^r}{N(N-1)^{r-1}}
,\qquad \delta_r:=\frac{(-1)^r}{N(N-1)^{r-1}}
\]
where $\delta_r$ is the difference between the collision probabilities
of $\xop[r]$ and a random function, and
$|\delta_r|=\Theta(N^{-r})$.

All attacks except the last one are based on the collision finding attack \cite{BHT97}. We actually \emph{count} the number of the (parts of) collisions in the following way:
\begin{enumerate}[nosep]
    \item For an appropriate set $S \subset \bit^n$, compute the table $T_S:=\{O(x):x\in S\}$, the set of distinct outputs obtained from the inputs in $S$.
    Define the indicator function $f_S(z)= \truth{O(z) \in T_S}$ for $z \in \bit^n \setminus S$.
    \item Approximately count the number of $z \in \bit^n \setminus S$ such that $f_S(z)=1$.
\end{enumerate}
Note that $f_S(z)= \truth{O(z) \in T_S}$ indicates that $z$ is a part of the \emph{cross collision} $O(z)=O(x)$ for some $x\in S$. 
The first step makes $s:=|S|$ queries to the oracle $O$. Let $K=|\{z\in \bit^n \setminus S : f_S(z)=1\}|$ be the number of inputs outside $S$ whose outputs collide with one of the stored values in $T_S.$
Using quantum counting \cite{BHMT02} with $t$ queries to $O$, we can obtain an estimate $\widetilde K$ satisfying
\[
\left|\frac{\widetilde K}{N-s} - \frac{K}{N-s} \right| = O\left(
    \frac{\sqrt{K/(N-s)}}{t} + \frac{1}{t^2}
\right)
\]
with a constant probability. In other words, if $s,t=\Theta(q)$ and $K=\Theta(s)$, the additive error $|\widetilde K-K|$ is bounded by
\[
O\left(\sqrt{\frac{N}{q}} + \frac{N}{q^2}\right).
\]

\newcommand{\Bin}{\mathsf{Bin}}
\newcommand{\std}{\mathsf{std}}
We first compare the distribution of $K$ under the two oracle choices.
If $O=\mathrm{rf}$, then conditioned on $T_S$, the values $O(z)$ for
$z\notin S$ are independent and uniform. Therefore,
\[
    K\mid T_S
    \sim
    \Bin\left(
        N-s,\frac{|T_S|}{N}
    \right).
\]
In particular, for $s=\Theta(q)$ and $s\le cN$ for a fixed constant
$c<1$, we have
\begin{align}
    \label{eqn: appendix_stat}
    \Exp_{\mathrm{rf}}[K]=\Theta(q),
    \qquad
    \std_{\mathrm{rf}}(K)=\Theta(\sqrt q).
\end{align}
For $\xop[r]$, the collision probability between every pair of distinct inputs differs from that of a random function by $\delta_r$.
As a heuristic approximation, we expect the number of cross collisions to differ by
\[
\left|
\Exp_{\xop}[K]-\Exp_{\mathrm{rf}}[K]
\right|
\approx
(N-s)s|\delta_r|
=
\Theta\left(
\frac{q}{N^{r-1}}
\right),
\]
where we use $s=\Theta(q)$ and $N-s=\Theta(N)$.
We use this heuristic estimate below to derive the expected advantage
of the collision-based attacks.

\paragraph{Small-query regime: $q\lesssim N^{1/3}$.}
In this regime, it suffices to search for a collision rather than
estimate $K$ accurately.
The algorithm first constructs $T_S$ using $s=\Theta(q)$ queries and
then uses $t=\Theta(q)$ further queries to search for a
$z\notin S$ satisfying $f_S(z)=1$.

For a fixed value of $K$, Grover search using $t$ queries finds such a
$z$ with probability
\[
\Theta\left(\frac{t^2K}{N}\right)
\]
which coincides with $\Theta\left(\frac{q^3}{N}\right)$ for our parameter choice when $K=\Theta(q)$.

Using the heuristic estimate above for the difference in $\Exp[K]$,
we expect the difference between the acceptance probabilities under
the two oracle choices to be of order
\[
\Theta\left(
\frac{t^2}{N}
\left|
\Exp_{\xop}[K]-\Exp_{\mathrm{rf}}[K]
\right|
\right)
=
\Theta\left(
\frac{q^3}{N^r}
\right).
\]
Thus the collision-finding attack is expected to achieve advantage
$\Theta(q^3/N^r)$ in this regime.

\paragraph{Moderate-query regime:
$N^{1/3}\lesssim q\lesssim N^{1/2}$.}
We now use quantum counting to estimate $K$.
Let $\widetilde K$ be the estimate obtained using $t=\Theta(q)$
queries in the second step.
The counting error discussed above (obtained with constant probability) is
\[
|\widetilde K - K | =
O\left(
        \sqrt{\frac Nq}
        +
        \frac{N}{q^2}
    \right) = O\left(\sqrt{\frac Nq}\right)
\]
since $q\gtrsim N^{1/3}$.
Also note that the standard deviation $\Theta(\sqrt q)$ of $K$ in \cref{eqn: appendix_stat} is relatively smaller than this error in this regime. This comparison will be important in the next regime.

Using the heuristic estimate above, we expect
\[
\left|
\Exp_{\xop}[K]-\Exp_{\mathrm{rf}}[K]
\right|
=
\Theta\left(\frac{q}{N^{r-1}}\right).
\]
Comparing this difference with the quantum-counting error
$\Theta(\sqrt{N/q})$ heuristically suggests distinguishing advantage
of order
\[
    \Theta\left(
        \frac{
            q/N^{r-1}
        }{
            \sqrt{N/q}
        }
    \right)
    =
    \Theta\left(
        \frac{q^{3/2}}{N^{r-1/2}}
    \right).
\]

 \paragraph{Large-query regime:
$N^{1/2}\lesssim q\lesssim N$.}
Continue to use $s=\Theta(q)$ queries to construct the table and
$t=\Theta(q)$ queries for quantum counting.
The additive error of the estimate is, given $q\gtrsim N^{1/2}$,
\[
    O\left(
        \sqrt{\frac Nq}
        +
        \frac{N}{q^2}
    \right) = 
    O(\sqrt q).
\]
Note that the standard deviation of $K$ in \cref{eqn: appendix_stat} is already $\Theta(\sqrt q)$.
This means that the values of $K$ for two different random oracles may be of order $\Theta(\sqrt q)$, which is no smaller than the estimation error.

Using the heuristic estimate above, we expect
\[
    \left|
        \Exp_{\xop}[K]-\Exp_{\mathrm{rf}}[K]
    \right|
    =
    \Theta\left(
        \frac{q}{N^{r-1}}
    \right).
\]
Comparing this difference with the above typical difference $\sqrt{q}$ of two random oracles suggests distinguishing advantage
of order, heuristically,
\[
    \Theta\left(
        \frac{
            q/N^{r-1}
        }{
            \sqrt q
        }
    \right)
    =
    \Theta\left(
        \frac{\sqrt q}{N^{r-1}}
    \right).
\]
At $q=\Theta(N)$, this heuristic scaling becomes
$\Theta(N^{3/2-r})$, matching the second bound of
\cref{thm: main_bound}.

\paragraph{Parity attack ($q\ge N/2$).}
Finally, suppose that $q\ge N/2$.  Fix a nonzero
$a\in\bit^n$ and define the Boolean function $g_O(x):=\langle a,O(x)\rangle$ for the oracle $O$, which can be computed by a single query to $O$.
For the $\xop[r]$ construction, assuming $n\ge2$, the parity of $g_{\xop[r]}$, i.e., $    \bigoplus_{x\in\bit^n} g_{\xop[r]}(x)$ becomes 0 because
\[
    \left\langle
        a,
        \bigoplus_{x\in\bit^n}\xop[r](x)
    \right\rangle
    =
    \left\langle
        a,
        \bigoplus_{i=1}^r\bigoplus_{x\in\bit^n}P_i(x)
    \right\rangle
    =
    \left\langle
        a,
        \bigoplus_{i=1}^r\bigoplus_{y\in\bit^n}y
    \right\rangle
    =
    0.
\]
For a random function, this parity is a uniformly random bit.
The parity of $g_O$ can be computed exactly with $N/2$ quantum queries by pairing up the $N$ inputs and computing the XOR of each pair with Deutsch's algorithm, which is optimal~\cite{FGGS98}. The distinguisher that outputs the parity thus has advantage $1/2$.

\ifnum\fullpage=0
\section{Missing Proofs for Fourier Analysis}\label{app:Fourier}

The following proof shows $\widehat{H *G }(\gamma) = \widehat{H}(\gamma)\widehat{G}(\gamma)$.
\begin{proof}[Proof of \cref{lem: convol}]
In \cite[Theorem 1.27]{o2014analysis}, the same statement for real-valued boolean functions is proven. It can be extended to our case as follows.
\begin{align*}
\widehat{H *G }(\gamma) ={}& \Exp_{f \gets \F}\left[\overline{\chi_\gamma(f)}H*G(f)\right]\\
={}& \Exp_{f \gets \F}\left[\overline{\chi_\gamma(f)}\Exp_{g\gets \F} [H(f+g)G(g)]\right]\\
={}& \Exp_{g,h \gets \F}\left[\overline{\chi_\gamma(g+h)}H(h)G(g)\right]\\
={}& \Exp_{g,h \gets \F}\left[\overline{\chi_\gamma(h)}H(h)\overline{\chi_\gamma(g)}G(g)\right]\\
={}& \widehat{H}(\gamma)\widehat{G}(\gamma)
\end{align*}
since $\overline{\chi_\gamma}=\chi_\gamma$ and by \cref{eqn: bilinear}.
\end{proof}

\begin{proof}[Proof of \cref{lem: convolDensity}]
Using $\dens{\D}(f)=|\F|\Pr_{\D}[f]$,
\begin{align*}
    (\dens{\D}*\dens{\mathcal E})(h)
    =\sum_{g\in\F}\frac{\dens{\D}(h+g)\dens{\mathcal E}(g)}{|\F|}
    &=|\F|\sum_{g\in\F}\Pr_{f\gets\D}[f=h+g]\Pr_{g'\gets\mathcal E}[g'=g]\\
    &=|\F|\Pr_{f\gets\D,\,g\gets\mathcal E}[f+g=h].\qedhere
\end{align*}
\end{proof}


We prove \cref{thm: Prob by functionals} here. For convenience, we recall the theorem statement.

\noindent\textbf{\cref{thm: Prob by functionals}.}
\emph{
    Let $f$ be a function from $\bit^n$ to $\bit^n$.
    Let $\mathcal A$ be a quantum oracle algorithm having access to $O_f$ and outputting a bit $b \in \bit$. 
    Then, there exists a functional $P_{\mathcal A}$ (which may depend on the input of $\mathcal A$, if any) of Fourier degree $\le 2q$ such that 
    the probability that $\mathcal A$ outputs $1$ equals to $P_{\mathcal A}(f)$.
    In other words, the functional $P_{\mathcal A}(f):=\Pr[\mathcal A^{O_f} \to 1]$ has the Fourier degree $\le 2q$ and can be decomposed by
    \[
        \Pr[\mathcal A^{O_f} \to 1] = P_{\mathcal A}(f) = \sum_{|\gamma|\le 2q} \widehat{P}_{\mathcal A}(\gamma) \chi_\gamma(f).
    \]
    In particular, $\|P_{\mathcal A}\|_2 \le 1.$ }
\begin{proof}[Proof of \cref{thm: Prob by functionals}]
    We first show the following claim.
    \begin{claim}
    The final state of a $q$-query algorithm (with the deferred measurement) can be written by
    \[
    \sum_z p_z^{(q)}(f) \ket{z}
    \]
    where $p_z^{(q)}$ is a functional of degree $\le q$.
    \end{claim} 
    \begin{proof}[Proof of Claim]
    Note that a $q$-query algorithm is represented by an alternating sequence of oracle query $O_f$ and an oracle-independent unitary $U$.
    Applying unitary $U$, as a linear map, does not increase the Fourier degree. We only focus on the oracle queries.

    Before proceeding, observe that for each query for the Fourier basis (\cref{eqn: Fourier_basis_state}) it holds that
    \[
    O_f (\chi_\gamma (f) \ket{x,\widehat a})  =\chi_\gamma (f)  \cdot \chi_{a}(f(x))\ket{x,\widehat a} = \chi_{\gamma + a \delta_x}(f) \ket{x,\widehat a}
    \]
    where $a\delta_x \in \hF$ is such that $a\delta_x(x)=a$ and $a\delta_x(y)=0^n$ for $y\neq x$.
    We use \cref{eqn: query} and \cref{eqn: bilinear} in the first and second equality. It holds that
    \[
    |\gamma + a \delta_x| \le |\gamma| + |a\delta_x| \le |\gamma|+1.
    \]

    We use the mathematical induction to prove the claim. For the state before any query, $p_z^{(0)}$ is independent of $f$, thus the base step holds. 
    
    Next, consider the state before applying the $k$-th query $\sum_z p_z^{(k-1)}(f) \ket{z}$ for $k\le q$. The Fourier degree of $p_z^{(k-1)}$ is bounded above by $k-1$ by the inductive hypothesis.
    Then, a single query increases the Fourier degree of each term in $p_z^{(k-1)}$ by at most 1, so that the degree after the query is bounded above by $(k-1)+1=k$.
    This proves the claim.
    \end{proof}

    Write $\ket{\phi_q} = \sum_z p_z(f) \ket{z}$ to denote the final state of the algorithm where $p_z$ is a functional of Fourier degree $\le q$.
    Let $(\Pi_0,\Pi_1=I-\Pi_0)$ be the final binary measurement of the algorithm. Then, the probability that $\mathcal A$ outputs 1 is
    \begin{align}
        \Pr[\mathcal A^{O_f}\to 1] &= \bra{\phi_q} \Pi_1 \ket{\phi_q}\\
        &= \sum_{z,w} \bra z \overline{p_z(f)} \Pi_1 p_w(f) \ket w\\
        &= \sum_{z,w} \bra z  \Pi_1  \ket w \cdot \overline{p_z(f)}p_w(f)
    \end{align}
    where $\deg(\overline{p_z}p_w) \le \deg(\overline{p_z}) + \deg(p_w) \le 2q$ because of \cref{eqn:deg triangle}.
    Since the summation does not increase the Fourier degree, $P_{\mathcal A}(f) = \Pr[\mathcal A^{O_f}\to 1]$ is a functional of Fourier degree $\le 2q$.

    The final ``in particular'' part is obvious because $|P_{\mathcal A}(f)|\le 1$ for all $f$.
\end{proof}

\fi
\ifnum\fullpage=0
\section{Higher-degree XoP components}

\subsection{Degree-4 components}\label{app: deg4}

\paragraph{Proof of \cref{eqn:deg4permutation}.}
Fix $\gamma \in \hF$ with $\supp(\gamma)=\{x_1,x_2,x_3,x_4\}.$ Write $\gamma(x_i) = y_i \neq 0^n$ for $i\in [4]$. For convenience, we recall \cref{eqn:deg4permutation}:
\[
    \hmuP(\gamma)=\dfrac{N\,C_{2,2}(\gamma)-6\,C_4(\gamma)}{(N-1)_3}.
\]
By \cref{eqn: simplifiaction},
\begin{align}\label{eqn: simple4}
    \hmuP(\gamma)
    =
    \sum_{z_1,z_2,z_3,z_4 \text{ distinct}}
    \frac{
    \chi_{y_1}(z_1)\chi_{y_2}(z_2)
    \chi_{y_3}(z_3)\chi_{y_4}(z_4)
    }{(N)_4}.
\end{align}

For four variables $z_1,z_2,z_3,z_4$, the inclusion-exclusion principle
for the six events $E_{ij}=[z_i=z_j]$ gives, after simplification,
\begin{align*}
    &\truth{z_1,z_2,z_3,z_4 \text{ distinct}}
    =
    1-\sum_{i<j}\truth{z_i=z_j}
    +2\sum_{i<j<k}\truth{z_i=z_j=z_k}
    \\
    &\qquad
    +
    \sum_{i<j, k<\ell, \{i,j,k,\ell\}=\{1,2,3,4\}}
    \truth{z_i=z_j\land z_k=z_\ell}
    -6\truth{z_1=z_2=z_3=z_4}.
\end{align*}
Applying this to \cref{eqn: simple4} over all $z_1,z_2,z_3,z_4$,
every term with a separated single vertex vanishes by \cref{eqn: characsum}.
Hence only the last two terms remain:
\begin{align*}
    \hmuP(\gamma)
    &=
    \sum_{\substack{i<j, k<\ell\\ \{i,j,k,\ell\}=\{1,2,3,4\}}}
    \sum_{z,w\in \bit^n}
    \frac{\chi_{y_i+y_j}(z)\chi_{y_k+y_\ell}(w)}{(N)_4}
    -
    \sum_{z\in \bit^n}
    \frac{6\chi_{y_1+y_2+y_3+y_4}(z)}{(N)_4}
    \notag \\
    &=
    \sum_{\substack{i<j, k<\ell\\ \{i,j,k,\ell\}=\{1,2,3,4\}}}
    \frac{N\truth{y_i=y_j}\truth{y_k=y_\ell}}{(N-1)_3}
    -
    \frac{6\truth{y_1+y_2+y_3+y_4=0^n}}{(N-1)_3}
    \notag \\
    &=
    \frac{N C_{2,2}(\gamma)-6C_4(\gamma)}{(N-1)_3},
\end{align*}
where we use \cref{eqn: characsum} in the second equality.
This proves \cref{eqn:deg4permutation}.

\paragraph{Detailed calculation of \cref{eqn: simplified deg4 upper}.}
Our interested term is
\begin{align*}
    &
    \left(\frac{3N}{(N-1)_3}\right)^{2r-2}
    \frac{
        \binom N4
        \left(
            9N^2(N-1)^2+36(N-1)^3
        \right)
    }{((N-1)_3)^2}
    \\
    &\qquad\qquad=
    \frac{3^{2r}}{N^{4r-6}}
    \frac{
        \displaystyle
        \frac1{24}+\frac{N-1}{6N^2}
    }{
        \displaystyle
        \left(1-\frac1N\right)^{2r-3}
        \left(1-\frac2N\right)^{2r-1}
        \left(1-\frac3N\right)^{2r-1}
    }.
\end{align*}
Observe that
$    \frac1{24}+\frac{N-1}{6N^2}\le \frac{1}{12}$ for $N\ge 4$
and
\begin{align*}
    \left(1-\frac1N\right)^{2r-3}
    \left(1-\frac2N\right)^{2r-1}
    \left(1-\frac3N\right)^{2r-1}
    \ge
    1-\frac{12r-8}{N}
    \ge \frac14
\end{align*}
for $N\ge16r$. This bounds the above term by
\[
\le \frac{3^{2r}}{N^{4r-6}} \cdot \frac{1/12}{1/4} = \frac{3^{2r-1}}{N^{4r-6}}
\]
as desired.

\subsection{Higher-degree components}\label{app: higher}

Now we extend the low-degree calculations to the higher-degree terms.
We prove the following lemmas.

\begin{lemma}\label{lem: higher degree tail}
For $r\ge 2$ and $5\le d_{\max}\le N/28887$, it holds that $\|\muXoPr^{=6}\|_2^2 \leq 3^{18}2^3\left(\frac{12}{N}\right)^{6r-9}$ and
\[
    \|\muXoPr^{=5}\|_2^2
    +
    \sum_{d=7}^{d_{\max}}
    \|\muXoPr^{=d}\|_2^2
    \le
    3^{22}2^{11}\left(\frac{10}{N}\right)^{6r-8}.
\]
\end{lemma}

\begin{proof}
\cref{lem: deg d magnitude bound,lem: deg d weight bound} give,
for every $5\le d\le d_{\max}$,
\begin{align*}
    \|\muXoPr^{=d}\|_2^2
    &=
    \sum_{|\gamma|=d}|\hmuP(\gamma)|^{2r}
    \\
    &\le
    \left(M^{=d}[\muP]\right)^{2r-2}
    \sum_{|\gamma|=d}|\hmuP(\gamma)|^2
    \\
    &\le
    \left(\frac{2d}{N}\right)^{
        (2r-2)\lceil d/2\rceil
    }
    3^{3d}
    \left(\frac Nd\right)^{\lfloor d/2\rfloor}
    \\
    &=
    3^{3d}
    2^{(2r-2)\lceil d/2\rceil}
    \left(\frac dN\right)^{
        (2r-1)\lceil d/2\rceil-d
    }
    =:
    W_d,
\end{align*}
where we use the following in the last line
\[
    \left\lfloor\frac d2\right\rfloor
    +
    \left\lceil\frac d2\right\rceil
    =
    d.
\]
Degree-6 is immediate by $W_6$, in particular $N\geq 96$. We will prove
\[
    \frac{W_{d+2}}{W_d}\le\frac12
\]
for every $5\le d\le d_{\max}-2$. Since
\(
    \left\lceil\frac{d+2}{2}\right\rceil
    =
    \left\lceil\frac d2\right\rceil+1,
\)
we have
\begin{align*}
    \frac{W_{d+2}}{W_d}
    &=
    3^6 2^{2r-2}
    \left(1+\frac2d\right)^{
        (2r-1)\lceil d/2\rceil-d
    }
    \left(\frac{d+2}{N}\right)^{2r-3}.
\end{align*}
For $d\ge5$,
\[
    \left\lceil\frac d2\right\rceil
    \le
    \frac{3d}{5},
\]
and hence
\[
    (2r-1)\left\lceil\frac d2\right\rceil-d
    \le
    \frac{6r-8}{5}d.
\]
Using $1+x\le e^x$, we obtain
\begin{align*}
    \left(1+\frac2d\right)^{
        (2r-1)\lceil d/2\rceil-d
    }
    &\le
    \exp\left(
        \frac2d
        \left(
            (2r-1)\left\lceil\frac d2\right\rceil-d
        \right)
    \right)
    \le
    e^{(12r-16)/5}.
\end{align*}
In addition, we observe that, for all $r\ge 2$,
\[
            3^6 2^{2r-2}e^{(12r-16)/5}
    < \frac{1}{2} \cdot 28887^{2r-3}.
\]
Using the above,
\begin{align*}
    \frac{W_{d+2}}{W_d}
    &\le
    3^6 2^{2r-2}e^{(12r-16)/5}
    \left(\frac{d+2}{N}\right)^{2r-3}
    \le
    \frac{1}{2} \left( 28887 \right) ^{2r-3}
    \left(\frac{d_{\max}}{N}\right)^{2r-3}
    \le
    \frac12
\end{align*}

Summing the odd and even degrees separately gives
\begin{align*}
    &\|\muXoPr^{=5}\|_2^2
    +
    \sum_{d=7}^{d_{\max}}
    \|\muXoPr^{=d}\|_2^2
    \le
    W_5+
    \sum_{d=7}^{d_{\max}}W_d
    \\
    &\le
    \left(
        1+\frac12+\frac1{2^2}+\cdots
    \right)
    (W_5+W_8)
    =
    2(W_5+W_8).
\end{align*}
Moreover,
\[
    \frac{W_5}{W_8}
    =
    \frac{5^{6r-8}}{3^9\cdot 2^{26r-38}}N^{2r-4}
    =
    \frac{5^4}{3^9\cdot 2^{14}}\left(\frac{125N}{8192}\right)^{2r-4}
    \ge
    \frac{5^4}{3^9\cdot 2^{14}}.
\]
Consequently,
\begin{align*}
    &\|\muXoPr^{=5}\|_2^2
    +
    \sum_{d=7}^{d_{\max}}
    \|\muXoPr^{=d}\|_2^2
    \le
    2\cdot\left(1+\frac{3^9\cdot2^{14}}{5^4}\right)W_5
    \\&\le
    2\cdot\frac{3^9\cdot2^{14}}{3^2\cdot2^6}\cdot3^{15}\cdot2^{6r-6}
    \left(\frac5N\right)^{6r-8}
    =
    3^{22}2^{11}
    \left(\frac{10}{N}\right)^{6r-8}.\qedhere
\end{align*}
\end{proof}

\section{Maximal Magnitude and Weight of the components of $\muP$}\label{app: magniweight}

\subsection{Maximal magnitude of components}
The goal of this section is to prove \cref{lem: deg d magnitude bound}, or more precisely
\[
    M^{=d}[\muP]
    \le
    \left(\frac{2d}{N}\right)^{\lceil d/2\rceil}.
\]

Let $\gamma\in\hF$ be a character with $|\gamma|=d$ such that $\supp(\gamma)=\{x_1<\cdots<x_d\}.$
Write $\gamma(x_i)=y_i$ for $i\in[d].$
As shown in \cref{eqn: simplifiaction}, we can write
\[
    \hmuP(\gamma)
    =
    \Exp_{(z_1,\ldots,z_d)\gets(\bit^n)^{(*d)}}
    \left[
        \prod_{i=1}^d\chi_{y_i}(z_i)
    \right].
\]
We observe the following fact: The conditional distribution of $z_d$ given the other elements
$(z_1,\ldots,z_{d-1})$ is uniform over the set
$\bit^n\setminus\{z_1,\ldots,z_{d-1}\},$
which has size $N-d+1$.
From this, we have
\begin{align*}
    &\Exp_{z_d\gets
    \bit^n\setminus\{z_1,\ldots,z_{d-1}\}}
    \left[
        \chi_{y_d}(z_d)
    \right]
    \\
    &\quad=
    \frac1{N-d+1}
    \sum_{z\notin\{z_1,\ldots,z_{d-1}\}}
    \chi_{y_d}(z)
    \\
    &\quad=
    \frac1{N-d+1}
    \left(
        \sum_{z\in\bit^n}\chi_{y_d}(z)
        -
        \sum_{i=1}^{d-1}\chi_{y_d}(z_i)
    \right)
    \\
    &\quad=
    -\frac1{N-d+1}
    \sum_{i=1}^{d-1}\chi_{y_d}(z_i),
\end{align*}
where the last equality follows from
\(
    \sum_{z\in\bit^n}\chi_{y_d}(z)=0
\)
because $y_d\neq0^n$.

Using this equation, we have
\begin{align*}
    \hmuP(\gamma)
    &=
    \Exp_{(z_1,\ldots,z_d)\gets(\bit^n)^{(*d)}}
    \left[
        \chi_{y_d}(z_d)
        \prod_{j=1}^{d-1}\chi_{y_j}(z_j)
    \right]
    \\
    &=
    \Exp_{(z_1,\ldots,z_{d-1})\gets(\bit^n)^{(*(d-1))}}
    \left[
        \Exp_{z_d\gets
        \bit^n\setminus\{z_1,\ldots,z_{d-1}\}}
        \left[
            \chi_{y_d}(z_d)
        \right]
        \prod_{j=1}^{d-1}\chi_{y_j}(z_j)
    \right]
    \\
    &=
    -\frac1{N-d+1}
    \sum_{i=1}^{d-1}
    \Exp_{(z_1,\ldots,z_{d-1})\gets(\bit^n)^{(*(d-1))}}
    \left[
        \chi_{y_d}(z_i)
        \prod_{j=1}^{d-1}\chi_{y_j}(z_j)
    \right]
    \\
    &=
    -\frac1{N-d+1}
    \sum_{i=1}^{d-1}
    \Exp_{(z_1,\ldots,z_{d-1})\gets(\bit^n)^{(*(d-1))}}
    \left[
        \chi_{y_i+y_d}(z_i)
        \prod_{\substack{j=1\\j\neq i}}^{d-1}
        \chi_{y_j}(z_j)
    \right].
\end{align*}

To simplify the above expression, define
$\gamma^{(i\leftarrow d)}\in\hF$ by
\[
    \gamma^{(i\leftarrow d)}(x)
    :=
    \begin{cases}
        y_i+y_d,
        & x=x_i,\\
        y_j,
        & x=x_j
          \text{ for some }j\in[d-1]\setminus\{i\},\\
        0^n,
        & \text{otherwise}.
    \end{cases}
\]
If $y_i+y_d=0^n$, then $\gamma^{(i\leftarrow d)}(x_i)=0$ and therefore $\gamma^{(i\leftarrow d)}$ has degree $d-2$. In other cases, $\gamma^{(i\leftarrow d)}$ has degree $d-1$.
This gives the following recursion.

\begin{lemma}\label{lem: xoprecursion}
Let $2\le d\le N$, and let $\gamma\in\hF$ satisfy $|\gamma|=d$.
Then
\[
    \hmuP(\gamma)
    =
    -\frac1{N-d+1}
    \sum_{i=1}^{d-1}
    \hmuP\!\left(\gamma^{(i\leftarrow d)}\right).
\]
\end{lemma}

We are now ready to prove \cref{lem: deg d magnitude bound}.

\begin{proof}[Proof of \cref{lem: deg d magnitude bound}]
It is immediate that $M^{=0}[\muP]=1$ and $M^{=1}[\muP]=0.$
Now fix $2\le d\le N$, and let $\gamma\in\hF$ satisfy
$|\gamma|=d$. By \cref{lem: xoprecursion},
\[
    \hmuP(\gamma)
    =
    -\frac1{N-d+1}
    \sum_{i=1}^{d-1}
    \hmuP\!\left(\gamma^{(i\leftarrow d)}\right).
\]
For each $i\in[d-1]$, the character
$\gamma^{(i\leftarrow d)}$ has support size either $d-1$ or $d-2$.
Therefore,
\[
    \left|
        \hmuP\!\left(\gamma^{(i\leftarrow d)}\right)
    \right|
    \le
    \max\left\{
        M^{=(d-1)}[\muP],
        M^{=(d-2)}[\muP]
    \right\}.
\]
Since the recursion holds for every degree-$d$ character, taking
absolute values gives
\begin{align}\label{ineq: recursive}
    M^{=d}[\muP]
    \le
    \frac{d-1}{N-d+1}
    \max\left\{
        M^{=(d-1)}[\muP],
        M^{=(d-2)}[\muP]
    \right\}.
\end{align}
We now prove the claimed bound by induction on $d\ge 2$. If $d>N/2$, the bound is immediate as $M^{=d}[\muP]\le 1\le (2d/N)^{\lceil d/2\rceil}$.
So let $2\le d\le N/2$ and assume the claim for all smaller degrees. Since $2d/N\le1$ and $\lceil (d-1)/2\rceil\ge\lceil (d-2)/2\rceil$, the induction hypothesis gives
\begin{align*}
    &\max\left\{
        M^{=(d-1)}[\muP],
        M^{=(d-2)}[\muP]
    \right\}\\
    &\quad\le
    \max\left\{
        \left(\frac{2(d-1)}N\right)^{\lceil(d-1)/2\rceil},
        \left(\frac{2(d-2)}N\right)^{\lceil(d-2)/2\rceil}
    \right\}
    \le
    \left(\frac{2d}{N}\right)^{\lceil(d-2)/2\rceil}.
\end{align*}
Moreover, $\frac{d-1}{N-d+1}\le\frac{2d}{N}$ for $d\le N/2$.
Plugging both bounds into~\eqref{ineq: recursive},
\[
    M^{=d}[\muP]
    \le
    \frac{2d}{N}
    \left(\frac{2d}{N}\right)^{\lceil(d-2)/2\rceil}
    =
    \left(\frac{2d}{N}\right)^{\lceil d/2\rceil},
\]
this completes the induction and proves the lemma.
\end{proof}

\subsection{Weight of components}
We prove that, for $5\le d\le N/16$, it holds that
\[
    \sum_{|\gamma|=d}|\hmuP(\gamma)|^2
    \le
    3^{3d}
    \left(\frac Nd\right)^{\lfloor d/2\rfloor}.
\]

\begin{proof}[Proof of \cref{lem: deg d weight bound}]
By \cref{eqn: simplifiaction}, $\hmuP(\gamma)$ depends only on the labels of $\gamma$, so for a support $S$ of size $d$ the quantity
\[
    B_d:=\sum_{\supp(\gamma)=S}|\hmuP(\gamma)|^2
\]
does not depend on $S$, and $\sum_{|\gamma|=d}|\hmuP(\gamma)|^2=\binom Nd B_d$. We will show that
\begin{align}\label{eqn: Bd identity}
    B_d = \sum_{j=0}^d (-1)^{d-j}\binom{d}{j}\frac{N^j}{(N)_j}
\end{align}
and
\begin{align}\label{eqn: Bdupper}
    B_d \le 2^{d+1} \left(\frac{8d}N\right)^{\lceil d/2 \rceil}
    \qquad\text{for } 5\le d\le N/16.
\end{align}
Given~\eqref{eqn: Bdupper}, the lemma follows from $\binom Nd\le(eN/d)^d$ and $d-\lceil d/2\rceil=\lfloor d/2\rfloor$:
\begin{align*}
    \sum_{|\gamma|=d}|\hmuP(\gamma)|^2
    =
    \binom Nd B_d
    &\le
    2^{d+1}
    \left(\frac{eN}{d}\right)^d
    \left(\frac{8d}{N}\right)^{\lceil d/2\rceil}\\
    &=
    2^{d+1}8^{\lceil d/2\rceil}e^d
    \left(\frac Nd\right)^{\lfloor d/2\rfloor}
    \le
    3^{3d}
    \left(\frac Nd\right)^{\lfloor d/2\rfloor},
\end{align*}
where the last inequality uses $2^{d+1}8^{\lceil d/2\rceil}\le 8^d$ for $d\ge 5$ (equivalently, $d+1\le 3\lfloor d/2\rfloor$) and $8e<27$.
 
\paragraph{The identity \cref{eqn: Bd identity}.}
Fix $S=\{x_1,\ldots,x_d\}$ and, for $y=(y_1,\ldots,y_d)\in(\bit^n)^d$, let $\gamma_y\in\hF$ be the index with $\gamma_y(x_i)=y_i$ and $\gamma_y(x)=0^n$ for $x\notin S$; thus $\supp(\gamma_y)=\{x_i:y_i\neq 0^n\}$.
Let $\nu$ be the uniform distribution on $(\bit^n)^{(*d)}$, viewed as a probability mass function on $(\bit^n)^d$. By~\eqref{eqn: simplifiaction} (which does not require the $y_i$ to be nonzero),
\[
    \hmuP(\gamma_y)=\sum_{z\in(\bit^n)^d}\nu(z)\prod_{i=1}^d\chi_{y_i}(z_i)
    \qquad\text{for all }y\in(\bit^n)^d,
\]
i.e., $(\hmuP(\gamma_y))_{y}$ is $N^d$ times the Fourier transform of $\nu$ on the group $(\bit^n)^d$. Parseval's identity on $(\bit^n)^d$ therefore gives
\[
    \sum_{y\in(\bit^n)^d}|\hmuP(\gamma_y)|^2
    =N^{2d}\cdot\frac{1}{N^d}\sum_{z\in(\bit^n)^d}\nu(z)^2
    =N^d\cdot\frac{(N)_d}{(N)_d^2}
    =\frac{N^d}{(N)_d}.
\]
On the other hand, if exactly the coordinates in $J\subseteq[d]$ of $y$ are nonzero, then $\hmuP(\gamma_y)$ is the Fourier coefficient of a degree-$|J|$ index, so grouping the $y$ according to $J$ yields $\sum_{y\in(\bit^n)^d}|\hmuP(\gamma_y)|^2=\sum_{J\subseteq[d]}B_{|J|}=\sum_{j=0}^d\binom djB_j$, with $B_0:=1$. Hence $N^d/(N)_d=\sum_{j=0}^d\binom dj B_j$ for every $d\ge 0$, and binomial inversion gives~\cref{eqn: Bd identity}.
 
\paragraph{The bound~\cref{eqn: Bdupper}.}
The right-hand side of~\cref{eqn: Bd identity} resembles the binomial expansion of $(1-1)^d=0$, and we exploit the resulting cancellations. Recall the elementary fact that
\begin{align}\label{eqn: binomial cancellation}
    \sum_{j=0}^d
    (-1)^{d-j}\binom{d}{j}p(j)
    =
    0
    \qquad
    \text{for every polynomial $p$ with } \deg(p)<d.
\end{align}
Write
\[
    \frac{N^j}{(N)_j}
    =
    \prod_{\ell=0}^{j-1}\frac{1}{1-\ell/N}
    =
    u_j(1/N),
    \qquad
    u_j(t)
    :=
    \prod_{\ell=0}^{j-1}\frac{1}{1-\ell t}
    =
    \sum_{k\ge0}p_k(j)t^k,
\]
where the last expression is the Taylor expansion of $u_j$ at $t=0$, which converges absolutely for $|t|<1/(j-1)$, and in particular at $t=1/N$ for all $j\le d\le N/16$.
Comparing the coefficients of $t^k$ in $u_{j+1}(t)=u_j(t)/(1-jt)=u_j(t)\sum_{\ell\ge0}j^\ell t^\ell$ gives
\[
    p_k(j+1)
    =
    \sum_{\ell=0}^k
    j^\ell p_{k-\ell}(j),
    \qquad\text{i.e.,}\qquad
    p_k(j+1)-p_k(j)
    =
    \sum_{\ell=1}^k
    j^\ell p_{k-\ell}(j),
\]
with $p_0(j)=1$ and $p_k(0)=0$ for $k\ge 1$ (since $u_0=1$).
 
We claim that $j\mapsto p_k(j)$ is a polynomial of degree at most $2k$. This is clear for $k=0$. If it holds for all $k'<k$, then each summand $j^\ell p_{k-\ell}(j)$ on the right-hand side above has degree at most $\ell+2(k-\ell)\le 2k-1$, so the forward difference $p_k(j+1)-p_k(j)$ is a polynomial of degree at most $2k-1$; since $p_k(j)=\sum_{i<j}(p_k(i+1)-p_k(i))$ and partial sums of a polynomial of degree at most $2k-1$ form a polynomial of degree at most $2k$, the claim follows.
Consequently, by~\cref{eqn: Bd identity} and~\cref{eqn: binomial cancellation}, all terms with $2k<d$, i.e., with $k<\lceil d/2\rceil$, cancel:
\begin{align*}
    B_d
    &=
    \sum_{j=0}^d
    (-1)^{d-j}\binom{d}{j}
    \sum_{k\ge0}p_k(j)N^{-k} \\
    &=
    \sum_{k\ge0}N^{-k}
    \sum_{j=0}^d
    (-1)^{d-j}\binom{d}{j}p_k(j)\\
    &=
    \sum_{k\ge\lceil d/2\rceil}N^{-k}
    \sum_{j=0}^d
    (-1)^{d-j}\binom{d}{j}p_k(j).
\end{align*}
 
It remains to bound $|p_k(j)|$ for $j\le d$ and $k\ge\lceil d/2\rceil$. For $j\ge 2$, expanding each factor of $u_j(t)=\prod_{\ell=1}^{j-1}(1-\ell t)^{-1}$ as a geometric series shows that
\[
    p_k(j)
    =
    \sum_{\substack{
        s_1,\ldots,s_{j-1}\ge0\\
        s_1+\cdots+s_{j-1}=k
    }}
    \prod_{\ell=1}^{j-1}\ell^{s_\ell}
\]
is a sum of $\binom{j+k-2}{k}\le 2^{j+k-1}$ nonnegative terms, each at most $j^k$; hence $0\le p_k(j)\le j^k2^{j+k-1}$, and this also holds for $j\in\{0,1\}$ since $p_k(0)=p_k(1)=0$ for $k\ge 1$. For $j\le d\le 2k$ we get
\[
    |p_k(j)|\le d^k 2^{j+k-1}\le d^k2^{3k}=(8d)^k.
\]
Therefore, using $\sum_j\binom dj=2^d$ and $8d/N\le 1/2$,
\begin{align*}
    B_d
    \le
    \sum_{k\ge\lceil d/2\rceil}N^{-k}
    \sum_{j=0}^d
    \binom{d}{j}|p_k(j)|
    \le
    2^d
    \sum_{k\ge\lceil d/2\rceil}
    \left(\frac{8d}{N}\right)^k
    \le
    2^{d+1}
    \left(\frac{8d}{N}\right)^{\lceil d/2\rceil},
\end{align*}
which is \cref{eqn: Bdupper}.
\end{proof}

\fi

\section{Reductions between Planted Distributions}\label{app: collision}
\subsection{Several planted collisions}\label{app: multi_advantage}
We recall the statement of the lemma for convenience.

\noindent\textbf{\cref{lem: multi_planted}.}\textit{
    $\Adv^{\mathsf{dist}}_q(\D_{\PC,k},\dist{\F}) =O(q^3/N^2)$ for $k=2,3.$}

\begin{proof}[Proof of \cref{lem: multi_planted}]
Fix $i\in\{1,2,3\}$ and a $q$-query algorithm $\Alg$. 

Let $g$ be an oracle from $\dist{\F}$ or $\DPC$. 
Given the oracle $g$, we will construct another oracle $h$ as follows.
Sample $2i-2$ distinct points $z_1, \ldots, z_{2i-2} \in [N]$ uniformly and sample independent uniform values $v_1, \ldots, v_{i-1} \in [N]$. Let $T = [N] \setminus \{z_1, ... , z_{2i-2}\}$. Define
    \[
    h(x) =
    \begin{cases}
    v_j & x \in \{z_{2j-1}, z_{2j}\}, \\
    g(x) & x \in T.
    \end{cases}
    \]
We define an algorithm $\Blg_i$ given the oracle $g$ for $i\in \{1,2,3\}$ as follows.
\begin{description}
    \item[$\Blg_i$:] It defines $h$ as above, and executes $\Alg$ with oracle $h$.
\end{description}
Note that one quantum query to $h$ can be simulated using at most two queries to $g$. In more detail, for every \(x,z\in\{0,1\}^n\), the simulation acts as follows:
\[
\begin{aligned}
\lvert x,0^n,z\rangle
&\xrightarrow{\text{query to } g}
  \lvert x,g(x),z\rangle\\
&\xrightarrow{\text{unitary}}
  \lvert x,g(x),z+h(x)\rangle\\
&\xrightarrow{\text{query to } g}
  \lvert x,0^n,z+h(x)\rangle.
\end{aligned}
\]

If $g\sim \D_{\F}$, then $h\sim\D_{\PC,i-1}$ because of newly planted collisions.
On the other hand, if $g\sim\DPC$, both inputs of the original planted collision pair of $g$ is contained in $T$ with
probability
\[
    p_i=\frac{(N-2i+2)_2}{(N)_2}\ge\frac15.
\]
On this event, 
$h\sim\D_{\PC,i}$ as there are $i$ pairs of the planted collisions. 
If any of the inputs of the original planted collision of $g$ is not contained in $T$, we can observe that $g|_T$ is a uniform random function (without planted collisions), as well as $h|_T$. Also note that $h|_{[N]\setminus T}$ consists of $i-1$ collisions, which shows that $h\sim \D_{\PC,i-1}$ in this case.

Consequently, the distributions of $h$ in the two cases, $g \sim \D_{\F}$ and $g\sim\DPC$, are exactly
\[
    \D_{\PC,i-1}
    \quad\text{and}\quad
    p_i\D_{\PC,i}+(1-p_i)\D_{\PC,i-1}.
\]
\minki{
For the educational purpose, let me highlight the difference between the original writing and my modifications. I write my opinion inside **}
\minki{
If $g\sim \D_{\F}$, then $h\sim\D_{\PC,i-1}$ *I added the reason why*.
If $g\sim\DPC$, \emph{its latent planted pair} *I am a bit confused about what it is before checking the other parts. latent is unclear, and there are planted collisions from $h$ as well.* is contained in $T$ with
probability
\[
    p_i=\frac{(N-2i+2)_2}{(N)_2}\ge\frac15.
\]
On this event, the surviving pair *What is surviving pair? In my opinion this is nonsense.* is uniform in $T$, so
$h\sim\D_{\PC,i}$ *I added the reason why*. \emph{Otherwise *Otherwise for what? Is that mean not on this event? not on $g\sim \DPC$?*} $g|_T$ is an ordinary uniform function:
\emph{discarding *What is discarding?*} either \emph{endpoint*What is endpoint?*} \emph{removes the only planted constraint *What is remove, what is planted constraint?*.} \emph{*In my opinion this sentence is unreadable except for someone who already know the proof.*}
Consequently, the two \emph{simulated distributions *What is simulated distributions? You even never mentioned what is simulation*} are exactly
\[
    \D_{\PC,i-1}
    \quad\text{and}\quad
    p_i\D_{\PC,i}+(1-p_i)\D_{\PC,i-1}.
\]}
\minki{up to here.}
Taking the difference of the two acceptance probabilities of $\Blg_i$ for the two cases gives
\[
    p_i\Adv^{\mathsf{dist}}_q
       (\D_{\PC,i},\D_{\PC,i-1};\Alg)
    =\Adv^{\mathsf{dist}}(\DPC,\D_{\F};\Blg_i)
    \le \Adv^{\mathsf{dist}}_{2q}(\DPC,\D_{\F}).
\]
We can use the hybrid argument (or the triangle inequality) to show, using $\D_{\PC,0}=\mathcal \dist_{\F}$, the following for $k\in \{2,3\}$:
\begin{align*}
    &\Adv^{\mathsf{dist}}_q(\D_{\PC,k},\D_{\F})
    \le\sum_{i=1}^k p_i^{-1}\Adv^{\mathsf{dist}}_{2q}(\DPC,\D_{\F}) \\
    &\qquad\le5k\Adv^{\mathsf{dist}}_{2q}(\DPC,\D_{\F})
    =O(q^3/N^2).\qedhere
\end{align*}
\end{proof}

\subsection{A planted triple from a planted pair}\label{app: triple_advantage}
We first prove the decisional version of the quantum search bound using the polynomial method~\cite{JACM:BBCMD01}.
\begin{claim}
Let $\mathcal Z_M$ be the distribution that always outputs the all-zero function $b_{\emptyset}(j) = 0$ for all $j \in [M]$.
Let $\mathcal S_M$ be a distribution obtained by sampling $z\gets [M]$ uniformly at random, and outputting the indicator function $b_z(j) = \truth{j=z}$.
Then, for every $q$, it holds that
\[
    \Adv^{\mathsf{dist}}_q(\mathcal S_M,\mathcal Z_M)
    \le \frac{16q^2}{M}
\]
\end{claim}

\begin{proof}
    Fix a $q$-query quantum algorithm $\Alg$ having access to an oracle $b:[M]\to \bit$. Write $x=(x_1,\cdots,x_M) \in \bit^M$ be an alternative representation of $b$ by choosing $x_i = b(i)$. Let $P(x)= \Pr[\Alg^x \to 1]$ be the acceptance probability of $\Alg$ given oracle $x$. It is clear that $0 \le P(x) \le 1$ for all $x\in \bit^M$.

    The polynomial method for quantum query algorithms \cite{JACM:BBCMD01} shows that $P$ can be written as a real multilinear polynomial of degree at most $2q$ on input $x$. 

    We use the standard symmetrization argument in the original polynomial method paper \cite[Section 3.1]{JACM:BBCMD01} by taking the average of all inputs. This shows that there exists a univariate real polynomial $p$ of degree at most $2q$ such that or every $k\in\{0,\cdots,M\}$, it holds that
    \[
    p(k)
    =
    \frac{1}{\binom Mk}
    \sum_{{x\in\bit^M, |x|=k}}
    P(x).
    \]
    In other words, $p(k)$ is the average acceptance probability of $\Alg$ over all Boolean functions of Hamming weight $k$. In particular, it holds that $p(0)$ and $p(1)$ correspond to the acceptance probability of $\Alg$ with oracle from $\mathcal Z_M$ and $\mathcal S_M$, respectively. 
    
    It remains to give an upper bound of $\delta:=|p(1)-p(0)|$. By the mean value theorem, there exists $z \in (0,1)$ such that $|p'(z)|=\delta$.
    Since the degree of $p$ is at most $2q$ and $0\le p(k)\le 1$ for every $k \in \{0,\cdots,M\}$, the results from \cite{ehlich1964schwankung,rivlin1966comparison} show that
    \[
    4q^2 \ge \deg(p)^2  \ge \frac{\delta M}{1+\delta}.
    \]
    If $16q^2/M \ge 1$, there is nothing to prove. Otherwise, rearranging above gives
    \[
    \delta \le \frac{4q^2}{M-4q^2} \le \frac{16q^2}{M}
    \]
    which proves the desired upper advantage bound.
\end{proof}

\noindent\textbf{\cref{lem: muthpc_smallq}.} $\Adv^{\mathsf{dist}}_q(\DthPC,\dist{\F}) =O(q^2/N)$.

\minki{\@ Junyoung, please revise the following proof.}
\junyoung{I done. Take a look.}

\begin{proof}[Proof of \cref{lem: muthpc_smallq}]
Fix a $q$-query algorithm $\Alg$. We first prove the intermediate bound
\[
    \Adv^{\mathsf{dist}}_q(\DthPC,\DPC)
    \le \frac{64q^2}{N-2}
\]
using the above lemma. Define an algorithm $\Blg$ with a Boolean oracle $b:[N] \to \{0, 1\}$ as follows.
\begin{enumerate}
\item Sample uniformly distinct inputs $a_1,a_2$, a uniform output $v$, and $R \gets \F$, all independently.
\item Define $h_b: [N] \to [N]$ by
    \[
        h_b(x):=
        \begin{cases}
            v    & \text{if } x\in\{a_1,a_2\} \text{ or } b(x)=1,\\
            R(x) & \text{otherwise}.
        \end{cases}
    \]
\item Run $\Alg$ with the oracle $O_{h_b}$.
\end{enumerate}
As one query to $O_{h_b}$ is simulated by two queries to $O_b$, $\Blg$ makes at most $2q$ queries to $b$.

Suppose $b \gets \mathcal{Z}_N$. As $b(x) = 0$ always, we have $h_b(a_1)=h_b(a_2)=v$ and $h_b(x)=R(x)$ for $x\notin\{a_1,a_2\}$. Since $v$ and the values $R(x)$ for $x\notin\{a_1,a_2\}$ are independent and uniform, $h_b$ is a uniform random function conditioned on $h_b(a_1)=h_b(a_2)$. Since $(a_1, a_2)$ is chosen uniformly, $h_b \sim \DPC$.

Suppose $b \sim \mathcal S_{N}$. Then we have $b(x)=\truth{x=z}$ for a randomly sampled $z \gets [N]$. Let $E$ be the event $z\in\{a_1,a_2\}$. Note that $\Pr[E]=2/N$. If $E$ occurs, since $b(x)=1$ implies $z \in \{a_1,a_2\}$, $h_b \sim \DPC$ by the above argument. If $E$ does not occur, we have $h_b(a_1)=h_b(a_2)=h_b(z)=v$ and $h_b(x)=R(x)$ for
$x\notin\{a_1,a_2,z\}$. Since $v$ and $R$ are independent and uniform, $h_b$ is a uniform random function conditioned on $h_b(a_1)=h_b(a_2)=h_b(z)$. Since $(a_1,a_2,z)$ is chosen uniformly, $h_b \sim \DthPC$.

Consequently, a $q$-query distinguisher for $\frac{N-2}{N}\DthPC + \frac{2}{N}\DPC$ and $\DPC$ gives a
$2q$-query distinguisher for $\mathcal S_{N}$ and $\mathcal Z_{N}$.
By the claim, we have
\begin{align*}
    \Adv^{\mathsf{dist}}_q(\DthPC,\DPC;\Alg)
    &= \frac{N}{N-2}\Adv^{\mathsf{dist}}_q\left(\frac{N-2}{N}\DthPC + \frac{2}{N}\DPC,\DPC;\Alg\right) \\
    &= \frac{N}{N-2}\Adv^{\mathsf{dist}}_{2q}(\mathcal S_{N},\mathcal Z_{N};\Blg) \leq \frac{64q^2}{N-2}.
\end{align*}
The triangle inequality now gives exactly the desired comparison:
\begin{align*}
    \Adv^{\mathsf{dist}}_q(\DthPC,\D_{\F})
    &\le
      \Adv^{\mathsf{dist}}_q(\DthPC,\DPC)
      +\Adv^{\mathsf{dist}}_q(\DPC,\D_{\F})\\
    &=O\left(\frac{q^2}{N}+\frac{q^3}{N^2}\right)=O\left(\frac{q^2}{N}\right),
\end{align*}
where we use $q \leq N$ in last equality.
\end{proof}

\subsection{Planted 4-XOR by splitting outputs}\label{app: quadxor_advantage}

\noindent\textbf{\cref{lem: mufox_smallq}.}
    $\Adv^{\mathsf{dist}}_q(\DfoX,\dist{\F}) =O(q^3/N^2)$.
\begin{proof}[Proof of \cref{lem: mufox_smallq}]
Fix $\Alg$ be a $q$-query algorithm.

Let $g:\bit^n \to \bit^n$ be either a uniform random function (from $\D_{\F}$) or a random function with a single planted collision (from $\DPC$). We define a function $h$ using $g$ that is either a uniform random function or a random function with a planted 4-xor tuple with a certain probability.

Formally, we define $h: \bit^n \to \bit^n$ as follows.
Let $M=N/2=2^{n-1}$. Sample a uniform random permutation $\pi$ of $\bit^n$ and independent and uniform random $R_1,\ldots,R_M\in\bit^n$.
For each $j\in \bit^{n-1}$, set
\[
    h(\pi(j\|0))=R_j,
    \qquad
    h(\pi(j\|1))=R_j+g(j\|0)
\]
where $\|$ denotes the concatenation.
Let $\Blg$ be an algorithm that defines $h$ as above, and executes $\Alg$ with oracle $h$. As noted before, each query to $h$ requires at most two queries to $g$.

We analyze two cases. If $g \sim \D_{\F}$, then it is not hard to see that $h \sim \D_{\F}$.
Consider $g\sim\DPC$ and let $\{u,v\}$ be the planted collision pair of $g$.
Define $E$ by the event that both of $u,v$ are included in $\{(j\|0): j\in \bit^{n-1}\}$.
This event occurs with probability at least
\[
    p=\frac{(M)_2}{(N)_2}
     =\frac{N-2}{4(N-1)}\ge\frac16.
\]
Conditioned on $E$, write $v=v_0\| 0$ and $u= u_0\|0$. 
The restriction of $g$ to $\{(j\|0): j\in \bit^{n-1}\}$ is uniform subject
to $g(u)=g(v)$. 
For $h$, this gives
\[
h(\pi(u_0\|0)) +h(\pi(u_0\|1)) + h(\pi(v_0\|0)) + h(\pi(v_0\|1)) = 0^n.
\]
The outputs of $h$ on those inputs $\pi(u_0\|0),\pi(u_0\|1),\pi(v_0\|0),\pi(v_0\|1)$ are uniform conditioned on the sum being $0^n$, meaning that they becomes the planted 4-xor tuple. All other outputs are uniformly and randomly distributed. In other words, $h \sim \DfoX$ in this case.
On the other hand, conditioned on $\lnot E$, we can see that the $h \sim \D_{\F}$.


Therefore, the distribution of the oracles given to $\Alg$ in this execution become $\D_{\F}$ and
$p\DfoX+(1-p)\D_{\F}$, respectively. Subtracting acceptance probabilities gives
\[
    p\Adv^{\mathsf{dist}}(\DfoX,\D_{\F};\Alg)
    =\Adv^{\mathsf{dist}}(\DPC,\D_{\F};\Blg)
    \le\Adv^{\mathsf{dist}}_{2q}(\DPC,\D_{\F}).
\]
Taking the maximum over $\Alg$ proves
\[
    \Adv^{\mathsf{dist}}_q(\DfoX,\D_{\F})
    \le6\Adv^{\mathsf{dist}}_{2q}(\DPC,\D_{\F})
    =O(q^3/N^2).\qedhere
\]
\end{proof}

\ifnum\fullpage=0
\section{Degree-6 component as planted three collisions plus small terms}\label{app: xopdeg6_smallq}

We use the following lemma.
\begin{lemma}\label{lem: binom_inequality}
    $\mathrm{For\ } x, y \in \mathbb{R}, |(x+y)^r-x^r|\leq r|y|(|x|+|y|)^{r-1}.$
\end{lemma}
\begin{proof}
Noting that $\binom{r}{k}\le k\binom{r}{k}=r\binom{r-1}{k-1}$ for $1\le k\le r$, we have
\[
|(x+y)^r-x^r| \le \sum_{k=1}^r  \binom{r}{k} |x^{r-k}y^k| \le r|y| \sum_{k=1}^r \binom{r-1}{k-1} |x|^{r-k}|y|^{k-1}
\]
which equals to $r|y|(|x|+|y|)^{r-1}.$
\end{proof}

\begin{lemma}
    $|\inner{\muXoPr^{=6},P_\Alg}| = O\left(\frac{q^3}{N^r}\right).$
\end{lemma}
\begin{proof}
We need to compute $\hmuP$ first. Let us define $Z_{2, 2, 2}, Z_{2, 4}, Z_{3, 3}, Z_6$ by
\begin{align*}
  Z_{2,2,2}
  &:=\sum_{M\ \mathrm{perfect\ matching\ of\ }[6]}\prod_{\{i,j\}\in M}\truth{z_i=z_j}, \\
  Z_{2,4}
  &:=\sum_{\substack{\{i, j\} \subset [6] \\ \{i, j\} \cap \{g, h, k, \ell\}=\varnothing}}\truth{z_i=z_j}\truth{z_g=z_h=z_k=z_\ell}, \\
  Z_{3,3}
  &:=\sum_{\substack{1 \in \{i, j, k\} \subset [6] \\ \{i, j, k\} \cap \{g, h, \ell\}=\varnothing}}\truth{z_i=z_j=z_k}\truth{z_g=z_h=z_\ell}, \\
  Z_6
  &:=\truth{z_1=z_2=z_3=z_4=z_5=z_6}. \\
\end{align*}
By the inclusion-exclusion principle for the fifteen events $E_{ij}=[z_i=z_j]$ gives,
\begin{align*}
    \truth{z_1, \cdots, z_6 \mathrm{\ distinct}} = 1 - \cdots - Z_{2, 2, 2} + 6Z_{2, 4} + 4Z_{3, 3} - 120Z_6,
\end{align*}
where the omitted terms correspond to partitions containing at least one
singleton block, which will be vanished.

Similarly, we define $C_{2,2,2}, C_{2,4}, C_{3,3}, C_6$ by
\begin{align*}
  C_{2,2,2}(\gamma)
  &:=\sum_{M\ \mathrm{perfect\ matching\ of\ }[6]}\prod_{\{i,j\}\in M}\truth{y_i=y_j}, \\
  C_{2,4}(\gamma)
  &:=\sum_{\substack{\{i, j\} \subset [6] \\ \{i, j\} \cap \{g, h, k, \ell\}=\varnothing}}\truth{y_i=y_j}\truth{y_g+y_h+y_k+y_\ell=0^n}, \\
  C_{3,3}(\gamma)
  &:=\sum_{\substack{1 \in \{i, j, k\} \subset [6] \\ \{i, j, k\} \cap \{g, h, \ell\}=\varnothing}}\truth{y_i+y_j+y_k=0^n}\truth{y_g+y_h+y_\ell=0^n}, \\
  C_6(\gamma)
  &:=\truth{y_1+y_2+y_3+y_4+y_5+y_6=0^n}. \\
\end{align*}
By applying \cref{eqn: characsum} and simplifying, we obtain
\begin{align*}
    \hmuP(\gamma) = \frac{-N^2C_{2,2,2}(\gamma) + 6NC_{2,4}(\gamma)+4NC_{3,3}(\gamma)-120C_6(\gamma)}{(N-1)_5}.
\end{align*}
set $R(\gamma) = 6NC_{2,4}(\gamma)+4NC_{3,3}(\gamma)-120C_6(\gamma)$. Note that $R(\gamma) = O(N)$. We rewrite
\begin{align*}
    \left(\hmuP(\gamma) (N-1)_5 \chi_\gamma\right)^r 
    &=\big(-N^2C_{2,2,2}(\gamma) + R(\gamma) \big)^r\chi_\gamma \\
    &=S(\gamma)\chi_\gamma + (-N^2)^rC_{2,2,2}(\gamma)\chi_\gamma,
\end{align*}
where $S(\gamma) = \big(-N^2C_{2,2,2}(\gamma) + R(\gamma) \big)^r - (-N^2)^rC_{2,2,2}(\gamma)$. For the second term, we have
\begin{align*}
    &\sum_{\supp(\gamma)=\{x_1<\ldots<x_6\}}C_{2,2,2}(\gamma)\chi_\gamma(f)\\
    &=\frac{15}{6!}\sum_{\supp(\gamma)=\{x_1,\ldots,x_6\}}\truth{y_1=y_2}\truth{y_3=y_4}\truth{y_5=y_6}(\gamma)\chi_\gamma(f)\\
    &=\frac{1}{48}\sum_{\substack{x_1,\ldots,x_6 \\ \text{distinct}}}\prod_{s=1}^3 (N\truth{f(x_{2s-1})=f(x_{2s})}-\constone) \\
    &=\frac{(N)_6}{48}(\mu_{\PC, 3} - 3\mu_{\PC, 2} + 3\muPC - \constone) \\
    &=\frac{(N)_6}{48}((\mu_{\PC, 3}-\constone) - 3(\mu_{\PC, 2}-\constone) + 3(\muPC-\constone)).
\end{align*}
This will be routinely bounded using \cref{lem: planted_collision,lem: multi_planted} below.

Now let us bound $S(\gamma)$. Partition the degree-six characters into
\begin{align*}
    \Gamma_0&:=\{\gamma:C_{2,2,2}(\gamma)=0\},\\
    \Gamma_1&:=\{\gamma:C_{2,2,2}(\gamma)=1\},\\
    \Gamma_{\geq2}&:=\{\gamma:C_{2,2,2}(\gamma)\geq2\}.
\end{align*}
For $J\in\{0,1,{\geq2}\}$, let $B_J:=\sum_{\gamma\in\Gamma_J}S(\gamma)\chi_\gamma$. Then we have
\begin{align*}
    \Big|\inner{\sum_{|\gamma|=6}S(\gamma)\chi_\gamma,P_\Alg}\Big| \leq |\inner{B_0,P_\Alg}| + |\inner{B_1,P_\Alg}| + |\inner{B_{\geq 2},P_\Alg}|
\end{align*}
and using Cauchy-Schwarz inequality
\begin{align*}
    |\inner{B_J,P_\Alg}| \leq \|B_J\|_2\|P_\Alg\|_2 \leq \|B_J\|_2 \le \left(\sum_{\gamma\in\Gamma_J}|S(\gamma)|^2\right)^{1/2}.
\end{align*}
If $\gamma\in\Gamma_0$, then $S(\gamma)=R(\gamma)^r=O(N^r)$.
\begin{align*}
    |\inner{B_0,P_\Alg}|
    &\leq \left( \sum_{\substack{|\gamma|=6 \\ C_{2,2,2}(\gamma)=0}} |S(\gamma)|^2 \right)^{1/2} \\
    &\leq \left( O(N^{12})O(N^{2r}) \right)^{1/2} = O(N^{r+6}).
\end{align*}
If $\gamma\in\Gamma_1$, we apply \cref{lem: binom_inequality}, which gives
\begin{align*}
    |S(\gamma)| &= \big|\big(-N^2C_{2,2,2}(\gamma) + R(\gamma) \big)^r - (-N^2)^rC_{2,2,2}(\gamma)\big| \\
    &= \big|\big(-N^2 + R(\gamma) \big)^r - (-N^2)^r\big| = O(N^{2r-1}).
\end{align*}
Since $C_{2,2,2}(\gamma)=1$, we have at most $O(N^3)$ choices of labels per support. Hence,
\begin{align*}
    |\inner{B_1,P_\Alg}|
    &\leq \left( \sum_{\substack{|\gamma|=6 \\ C_{2,2,2}(\gamma)=1}} |S(\gamma)|^2 \right)^{1/2} \\
    &\leq \left( O(N^{3+6})O(N^{(2r-1)\cdot 2}) \right)^{1/2} = O(N^{2r+3.5}).
\end{align*}
If $\gamma\in\Gamma_{\geq2}$, then
\begin{align*}
    |S(\gamma)| &= \big|\big(-N^2C_{2,2,2}(\gamma) + R(\gamma) \big)^r - (-N^2)^rC_{2,2,2}(\gamma)\big| \\
    &= \big|O(N^2)^r - (-N^2)^rO(1)\big| = O(N^{2r}).
\end{align*}
Since $C_{2,2,2}(\gamma)\ge2$, there exist indices $i, j$ and $g, h, k, l$ such that $y_i=y_j$ and $y_g=y_h=y_k=y_l$. Hence, we have at most $O(N^2)$ choices of labels per support. Hence,
\begin{align*}
    |\inner{B_{\geq 2},P_\Alg}|
    &\leq \left( \sum_{\substack{|\gamma|=6 \\ C_{2,2,2}(\gamma)\geq 2}} |S(\gamma)|^2 \right)^{1/2} \\
    &\leq \left( O(N^{2+6})O(N^{2r\cdot 2}) \right)^{1/2} = O(N^{2r+4})
\end{align*}
Finally,
\begin{align*}
    |\inner{\muXoPr^{=6},P_\mathcal{A}}|
    &\leq \frac{N^{2r}\Big|\inner{\sum_{|\gamma|=6}C_{2,2,2}(\gamma)\chi_\gamma,P_\mathcal{A}}\Big| + \sum_{J\in\{0,1,\ge 2\}}|\inner{B_J,P_\mathcal{A}}|}{(N-1)_5^r} \\
    &\leq \frac{O(q^3N^{2r+4}) + O(N^{r+6}) + O(N^{2r+3.5}) + O(N^{2r+4})}{\Theta(N^{5r})} \\
    &= O\Big(\frac{q^3}{N^{3r-4}}\Big) + O\Big(\frac{1}{N^{4r-6}}\Big) + O\Big(\frac{1}{N^{3r-3.5}}\Big) + O\Big(\frac{1}{N^{3r-4}}\Big)
    \\&= O\left(\frac{q^3}{N^r}\right).
\end{align*}
This concludes the proof.
\end{proof}
\fi
\section{Missing Proofs for the Compressed Oracles}\label{app: compressed}

\noindent\textbf{\cref{lem: pc recorded counts}.}\textit{
Let $\ket{\Xi_t}$ be the final state over the algorithm and database registers outputted by some algorithm after making $t\le N$ queries to the compressed oracle $O$. 
Then
\[
    \bra{\Xi_t}\mathsf D_{\mathrm{cnt}}\ket{\Xi_t}\le t,
    \qquad
    \bra{\Xi_t}\mathsf{Col}_{\mathrm{cnt}}\ket{\Xi_t}\le129t.
\]}

\begin{proof}
After $s$ queries, the state is supported on databases $D$ satisfying $
    |\{x:D(x)\ne\bot\}|\le s.$
Since the state has norm at most one, this proves the first inequality.

We now bound the increase in the collision count caused by a query.
Write
\[
    \chi_b(v):=(-1)^{b\cdot v},
    \qquad
    \ket{\chi_b}:=\frac1{\sqrt N}\sum_v\chi_b(v)\ket v
\]
for the Fourier basis of the response register.
In particular, $\ket{\chi_0}=\ket{+}$.
Fix a query input $x$ and a response state $\ket{\chi_b}$.
For $b\ne0$, the query acts on database register $x$ by the operator
\[
    O_b
    =
    Z_b+\ketbra{+}
    -\ket{+}\bra{\chi_b}
    -\ket{\chi_b}\bra{+}
    +\ket{\chi_b}\bra{\bot}
    +\ket{\bot}\bra{\chi_b},
    \qquad
    Z_b:=\sum_v\chi_b(v)\ketbra v,
\]
as follows from~\cite[Lemma~4.3]{EC:CFHL21}.
For $b=0$, the query is the identity.

Also fix the computational-basis values of the database registers other than $x$, and put
\[
    m_v:=|\{y\ne x:D(y)=v\}|,
    \qquad
    \ell:=\sum_v m_v,
    \qquad
    c:=\sum_v\binom{m_v}{2}.
\]
Thus, $\ell$ is the number of inputs $y\ne x$ with $D(y)\ne\bot$, and $c$ is the number of collision pairs not involving $x$.
In particular,
\[
    \sum_v m_v^2=\ell+2c.
\]
On this subspace, the collision-count operator is $cI+A$, where
\[
    A:=\sum_v m_v\ketbra v
\]
acts on database register $x$ and counts collision pairs involving $x$.
We have
\[
    \|A\ket{+}\|^2
    =
    \|A\ket{\chi_b}\|^2
    =
    \frac{\ell+2c}{N}.
\]
Moreover, $[A,Z_b]=0$ and $A\ket{\bot}=0$.
In the expansion of $O_b$, the commutator with each of
$\ketbra{+}$, $\ket{+}\bra{\chi_b}$, and
$\ket{\chi_b}\bra{+}$ has norm at most
$2\sqrt{(\ell+2c)/N}$.
The commutator with each of the remaining two terms has norm at most
$\sqrt{(\ell+2c)/N}$.
Therefore,
\[
    \|[A,O_b]\|
    \le8\sqrt{\frac{\ell+2c}{N}}.
\]
Since $O_b$ is unitary, the change in the collision-count operator satisfies
\[
    \|O_b^\dagger A O_b-A\|
    =
    \|O_b^\dagger[A,O_b]\|
    \le8\sqrt{\frac{\ell+2c}{N}}.
\]
The same bound holds for $b=0$, since the change is zero.

Let $\ket{\Xi_s}$ denote the state after $s$ queries and the subsequent algorithm operations, just before the next query, and put
\[
    k_s:=\bra{\Xi_s}\mathsf{Col}_{\mathrm{cnt}}\ket{\Xi_s}.
\]
Both $O$ and $\mathsf{Col}_{\mathrm{cnt}}$ preserve the subspaces specified by the query input $x$, the Fourier-basis response $b$, and the values of the other database registers.
Index these mutually orthogonal subspaces by $\lambda$.
Let $p_\lambda$ be the squared norm of the component of $\ket{\Xi_s}$ in subspace $\lambda$, and let $\ell_\lambda,c_\lambda$ be the corresponding values of $\ell,c$.
Then
\[
    \sum_\lambda p_\lambda=\|\Xi_s\|^2\le1,
    \qquad
    \ell_\lambda\le s,
    \qquad
    \sum_\lambda p_\lambda c_\lambda\le k_s.
\]
The last inequality holds because $c_\lambda$ counts only collision pairs not involving the queried input.
Applying the preceding operator bound in each subspace and then the Cauchy--Schwarz inequality gives
\begin{align*}
    &\bra{\Xi_s}O^\dagger\mathsf{Col}_{\mathrm{cnt}}O
        \ket{\Xi_s}-k_s\\
    &\qquad\le
        \frac8{\sqrt N}
        \sum_\lambda p_\lambda
            \sqrt{\ell_\lambda+2c_\lambda}\\
    &\qquad\le
        \frac8{\sqrt N}
        \left(\sum_\lambda p_\lambda\right)^{1/2}
        \left(
            \sum_\lambda p_\lambda
                (\ell_\lambda+2c_\lambda)
        \right)^{1/2}\\
    &\qquad\le
        8\sqrt{\frac{s+2k_s}{N}}.
\end{align*}

A unitary acting only on the algorithm's registers leaves the expectation of $\mathsf{Col}_{\mathrm{cnt}}$ unchanged.
An orthogonal projection on those registers cannot increase this expectation, provided that the resulting state is not renormalized.
Indeed, such a projection commutes with the positive operator $\mathsf{Col}_{\mathrm{cnt}}$.
Consequently,
\begin{align}\label{eqn: pc count recurrence}
    k_{s+1}
    \le k_s+8\sqrt{\frac{s+2k_s}{N}}.
\end{align}
Starting from $k_0=0$, we prove $k_s\le129s$ by induction.
If $k_s\le129s$ and $s<N$, then
\[
    k_{s+1}
    \le129s+8\sqrt{\frac{259s}{N}}
    \le129s+8\sqrt{259}
    \le129(s+1).
\]
This proves the second inequality.
\end{proof}

\paragraph{Proof of \cref{eqn: pc B bound}.}
We will prove that for
every joint algorithm-database state $\ket{\Xi}$ whose database registers are supported on $\operatorname{span}\{\ket{\widehat f}:f\in\F\}$, it holds that
\begin{align*}
    \|B^\dagger\ket{\Xi}\|^2
    \le\frac4N
        \bra{\Xi}\bigl(\mathsf D_{\mathrm{cnt}}
                    +\mathsf{Col}_{\mathrm{cnt}}\bigr)\ket{\Xi}.
\end{align*}
Put $\ket{z_v}:=\ket{\widehat v}-\ket{\bot}/\sqrt N$.
Then $\langle z_v|z_w\rangle=\truth{v=w}-1/N$ and
$\sum_v\ketbra{z_v}=Q$ on $\operatorname{span}\{\ket{\widehat v}:v\in\bit^n\}$.
For a fixed query input $x$, define
\[
    n_v:=\sum_{y\ne x}(\ketbra{z_v})_y,
    \qquad
    b_v:=\frac1{\sqrt N}\sum_{y\ne x}(\ket{\bot}\bra{z_v})_y.
\]
The identities for $\ket{z_v}$ give
\begin{align*}
    B_xB_x^\dagger
    &=\frac{\ketbra{\bot}_x}{N}
        \sum_v(n_v+b_v)^\dagger(n_v+b_v),\\
    \sum_v\|n_v\ket{\Xi}\|^2
    &=
    \left(1-\frac1N\right)
        \bra{\Xi}\mathsf D_{\mathrm{cnt},\ne x}\ket{\Xi}
    +2\bra{\Xi}\mathsf{Col}_{\mathrm{cnt},\ne x}\ket{\Xi},\\
    \sum_v\|b_v\ket{\Xi}\|^2
    &\le
    \frac{N-1}{N}
        \bra{\Xi}\mathsf D_{\mathrm{cnt},\ne x}\ket{\Xi}.
\end{align*}
Here the subscript $\ne x$ means that database register $x$ is omitted from the count, and the last line follows from the Cauchy--Schwarz inequality over the $N-1$ database registers.
Using $\|a+b\|^2\le2\|a\|^2+2\|b\|^2$, dropping the projector $\ketbra{\bot}_x$, and summing over the blocks corresponding to the query inputs proves \cref{eqn: pc B bound}.

\end{document}